\documentclass[12pt]{article}

\usepackage[utf8]{inputenc}               
\usepackage[T1]{fontenc}                  
\usepackage[font=small]{caption}          
\usepackage{amsmath,amsthm,amssymb}     
\usepackage{amsfonts}                     
\usepackage{mathtools}                    
\usepackage{bm, bbm}                      
\usepackage{url}                          
\usepackage{float}                        
\usepackage{setspace}                     
\usepackage{orcidlink}                    
\usepackage{placeins}                     
\usepackage{multirow}                     
\usepackage{array}

\usepackage[authoryear,round,comma]{natbib} 

\usepackage{hyperref}                     
\hypersetup{colorlinks=true,citecolor=black,linkcolor=black,urlcolor=black}

\usepackage{authblk}                      

\usepackage{graphicx}                     

\usepackage{geometry}                     
\usepackage{appendix}                     

\newcommand{\beginappendix}{              
    \renewcommand{\thesection}{\Alph{section}}
    \renewcommand{\thesubsection}{\Alph{section}.\arabic{subsection}}
    \setcounter{table}{0}
    \renewcommand{\thetable}{A.\arabic{table}}
    \setcounter{figure}{0}
    \renewcommand{\thefigure}{A.\arabic{figure}}
}

\usepackage{listings}                     
\definecolor{DarkGreen}{RGB}{1,50,32}
\newtheorem{lemma}{Lemma} 
\newtheorem{theorem}{Theorem}
\newtheorem{proposition}{Proposition}
\newtheorem{remark}{Remark}

\DeclarePairedDelimiter\abs{\lvert}{\rvert}
\DeclarePairedDelimiter\norm{\lVert}{\rVert}

\newcommand{\SUM}[2]{\sum_{#1}^{#2}}
\newcommand{\DERIV}[2]{\nabla_{#2} {\, #1}}
\newcommand{\DDERIV}[2]{\nabla_{#2}^2 {\, #1}}

\def\hat{\widehat}
\def\tilde{\widetilde}
\def\pr{\mathrm{pr}}
\def\tr{\mathrm{tr}}
\def\det{\mathrm{det}}
\def\E{\mathbb{E}}
\def\Cov{\mathrm{Cov}}

\def\sp{\mathbb{P}}
\def\bzero{\bm{0}}
\def\hpi{\hat\pi}
\def\bb{\bm{b}}

\def\bl{\bm\ell}

\def\hbeta{\hat\beta}
\def\bbeta{\bm\beta}
\def\hbbeta{\hat{\bm\beta}}
\def\tbbeta{\tilde{\bm\beta}}
\def\bpsi{\bm\psi}

\def\hbpsi{\hat{\bpsi}}
\def\tbpsi{\tilde{\bpsi}}

\def\hwm{\hat{m}}
\def\twm{\tilde{m}}
\def\bmu{\bm\mu}
\def\hmu{\hat\mu}
\def\tmu{\tilde\mu}
\def\hbmu{\hat{\bm\mu}}
\def\tbmu{\tilde{\bm\mu}}
\def\hh{\hat{h}}
\def\th{\tilde{h}}
\def\barh{\bar{h}}
\def\hdelta{\hat\delta}
\def\tdelta{\tilde\delta}

\def\hU{\hat{U}}
\def\tU{\tilde{U}}

\def\hB{\hat{B}}
\def\eB{\bar{B}}
\def\tB{\tilde{B}}

\begin{document}

\title{\bf Bias reduction in g-computation for covariate adjustment in randomized clinical trials}

\author[1]{Xin Zhang\orcidlink{0000-0001-9894-2634}\thanks{Corresponding author. Email: \href{xin.zhang6@pfizer.com}{xin.zhang6@pfizer.com}.}}
\author[2]{Lin Liu\orcidlink{0000-0002-9883-7962}\thanks{Email: \href{linliu@sjtu.edu.cn}{linliu@sjtu.edu.cn}}}
\author[1,3]{Haitao Chu\orcidlink{0000-0003-0932-598X}\thanks{Email: \href{chux0051@umn.edu}{chux0051@umn.edu}}}

\affil[1]{Data Sciences and Analytics, Pfizer Inc.}
\affil[2]{Institute of Natural Sciences, MOE--LSC, School of Mathematical Sciences, CMA--Shanghai, and SJTU--Yale Joint Center for Biostatistics and Data Science, Shanghai Jiao Tong University}
\affil[3]{Division of Biostatistics and Health Data Science, University of Minnesota Twin Cities}

\date{}

\maketitle

\begin{abstract}
G-computation is a powerful method for estimating unconditional treatment effects with covariate adjustment in randomized clinical trials. It typically relies on fitting canonical generalized linear models. However, this could be problematic when the sample size or event number is small relative to the number of covariates. Common issues include the underestimation of the variance and the potential nonexistence of maximum likelihood estimators. Bias reduction methods are commonly employed to address these issues, including Firth correction, which guarantees the existence of corresponding estimates. Yet, their application within g-computation remains underexplored. In this article, we analyze the asymptotic bias of g-computation estimators and propose a novel bias-reduction method that improves both estimation and inference. Our approach performs bias correction via generalized Oaxaca-Blinder estimators, and thus the resulting estimators are guaranteed to be bounded. The proposed debiased estimators use slightly modified versions of maximum likelihood or Firth correction estimators for nuisance parameters. We also introduce a simple small-sample bias adjustment for variance estimation to improve finite-sample inference validity. Through extensive simulations, we demonstrate that our proposed method offers superior finite-sample performance, effectively addressing the bias-efficiency tradeoff. Finally, we illustrate its practical utility by reanalyzing a completed randomized clinical trial.  In this example, our method improves precision in a small subgroup analysis for which the standard method fails to fit the regression model. 
\end{abstract}


\newpage

\doublespacing

\section{Introduction} \label{sec:intro}

In the analysis of randomized controlled trials (RCTs), it is often desirable to adjust for baseline covariates by fitting a regression model \citep{hernandez2004covariate,hernandez2006randomized,lee2022benefits}. When covariates are strong prognostic factors of the outcome, covariate adjustment generally improves statistical efficiency, as repeatedly shown in simulation studies \citep[e.g.,][]{kahan2014risks}. Regulatory guidance has also detailed practical considerations for regression-based adjustment \citep{ema2015, fda2023}.

Among regression-based adjustment methods, g-computation provides a principled approach to estimating unconditional treatment effects in RCTs \citep{robins1986new,ge2011covariate}. G-computation has attracted growing attention among trial statisticians \citep{benkeser2021improving,van2024covariate} and is endorsed in recent regulatory guidance \citep{fda2023}. Closely related ideas have long been used in epidemiology (regression standardization), survey sampling (generalized regression) and the missing-data literature (``regression imputation''). In practice, parametric models, most commonly canonical generalized linear models (GLMs), are used as working outcome models \citep{benkeser2021improving, van2024covariate} and are usually fit by maximum likelihood estimation (MLE), so the performance of the g-computation estimator depends on the goodness of fit of the GLM working model.

However, MLE of GLMs can be unstable or biased when the sample size or number of events is small relative to the number of covariates. Such scenarios are not uncommon in RCTs. For example, early-phase studies often have limited sample sizes, or placebo response rates may be close to zero. In these settings, g-computation with MLE tends to underestimate variance, inflating Type I error and reducing interval coverage, thereby undermining the validity of statistical inference. These issues have been documented empirically \citep[e.g.,][]{tackney2023comparison}. In particular, for binary outcomes, small event numbers can lead to data separation in logistic regression \citep{albert1984existence}. For instance, this can happen when outcomes within a particular stratum are all 1s or all 0s. Under data separation, MLE does not exist, and the model fit will be highly unstable in finite samples. Similar phenomena arise for count outcomes fit by Poisson log-linear models \citep{joshi2022solutions}. These issues in outcome model fit propagate to g-computation, yielding unreliable treatment effect estimates.

In the literature, several bias reduction methods have been proposed for parametric models (including GLMs) to improve either point \citep{cox1968general,cordeiro1991bias} or variance estimation \citep{mancl2001covariance,fay2004small}. These methods address small-sample bias by removing all or part of the first-order ($n^{-1}$) bias inherent in the MLE. When MLEs do not exists, a particular method called Firth correction (FC) provides a principled approach to resolve this issue \citep{kosmidis2020jeffreys}. FC has been applied to various GLMs \citep[e.g.,][]{heinze2002solution,joshi2022solutions}. While it was initially developed to correct for the first-order bias of MLEs \citep{firth1993bias}, it also solves the data separation issue and achieves better finite-sample performance in the presence of sparse data \citep{sur2019modern, kosmidis2020jeffreys, joshi2022solutions}.

It is thus appealing to use FC, instead of the MLE, to estimate nuisance parameters in working models for g-computation.  The consistency of the resulting estimator and the validity of statistical inference can be justified within the framework of $M$-estimation (Section~\ref{sec:gc-fc}). As can be seen in a simulation experiment (Figure~\ref{fig:gc_fc}), the FC-based g-computation estimator has better interval coverage probabilities. Unfortunately, such estimators would introduce biases leading to non-negligible underestimation of treatment effects, and thus harm the power of hypothesis testing.

To address the two key challenges in g-computation discussed so far--namely, the underestimation of variability for statistical inference due to overfitting (especially with MLE), and the non-negligible bias introduced by FC--it is indispensable to gain a deeper understanding of the asymptotic bias in g-computation. In this article, we characterize the first-order biases of these estimators and propose a novel bias-reduction approach that improves both estimation accuracy and inference validity. The proposed approach is developed under simple randomization and does not rely on correctly specified working models. The resulting treatment effect estimators (e.g., risk differences [RD], risk ratios [RR]) are the plugin estimators with the debiased estimators of treatment-specific means, which are guaranteed to be bounded. Those debiased estimators take the form of generalized Oaxaca-Blinder (gOB) estimators \citep{guo2023generalized}, leveraging nuisance parameters obtained from slightly modified versions of MLE or FC. The corresponding variance estimators are adjusted for small-sample bias using leverage scores of the fitted working models.

The rest of this article is organized as follows. Section~\ref{sec:pre} introduces g-computation for covariate adjustment in RCTs and bias-reduction methods for GLMs. In Section~\ref{sec:gc-fc}, we evaluate the use of FC to estimate nuisance parameters. The asymptotic biases of g-computation estimators (with MLE or FC) are characterized in Section~\ref{sec:theory}. The proposed debiased estimators (with MLE or FC) for both point and variance estimation are provided in Section~\ref{sec:method}, and their finite-sample performance is evaluated by simulation experiments. In Section~\ref{sec:app}, we apply the proposed method to a completed phase 3 randomized trial. Section~\ref{sec:disc} concludes our article with a discussion of limitations. All technical derivations can be found in Appendix~\ref{append:proof}, and the R code for replication purposes is included in Appendix~\ref{append:code}.

\section{Preliminaries} \label{sec:pre}

We begin by introducing the notations and assumptions. Here we denote the baseline covariates used for adjustment by $W_i$, and the treatment arm indicator by $A_i$, which takes values from $1$ (the reference arm) to $k$ (for $k-1$ tested arms). The true randomization probability for each arm is $\pi_a \in (0, 1)$ and $\SUM{a=1}{k} \pi_a = 1$. Let $Y_i$ be the observed outcome under the assigned treatment arm $A_i$ and $Y_i(a)$ denote the potential outcome under the treatment arm $a$. Throughout this article, $( W_i, Y_i(1), \ldots, Y_i(k) )$ for $i = 1, \ldots, n$ are independent and identically distributed (i.i.d.) draws from a superpopulation distribution $\mathbb{P}$. We define the observed data vector as $D_i \coloneqq ( Y_i, A_i, W_i^\top )^\top$ and assume simple randomization for treatment assignment. This implies that $D_1, \dots, D_n$ are also i.i.d. samples. We let $r_{ i | a }$ denote the true outcome model of subject $i$ under treatment assignment $A_i = a$, i.e. $r_{ i | a } \coloneqq \E [ Y_{i} (a) | W_{i} ] \equiv \E [ Y_{i} | W_{i}, A_{i} = a ]$. For theoretical results in this article, we follow the assumptions used in \citet{kosmidis2024empirical}. We will also assume that $Y_i$ has bounded second moment, and $W_i$ has bounded fourth moment. We use the conventional notation for the stochastic order ($O_{\sp}, o_{\sp}$) and non-stochastic order ($O, o$), equally applicable to scalars, vectors and matrices.  
Moreover, we invoke three standard causal identification assumptions: exchangeability (no unmeasured confounding), $Y_i(a) \perp A_i \mid W_i$; positivity (overlap), $0 < P(A_i = a \mid W_i = w) < 1$; and consistency (well-defined interventions), whereby $Y_i = Y_i(a)$ whenever $A_i = a$. In a properly conducted RCT, these assumptions hold by design.

\subsection{G-computation in RCTs} \label{sec:pre-gc}

G-computation typically employs GLMs as working models to adjust for baseline covariates in RCTs. Such working models take the form $m ( X_i^\top \bbeta )$, where $\bbeta$ is a $p$-dimensional nuisance parameter and $X_i$ (of $p$ dimensions) usually comprises an intercept, the treatment assignment indicators for $A_i$ and the baseline covariates $W_i$; $X_i$ may also include interactions between $A_i$ and $W_i$. Let $g(\cdot)$ be the link function for GLMs and $m \coloneqq g^{-1}$ be its inverse. Throughout this article, we only consider GLMs with canonical link functions and treat $p$ as fixed. Moreover, we assume that $m(\cdot)$ is non-decreasing, continuously twice-differentiable (with its first and second derivatives denoted as $m'$ and $m''$, respectively). In most of GLMs encountered in practice, $m(\cdot)$ is non-decreasing.

Let $\hbbeta$ be the MLE of $\bbeta$ obtained by solving $\SUM{i=1}{n} U_i ( \bbeta ) \equiv \bzero_p$, where $U_i ( \bbeta ) \coloneqq \{ Y_i - m ( X_i^\top \bbeta ) \} X_i$ is the usual score function of canonical GLMs. Write $\hwm_i \coloneqq m ( X_i^\top \hbbeta )$ and $\hwm_{i|a} \coloneqq m ( X_{i|a}^\top \hbbeta )$, where $X_{i|a}$ is defined as $X_i$ but setting $A_i = a$. The g-computation estimator of $\mu_a \coloneqq \E [ Y ( a ) ]$ is $\hmu_a = n^{-1} \SUM{i=1}{n} \hwm_{i|a}$, and the one for the unconditional treatment effect $\delta ( \mu_a, \mu_b )$ is $\hdelta \coloneqq \delta ( \hmu_a, \hmu_b )$. For example, $\hdelta$ can be $\hmu_{2} - \hmu_{1}$, $\hmu_{2} / \hmu_{1}$, etc. Notably, using canonical GLMs as working models ensures the consistency of g-computation even under model misspecification \citep{freedman2008randomization}. This is because that under randomization $\hwm_i$ satisfies the so-called ``prediction unbiasedness'' \citep{guo2023generalized}, i.e.
\begin{equation}
\SUM{i \colon A_i = a}{} Y_i \equiv \SUM{i \colon A_i = a}{} \hwm_i. \label{eq:predict-mle}
\end{equation}

Statistical inference for g-computation under misspecified working models can be conducted using Wald or score-based methods, the finite-sample performance of which were recently studied in \citet{zhang2025robust}. Both methods rely on estimating the variance(-covariance) of $\hbmu \coloneq ( \hmu_1, \ldots, \hmu_k)^\top$. This can be obtained using either the influence function (IF) of $M$-estimators \citep{stefanski2002calculus} or the efficient IF of the AIPW estimator \citep{tsiatis2008covariate}, leading to the so-called IF-based variance estimator \citep{boos2013essential}. Following \citet{yuan2012variable}, we refer to the former as the \textit{empirical} IF (directly obtained from the empirical sandwich variance matrix) and the latter as the \textit{theoretical} IF (does not involve the sandwich matrix of $\hbbeta$). Let $\bbeta_0$ be the solution to the population score equation $\E [ U_i ( \bbeta ) ] \equiv 0$,  $m_i \coloneqq m ( X_i^\top \bbeta_0 )$, $m_{i|a} \coloneqq m ( X_{i|a}^\top \bbeta_0 )$, and $m_{i|a}' \coloneqq m' ( X_{i|a}^\top \bbeta_0 )$. The empirical/theoretical IFs of $\hmu_a$ are then 
\begin{align}
\psi_i^a & = \E [ m_{1|a}' \cdot X_{1|a}^\top ] \bpsi_i^{\bbeta} + m_{i|a} - \mu_a \label{eq:if-mu-mest} \\
         & \equiv \frac{ I ( A_i = a ) }{\pi_a} ( Y_i - m_i ) + m_{i|a} - \mu_a, \label{eq:if-mu-aipw}
\end{align}
respectively, where $\bpsi_i^{\bbeta} \coloneqq B^{-1} U_i$ is the IF of $\hbbeta$, with $B \coloneqq \E [ - \DERIV{ U_i (\bbeta_0) }{} ]$ (the subscript $\bbeta$ is omitted in the gradient operator), $U_i \coloneqq U_i ( \bbeta_0 ) = (Y_i - m_i) X_i$, and $I(\cdot)$ be the indicator function. \citet[][Proposition 1]{zhang2025robust} have established that this equivalence holds when parametric working models are misspecified under simple or stratified randomization, since
\begin{equation}
\E [ m_{i|a}' \cdot X_{i|a}^\top ] B^{-1} X_i \equiv \frac{ I ( A_i = a ) }{\pi_a}. \label{eq:zhang}
\end{equation}
For the sake of completeness, we provide a proof of \eqref{eq:zhang} in Appendix~\ref{append:proof-zhang}.

One commonly used strategy to construct g-computation estimators is to posit a single working model, as described at the beginning of this section, fitted with data pooled from all $k$ treatment arms. An alternative widely-used strategy, with the potential to offer more efficiency gains, is to posit a separate working model for each $\mu_a$, which only uses data from the corresponding arm \citep{tsiatis2008covariate}. We call the former approach the pooled working model, and the latter the stratified working model. \citet[][Supporting Information Section A]{zhang2025robust} show that stratified and pooled working models can be represented within a unified framework using a single working model. Consequently, all preceding results apply to both types of working models (see Appendix~\ref{append:wm-strat-mle} for details). This unified formulation also facilitates the development of the theoretical results and the proposed estimators in Sections~\ref{sec:theory} and~\ref{sec:method}.

\subsection{Bias correction in canonical GLMs} \label{sec:pre-bias}

Provided that the model is correctly specified and the i.i.d. assumption holds, the first-order ($n^{-1}$) bias of MLEs has been studied for general parametric models \citep{cox1968general} and for GLMs \citep{cordeiro1991bias}. Their closed-form formulae are available, which provide explicit debiased estimators of $\hbbeta$ but require the existence of MLEs. \citet{firth1993bias} has proposed an alternative approach to obtain debiased estimators by augmenting the score equation with a penalty term. Let $\tbbeta$ be the estimator of $\bbeta$ using Firth correction (FC). For canonical GLMs, the modified score equation reads as 
$
\SUM{i=1}{n} \tU_i + \DERIV{ \log \det ( \tB ) }{} / 2 \equiv \bzero_p,
$
where $\tU_i$ and $\tB$ denote the empirical versions of $U_i$ and $B$, respectively. It guarantees the boundedness of $\tbbeta$ for certain GLMs \citep{kosmidis2020jeffreys}, including those with canonical links, provided that the design matrix is of full rank. 

Under model misspecification, neither the estimator explicitly debiasing MLEs nor the implicit one using FC ($\tbbeta$) are free of the first-order bias. As both $\hbbeta$ and $\tbbeta$ are $M$-estimators, their asymptotic bias formulae can be directly obtained using the formula provided in \citet{kosmidis2024empirical} and read as
\begin{align}
\E ( \hbbeta - \bbeta_0 ) & = \underbrace{ \frac{1}{2} \, B^{-1} \DERIV{ \tr ( B^{-1} M ) }{} }_{ \eqcolon \ \bb_1 ( \hbbeta ) } / n + O ( n^{-3/2} ), \label{bias-MLE} \\ 
\E ( \tbbeta - \bbeta_0 ) & = \underbrace{ \frac{1}{2} \, B^{-1} \left\{ \DERIV{ \tr ( B^{-1} M ) }{} + \DERIV{ \log \det ( B ) }{} \right\} }_{ \eqcolon \ \bb_1 ( \tbbeta ) } / n + O ( n^{-3/2} ), \label{bias-Firth}
\end{align}
where $M \coloneq \E ( U_i^{\otimes 2} )$ is the meat matrix. They also have proposed both explicit and implicit bias reduction methods without requiring correct model specification. The utility of their method in g-computation remains underexplored. Besides, their implicit method does not guarantee the existence of corresponding estimates, which is one of the major hurdles to implement g-computation in practice. 

\citet{kosmidis2024empirical} obtain \eqref{bias-MLE} and \eqref{bias-Firth} using the tensor-based method described in \citet{mccullagh2018tensor}. For the sake of completeness, we provide a proof of the corresponding bias formulae using only linear algebra in Appendix~\ref{append:proof-bias-beta}.

\section{G-computation with FC}  \label{sec:gc-fc}

As indicated in Section~\ref{sec:intro}, small sample sizes or rare events (relative to the number of covariates) pose significant challenges to the application of g-computation, as 1) canonical GLM working models might fit the data poorly and 2) MLEs may not even exist, leading to unreasonably large estimates (in absolute values). It is well known that FC could improve the finite-sample performance of GLM model fit when MLEs do not exist, thus resolving the latter issue \citep{kosmidis2020jeffreys}. Several simulation studies also report that compared to MLE (when they exist), FC has superior finite sample performance when the sample size is small relative to the number of covariates \citep{sur2019modern,kosmidis2020jeffreys,joshi2022solutions}. Therefore, using FC to estimate nuisance parameters is a natural candidate solution to resolve both issues. In this section, we illustrate how FC can be seamlessly integrated into g-computation. 

Let 
$
h_{ii} ( \bbeta ) \coloneqq m' ( X_i^\top \bbeta ) \cdot X_i^\top \{ \SUM{j=1}{n} m' ( X_j^\top \bbeta ) \cdot X_j X_j^\top \}^{-1} X_i
$
be the $i$th diagonal element of the hat matrix (i.e., the leverage score) from $U_i ( \bbeta )$. The Firth's modified score equation is 
\begin{equation}
\label{eq:score-fc}
\begin{gathered}
\SUM{i=1}{n} U_i ( \bbeta ) + \Delta^{(n)} ( \bbeta ) = \mathbf{0}_{p}, \text{ where } \Delta^{(n)} ( \bbeta ) \coloneqq \frac{1}{2} \, \SUM{i=1}{n} h_{ii} ( \bbeta ) \cdot \frac{ m'' ( X_i^\top \bbeta ) }{ m' ( X_i^\top \bbeta ) } \cdot X_{i}.
\end{gathered}
\end{equation}
Note that for binary regression with logit link, $m''( X_i^\top \bbeta ) / m' ( X_i^\top \bbeta ) = 1 - 2 m ( X_i^\top \bbeta )$, whereas for Poisson regression with log link, $m'' ( X_i^\top \bbeta ) / m' ( X_i^\top \bbeta ) = 1$. The FC estimator $\tbbeta$ is then obtained by solving $\SUM{i=1}{n} U_i ( \tbbeta ) + \Delta^{(n)} ( \tbbeta ) = \bzero_p$. Since $\Delta^{(n)} ( \bbeta )$ is $O_\sp ( 1 )$, $\sqrt{n} ( \tbbeta - \bbeta_0 )$ has the same asymptotic normal distribution as $\sqrt{n} ( \hbbeta - \bbeta_0 )$, i.e., with the same IF, $\bpsi_i^{\bbeta}$, and the same asymptotic variance, $\E ( \bpsi_i^{\bbeta \, \otimes 2} )$ \citep{firth1993bias,kosmidis2024empirical}.  

Write $\twm_i \coloneqq m ( X_i^\top \tbbeta )$ and $\twm_{i|a} \coloneqq m ( X_{i|a}^\top \tbbeta )$, the g-computation estimator of $\mu_a$ with FC is then $\tmu_a = n^{-1} \SUM{i=1}{n} \twm_{i|a}$ and analogously $\tdelta_{a,b} \coloneqq \delta ( \tmu_a, \tmu_b )$ for estimating treatment effects. In contrast to $\hmu_a$, the prediction unbiasedness \eqref{eq:predict-mle} no longer holds for $\tmu_a$ due to $\Delta^{(n)} ( \bbeta )$, as shown in the following lemma. 
\begin{lemma} \label{lem:predict-fc}
Let $\th_{ii} \coloneqq h_{ii} (\tbbeta)$. Then, 
\[
\SUM{ i \colon A_i = a }{} \twm_i \equiv  \SUM{ i \colon A_i = a }{} Y_i + \frac{1}{2} \SUM{ i \colon A_i = a }{} \th_{ii} \cdot \frac{\twm_i''}{\twm_i'} = \SUM{ i \colon A_i = a }{} Y_i + O_{\sp} ( 1 ).
\]
\end{lemma}
\noindent The proof of Lemma~\ref{lem:predict-fc} can be found in Appendix~\ref{append:proof-predict-fc}. The kind of bias exhibited in Lemma~\ref{lem:predict-fc} has been previously discussed in \citet{puhr2017firth}, but its impact on g-computation has not been explored.  

Similar to $\hbmu$, $\tbmu = ( \tmu_1, \ldots, \tmu_k )^\top$ is a partial $M$-estimator. Fortunately, since $\Delta^{(n)} ( \bbeta ) = O_\sp ( 1 )$, the estimating equation for $( \tbmu^\top, \tbbeta^\top )^\top$ only differs from that for $( \hbmu^\top, \hbbeta^\top )^\top$ with a term of $o_{\sp} ( n^{1/2} )$. This suggests that (as $p$ is fixed) $\tbmu$ is still $\sqrt{n}$-consistent and its asymptotic normal distribution is the same as the one for $\hbmu$ \citep[][p. 30]{stefanski2002calculus}. Consequently, their variance estimators, test statistics and confidence intervals (CIs) can be calculated in the same manner, using the corresponding empirical estimates of $\bpsi_i^a$ based on \eqref{eq:if-mu-mest} or \eqref{eq:if-mu-aipw}. Those estimators are also applicable for stratified working models with FC as shown in Appendix~\ref{append:wm-strat-fc}.

Since $\tmu_a$ is consistent, the extra bias shown in Lemma~\ref{lem:predict-fc} is negligible for large sample sizes. To illustrate the impact of this bias on g-computation estimators when sample sizes are small, we first present some numerical results from a simulation experiment.

\begin{figure}[t!]
\centering
\includegraphics[width=\linewidth]{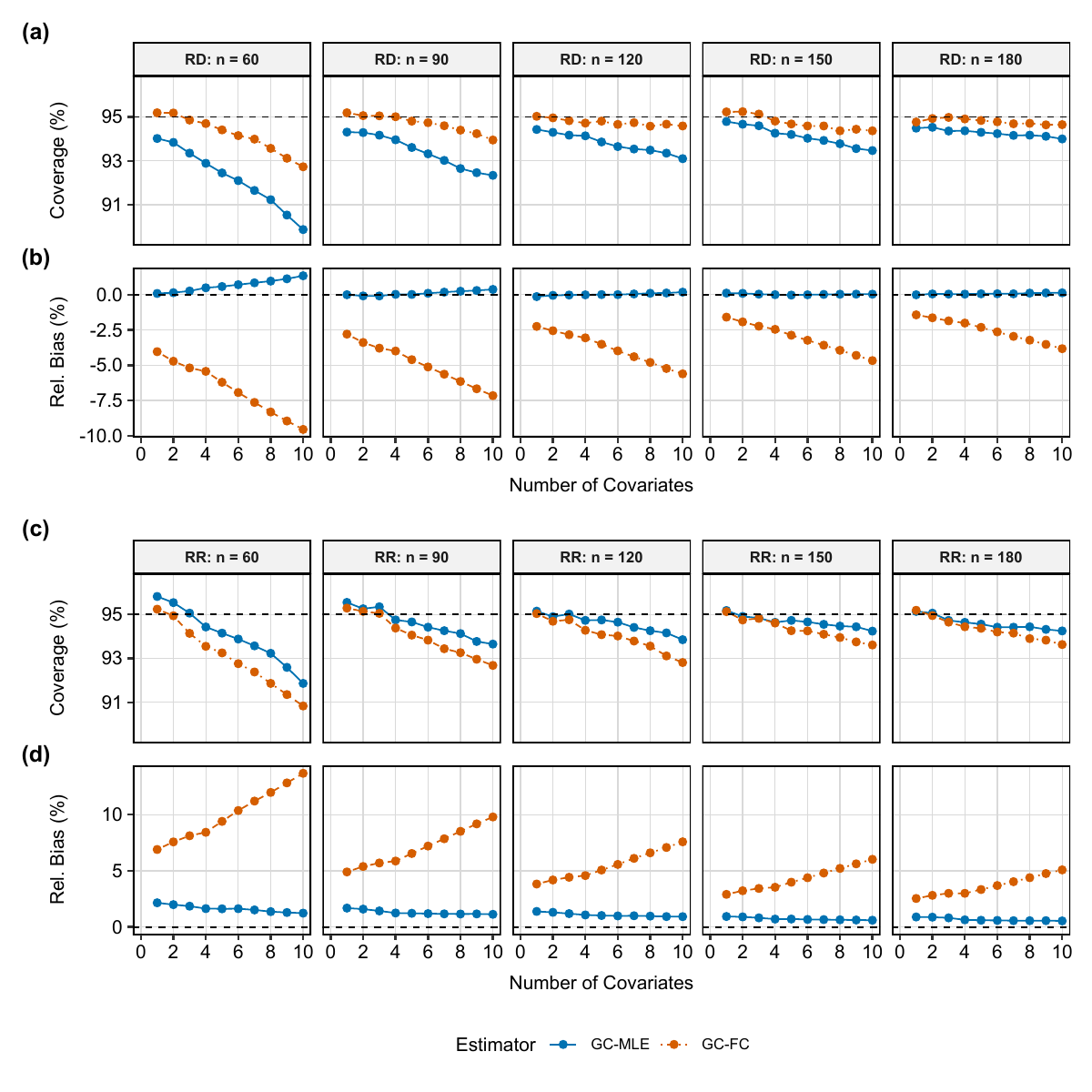}
\caption{
Coverage of the 95\% CI and bias of treatment effect estimates for \textbf{GC-MLE} ($\hmu_a$) and \textbf{GC-FC} ($\tmu_a$). RD: $\mu_2-\mu_1$; RR: $\mu_1/\mu_2$; $n$: sample size; Rel.: Relative. 
}
\label{fig:gc_fc}
\end{figure}

\paragraph{Simulation Experiment I(a)} 

Figure~\ref{fig:gc_fc} presents results from a simulation experiment to evaluate the finite-sample performance of \textbf{GC-MLE} ($\hmu_a$) and \textbf{GC-FC} ($\tmu_a$). We consider a hypothetical trial with two arms (1:1 allocation) with sample sizes ranging from $n=60$ to $180$, and outcome probabilities $( \mu_1, \mu_2 ) = ( 0.25, 0.60 )$. Two estimands are evaluated: RD ($\mu_2-\mu_1$) and RR ($\mu_1/\mu_2$). Details of the data generating process can be found in Appendix~\ref{append:sim-dgp}. 

A single logistic regression working model (including a treatment indicator variable), fitted using data pooling from both arms (i.e., a pooled working model), is used for both estimators. The number of baseline covariates for adjustment is up to $10$. The estimation procedures for the nuisance parameters are described in Appendix~\ref{append:sim-est}. Wald-type confidence intervals are constructed.  

For RD ($\mu_2-\mu_1$), \textbf{GC-FC} improves interval coverage relative to \textbf{GC-MLE} (Figure~\ref{fig:gc_fc}a), with the most pronounced improvement in small samples. However, it incurs non-negligible bias in point estimation (Lemma~\ref{lem:predict-fc}), particularly when sample sizes are small and/or the number of adjusted covariates is large, as shown in Figure~\ref{fig:gc_fc}b. This negative bias leads to underestimation of treatment effects and can reduce the power of tests of the null hypothesis.

For RR ($\mu_1/\mu_2$), \textbf{GC-FC} exhibits poorer interval coverage than \textbf{GC-MLE} in small samples and when the number of adjusted covariates is large. Moreover, the incurred bias is positive, implying a smaller estimated risk reduction and thus an underestimation of the treatment effect, which can reduce the power of tests of the null hypothesis. 

The above empirical results motivate us to perform bias reduction on \textbf{GC-FC} or \textbf{GC-MLE}, which we describe in detail in the next two sections.

\section{Bias Characterization in G-computation} \label{sec:theory}

In this section, we characterize the first-order bias of g-computation estimators with MLE ($\hmu_a$) and FC ($\tmu_a$), respectively, under misspecified working models. To begin with, we present the first-order bias formulae for $\hbbeta$ and $\tbbeta$, respectively. 
\begin{proposition} \label{prop:beta-glm}
Let $m_i' \coloneqq m' ( X_i^\top \bbeta_0 )$ and $m_i'' \coloneqq m'' ( X_i^\top \bbeta_0 )$.  Under potential model misspecification,
\begin{align*}
\bb_1 ( \hbbeta ) & = - B^{-1} \underbrace{ \E \left[ m_i' (Y_i - m_i) \cdot X_i^\top B^{-1} X_i \cdot X_i \right] }_{ \eqcolon \, H_1 } - B^{-1} \underbrace{ \E \left[ \frac{1}{2} m_i'' \cdot X_i^\top B^{-1} M B^{-1} X_i \cdot X_i \right] }_{ \eqcolon \, H_2 }, \\
\bb_1 ( \tbbeta ) & = \bb_1 ( \hbbeta ) + B^{-1} \underbrace{ \E \left[ \frac{1}{2} m_i'' \cdot X_i^\top B^{-1} X_i \cdot X_i \right] }_{ \eqcolon \, H_3 }.
\end{align*}
\end{proposition}
\noindent The proof of Proposition~\ref{prop:beta-glm} can be found in Appendix~\ref{append:proof-bias-beta-glm}. Assuming that working models are correctly specified, $H_1 \equiv \bzero_p$, $H_2 \equiv H_3$, and thus $\bb_1 ( \tbbeta ) \equiv \bzero_p$ \citep{firth1993bias}. When they are misspecified, however, the first-order bias in $\tbbeta$ persists. 

We now state the main result of this section, which characterizes the asymptotic biases of the g-computation estimators $\hmu_a$ and $\tmu_a$, respectively. They are derived using the second-order stochastic expansion of the corresponding nuisance parameter estimators, further simplified using the relationship indicated in \eqref{eq:zhang}. The proof can be found in Appendix~\ref{eq:append-proof-bias-gc}. 
\begin{theorem} \label{thm:gc}
Let $m_{i|a}'' \coloneqq m'' ( X_{i|a}^\top \bbeta_0 )$. Under potential model misspecification, $\E [ \hmu_a - \mu_a ] = n^{-1} b_1 ( \hmu_a ) + O ( n^{-3/2} )$, where $b_1 ( \hmu_a ) \coloneq b_1^{(1)} ( \hmu_a ) + b_1^{(2)} ( \hmu_a )$ with 
\begin{align*}
b_1^{(1)} ( \hmu_a ) & \coloneq - ( 1 - \pi_a ) \E \left[ m_{i|a}' \cdot X_{i|a}^\top B^{-1} X_{i|a} \cdot \{ Y_i (a) - m_{i|a} \} \right], \\
b_1^{(2)} ( \hmu_a ) & \coloneq \SUM{ b \neq a }{} \pi_b \E \left[ m_{i|a}' \cdot X_{i|a}^\top B^{-1} X_{i|b} \cdot \{ Y_i (b) - m_{i|b} \} \right].
\end{align*}
In addition, $\E ( \tmu_a - \mu_a ) = n^{-1} b_1 ( \tmu_a ) + O ( n^{-3/2} )$, where $b_1 ( \tmu_a ) \coloneq b_1^{(1)} ( \tmu_a ) + b_1^{(2)} ( \tmu_a )$ with $b_1^{(1)} ( \tmu_a ) \equiv b_1 ( \hmu_a )$ and
\[
b_1^{(2)} ( \tmu_a ) \coloneq \frac{1}{2} \E \left[ m_{i|a}'' \cdot X_{i|a}^\top B^{-1} X_{i|a} \right].
\]
\end{theorem}

\begin{remark} \label{rem:gc-bias}
For stratified working models, we have that $X_{1|a}^\top B^{-1} X_{1|b} \equiv 0$, which implies that $b_1^{(2)} ( \hmu_a ) \equiv 0$.  On the contrary, for pooled working models, $b_1^{(2)} ( \hmu_a )$ is generally nonzero.  
\end{remark}
\noindent Our proof in Appendix~\ref{eq:append-proof-bias-gc} elucidates the two sources of bias of $\hmu_a$, through $b_1^{(1)} ( \hmu_a )$ and $b_1^{(2)} ( \hmu_a )$, respectively. The former one is caused by the part of the estimation error in $\hbbeta$ (i.e., $-B^{-1} H_1$) carried over through those $\hwm_{i|a}$s under $A_i \neq a$, while the latter one is caused by using those same $X_i$s under $A_i \neq a$ to estimate both the nuisance parameters and treatment-specific means. It is noteworthy that, thanks to the prediction unbiasedness \eqref{eq:predict-mle}, those $\hwm_{i|a}$s with $A_i = a$ do not contribute to $b_1 ( \hmu_a )$.  For $\tmu_a$, it shares the same two sources (as indicated by $b_1^{(1)} ( \tmu_a ) = b_1( \hmu_a )$) and has one additional source of bias ($b_1^{(2)} ( \tmu_a )$), due to the augmented term in \eqref{eq:score-fc}. 

Finally, we analyze the magnitude of $b_1 ( \hmu_a )$ and $b_1 ( \tmu_a )$. Recall that $r_{i|a} $ denotes the true outcome model give $W_i$ under $A_i=a$. Apparently, $b_1^{(1)} ( \hmu_a ) \equiv 0$ if $m_{i|a} \equiv r_{i|a}$ and $b_1^{(2)} ( \hmu_a ) \equiv 0$ if $m_{i|b} \equiv r_{i|b}$ for all $b \neq a$. This indicates that, when working models are correctly specified, $\hmu_a$ is automatically free of first-order bias regardless of the magnitude of $\bb_1 ( \hbbeta )$. In contrast, such bias still persists for $\tmu_a$ since $b_1 ( \tmu_a ) \equiv b_1^{(2)} ( \tmu_a )$ which is generally nonzero.  

Under misspecified working models, we can bound $b_1^{(1)} ( \hmu_a )$, $b_1^{(2)} ( \hmu_a )$ and $b_1^{(2)} ( \tmu_a )$ as in Proposition~\ref{prop:bound-gc} below. The proof can be found in Appendix~\ref{append:proof-bound-gc}. 
\begin{proposition} \label{prop:bound-gc}
Recall that $p$ is the dimension of $X_i$. For pooled working models, 
\begin{align*}
\abs{ b_1^{(1)} ( \hmu_a ) } & \leq c_a ( \pi_a^{-1} - 1 ) \cdot p \cdot \norm{ m_{i|a} - r_{i|a} }_\infty, \\
\abs{ b_1^{(2)} ( \hmu_a ) } & \leq c_a^\ast c_a^{\ast\ast} \cdot p \cdot \sup_{ b \neq a }\norm{ m_{i|b} - r_{i|b} }_\infty, \\
\abs{ b_1^{(2)} ( \tmu_a ) } & \leq c_a (2\pi_a)^{-1} \cdot p \cdot \norm{ m_{i|a}'' / m_{i|a}' }_\infty,
\end{align*}
where $c_a, c_a^{\ast\ast} \in (0 ,1)$ and $c_a^\ast>0$.  Besides, for stratified working models, recalling that $p_a$ is the dimension of $Z_{i|a}$, 
\begin{gather*}
\abs { b_1 ( \hmu_a ) } \leq ( \pi_a^{-1} - 1 ) \cdot p_a \cdot \norm{ m_{i|a} - r_{i|a} }_\infty, \\
\abs{ b_1 ( \tmu_a ) - b_1 ( \hmu_a ) } \leq ( 2 \pi_a )^{-1}\cdot p_a \cdot \norm{ m_{i|a}'' / m_{i|a}' }_\infty.
\end{gather*}
\end{proposition}
\noindent Proposition~\ref{prop:bound-gc} suggests that $b_1 ( \hmu_a )$ is determined by the bias of the working model from the true model ($\norm{ m_{i|a} - r_{i|a} }_\infty$) and the degree of freedom of the working model ($p$). This result implies a trade-off between the misspecification bias and the estimation error. Though introducing a flexible working model with more parameters could potentially reduce the misspecification bias, it will inevitably increase the estimation error through a larger $p$. At finite sample sizes, all the potential gain via a flexible working model (including more efficiency gain) would be compromised by the amplified estimation error.  

In addition, $b_1 (\hmu_a)$ can be negligible even under misspecified working models in certain scenarios. For example, when events are rare, say $r_{i|a}$ and $m_{i|a}$ are of $o_{\sp} ( p^{-1} )$, $p \times \abs{ m_{i|a} - r_{i|a} }$ also decays to zero. Another occasion is when near perfect data separation for binary outcomes with non-rare events, $r_{i|a}, m_{i|a} \approx$ either $0$ or $1$ at the same $W_i$ for the majority of individuals and thus $\abs{ m_{i|a} - r_{i|a} } \approx 0$.  

On the contrary, $b_1 ( \tmu_a )$ is generally nonzero, unless $m_{i|a}'' / m_{i|a}' \equiv 0$ for all $W_i$. This is impossible for the log link since $m_{i|a}'' / m_{i|a}' \equiv 1$. For the logit link, this implies that $m_{i|a} \equiv 1/2$, a condition that is unlikely to occur in practice. Moreover, $b_1 ( \tmu_a )$  increases with $p$. This result holds even under a correctly specified working model, where $\tbbeta$ is free from first-order bias. This phenomenon helps explain the simulation results shown in Figure~\ref{fig:gc_fc}.

\section{Methodology: Bias Reduction} \label{sec:method}

In this section, we present debiased estimators for $\bmu$ (recall that $\bmu=(\mu_1,\ldots,\mu_k)^\top$) which are bounded (Section~\ref{sec:method-mu}) and a small-sample bias adjustment for the corresponding variance estimators (Section~\ref{sec:method-var}). The adjusted variance estimators can also be used for the standard g-computation. Subsequently, the debiased estimator for the treatment effect, $\delta ( \mu_a, \mu_b )$, can be obtained by simply replacing $(\mu_a, \mu_b)$ with their corresponding debiased estimators.

\subsection{Bias correction in treatment-specific mean estimation} \label{sec:method-mu}

A na\"{i}ve approach to debiasing $\hmu_a$ or $\tmu_a$ is to replace $\hbbeta$ or $\tbbeta$ with a corresponding bias-corrected estimator that removes the first-order bias. Such nuisance parameter estimators can be constructed by explicitly subtracting off the estimated bias, as characterized in Proposition~\ref{prop:beta-glm}. However, the resulting estimators of $\mu_a$ do not necessarily eliminate the first order bias of $\hmu_a$ or $\tmu_a$, as presented in Theorem~\ref{thm:gc}. An alternative approach for debiasing is to directly subtract the estimated first-order bias from the corresponding g-computation estimator. However, those estimators are not bounded. 

We first develop two estimators to reduce the bias of $\hmu_a$. By leveraging the prediction unbiasedness \eqref{eq:predict-mle}, \citet{guo2023generalized} provide an alternative formulation of g-computation estimators, namely as the so-called gOB estimators. In particular, $\hmu_a$ can be equivalently represented as 
\begin{equation}
\hmu_a \equiv \frac{1}{n} \left\{ \SUM{ i \colon A_i = a }{} Y_i + \SUM{ i \colon A_i \neq a }{} m ( X_{i|a}^\top \hbbeta ) \right\}, \label{eq:gob-mle} 
\end{equation}
which simply replaces those $\hwm_{i|a}$s of $A_i = a$ by $Y_i$. This formulation also implies that the bias incurred by $\hbbeta$ contributes to the bias of $\hmu_a$ only through those $\hwm_{i|a}$s with $A_i \neq a$. The following lemma formally demonstrates this property.
\begin{lemma}  \label{lem:gc}
$
\E [ \hwm_{i|a} - \mu_a ] = ( 1 - \pi_a ) \cdot \E [ \hwm_{i|a} - m_{i|a} | A_i \neq a ]. 
$
\end{lemma}
\noindent The proof can be found in Appendix~\ref{append:proof-gc-bias-conditional}. Building upon the gOB estimator in \eqref{eq:gob-mle}, we now describe our two debiased gOB estimators with $\hmu_a$. Let $\hh_{ii}$ and $\hbpsi^{\bbeta}_i$ be the estimated leverage score and IF of $\hbbeta$, respectively. The two proposed estimators, denoted by $\hmu_a^1$ and $\hmu_a^2$, is obtained by replacing $X_{i|a}^\top \hbbeta$ in \eqref{eq:gob-mle} with $X_{i|a}^\top \hbbeta^1$ or $X_{i|a}^\top \hbbeta_i^2$, respectively:
\begin{equation}
\label{eq:gob-mle-bc} 
\hbbeta^1 \coloneqq \hbbeta + \frac{1}{n} \SUM{i=1}{n} \hh_{ii} \hbpsi^{\bbeta}_i, \quad 
\hbbeta_i^2 \coloneqq \hbbeta^1 - \frac{1}{n} \, \hbpsi^{\bbeta}_i.
\end{equation}

\begin{theorem} \label{thm:gob-mle}
With potential model misspecification, $\E [ \hmu_a^1 - \mu_a ] = O ( n^{-3/2} )$ under stratified working models, or $n^{-1} b_1^{(2)} ( \hmu_a ) + O ( n^{-3/2} )$ under pooled working models; and under pooled working models, $\E [ \hmu_a^2 - \mu_a ] = O ( n^{-3/2} )$.
\end{theorem} 
\noindent The proof can be found in Appendix~\ref{append:proof-gob-mle}. Under pooled working models, the first-order bias of $\hmu_a^1$ still persists since the modification for $\hbbeta$ through adding $n^{-1} \SUM{i}{} \hh_{ii} \hbpsi^{\bbeta}_i$ only removes $n^{-1} b_1^{(1)} ( \hmu_a )$. To completely remove $n^{-1} b_1 ( \hmu_a )$, a further modification is required. Specifically, for each $X_{i|a}$ ($A_i \neq a$), from its associated $\hbbeta^1$ we subtract one additional term $n^{-1} \hbpsi_i^{\bbeta}$.  

Next, motivated by $\hmu_a^1$ and $\hmu_a^2$ we develop the debiased estimators with $\tmu_a$. Our next three gOB estimators, denoted by $\tmu_a^0$, $\tmu_a^1$ and $\tmu_a^2$, are then constructed by replacing $X_{i|a}^\top \hbbeta$ in \eqref{eq:gob-mle} with $X_{i|a}^\top \tbbeta^0$, $X_{i|a}^\top \tbbeta^1$ or $X_{i|a}^\top \tbbeta_i^2$, respectively:
\begin{equation}
\label{eq:gob-fc-bc} 
\begin{gathered}
\tbbeta^0 \coloneqq \tbbeta - \frac{1}{2n} \tB^{-1} \SUM{i=1}{n} \th_{ii} \cdot \frac{\twm_i''}{\twm_i'} \cdot X_i, \\
\tbbeta^1 \coloneqq \tbbeta^0 + \frac{1}{n} \SUM{i=1}{n} \th_{ii} \tbpsi_i^{\bbeta}, \quad 
\tbbeta_i^2 \coloneqq \tbbeta^1 - \frac{1}{n} \, \tbpsi^{\bbeta}_i.
\end{gathered}
\end{equation}

\begin{theorem} \label{thm:gob-fc}
With potential model misspecification, $\E [ \tmu_a^0 - \mu_a ] = n^{-1} b_1^{(1)} ( \tmu_a ) + O ( n^{-3/2} )$; $\E [ \tmu_a^1 - \mu_a ] = O ( n^{-3/2} )$ under stratified working models, or $n^{-1} b_1^{(2)} ( \hmu_a ) + O ( n^{-3/2} )$ under pooled working models; and under pooled working models, $\E [ \tmu_a^2 - \mu_a ] = O ( n^{-3/2} )$.
\end{theorem}
\noindent The proof can be found in Appendix~\ref{append:proof-gob-fc}. Our first proposal in \eqref{eq:gob-fc-bc}, only removes $n^{-1} b_1^{(2)} ( \tmu_a )$, which is the bias due to $\Delta^{(n)} (\bbeta)$ in \eqref{eq:score-fc}, while the last two proposed estimators are in the same fashion as the two in \eqref{eq:gob-mle-bc}. When $p/n$ is negligible, $n^{-1} b_1^{(1)} ( \tmu_a ) = n^{-1} b_1 ( \hmu_a )$ is marginal and $n^{-1} b_1^{(2)} ( \tmu_a )$ dominates $n^{-1} b_1 ( \tmu_a )$ (Proposition~\ref{prop:bound-gc}). Therefore, $\tmu_a^0$ is particularly suitable for low-dimensional settings, which are common in RCTs. 

All of our proposed debiased estimators in \eqref{eq:gob-mle-bc} and \eqref{eq:gob-fc-bc} are easy to implement as the leverage score, bread matrix, and IF can be directly obtained from the output of fitting a GLM using off-the-shelf software packages. Moreover, we provide the specific versions for stratified working models in Appendix~\ref{append:wm-strat-br}.

\begin{figure}[t!]
\centering
\includegraphics[width=\linewidth]{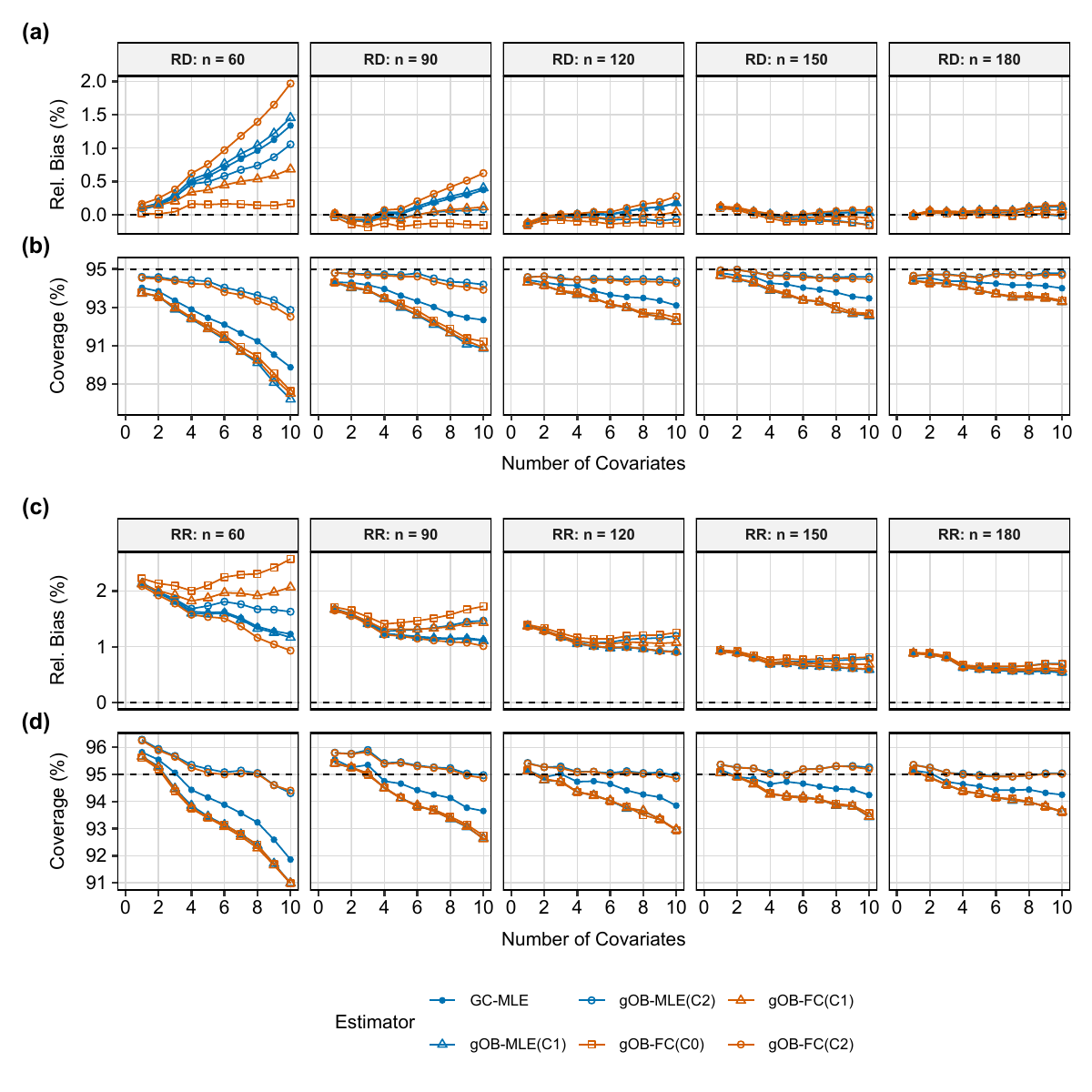}
\caption{
The bias of treatment effect estimates and coverage of the 95\% CI for \textbf{GC-MLE} ($\hmu_a$), \textbf{gOB-MLE(C1)} ($\hmu_a^1$), \textbf{gOB-MLE(C2)} ($\hmu_a^2$), \textbf{gOB-FC(C0)} ($\tmu_a^0$), \textbf{gOB-FC(C1)} ($\tmu_a^1$) and \textbf{gOB-FC(C2)} ($\tmu_a^2$). The debiased gOB estimators use the (\textit{theoretical}) IF-based variance estimator for interval estimation. RD: $\mu_2-\mu_1$; RR: $\mu_1/\mu_2$; $n$: sample size; Rel.: Relative. 
}
\label{fig:bc_gob}
\end{figure}

\paragraph{Simulation Experiment I(b)}  Figure~\ref{fig:bc_gob}a~and~\ref{fig:bc_gob}c present results from \textbf{Experiment I}, continuing the evaluation of bias in the debiased gOB estimators proposed in \eqref{eq:gob-mle-bc} and \eqref{eq:gob-fc-bc}. We denote these estimators as \textbf{gOB-MLE(C1)} ($\hmu_a^1$), \textbf{gOB-MLE(C2)} ($\hmu_a^2$), \textbf{gOB-FC(C0)} ($\tmu_a^0$), \textbf{gOB-FC(C1)} ($\tmu_a^1$), and \textbf{gOB-FC(C2)} ($\tmu_a^2$). Because the working models are pooled, only \textbf{gOB-MLE(C2)} and \textbf{gOB-FC(C2)} eliminate the $O(n^{-1})$ bias (i.e., they are fully debiased), whereas the remaining variants retain residual $O(n^{-1})$ bias (i.e., they are only partially debiased). Results for \textbf{GC-FC} are omitted because its bias (see Figure~\ref{fig:gc_fc}b~and~\ref{fig:gc_fc}d) is substantially larger than that of the other estimators.

Across all three FC-based debiased estimators, the bias for both $\mu_2-\mu_1$ and $\mu_1/\mu_2$ is substantially reduced relative to \textbf{GC-FC}. Their bias profiles are comparable to that of \textbf{GC-MLE} and, in many scenarios, even smaller. When sample sizes are at least 120 or the number of adjusted covariates is fewer than six, the bias profiles of the five debiased gOB estimators and \textbf{GC-MLE} are largely indistinguishable, despite modest variation across scenarios defined by sample size, covariate count, and estimand (difference/ratio). However, when $n=60$, \textbf{gOB-FC(C0)} attains the smallest bias for estimating $\mu_2-\mu_1$, whereas for $\mu_1/\mu_2$ it exhibits the largest bias. The same pattern holds when $n=90$ and the number of adjusted covariates exceeds six.

Figure~\ref{fig:mse} (Appendix~\ref{append:sim-res}) reports the root mean square error (RMSE) of all debiased gOB estimators, expressed relative to the RMSE of the corresponding \textbf{GC-MLE}. In most scenarios, FC-based gOB estimators yield lower RMSE than their MLE-based counterparts; in particular, \textbf{gOB-FC(C0)} attains the smallest RMSE in nearly all scenarios. Nevertheless, the RMSE differences among the debiased estimators are generally modest.

Figure~\ref{fig:bc_gob}b~and~\ref{fig:bc_gob}d summarize the interval coverage of the proposed gOB estimators, with variance estimation based on the \textit{theoretical} IF \eqref{eq:if-mu-aipw}. While \textbf{gOB-MLE(C1)}, \textbf{gOB-FC(C0)}, and \textbf{gOB-FC(C1)} under-perform \textbf{GC-MLE}, \textbf{gOB-MLE(C2)} and \textbf{gOB-FC(C2)} show comparably strong improvements. For $\mu_1/\mu_2$, the coverage probabilities remain close to or slightly above the nominal level in most scenarios (i.e., adjusting for up to 8 covariates when $n=60$ and up to 10 covariates when $n=90-180$). For $\mu_2 -\mu_1$, the improved coverage probabilities still falls below the nominal level in most scenarios.

\subsection{Small-sample bias adjustment for variance estimation} \label{sec:method-var}

The proposed small-sample bias adjustment for variance estimation is motivated by the gOB estimator with $\hbbeta^1$ in \eqref{eq:gob-mle-bc}. Specifically, $\hbbeta^1$ can be treated as an alternative estimator of the nuisance parameter ($\bbeta_0$), as $\hbbeta^1 = \hbbeta + O_{\sp} ( n^{-1} )$. This suggests that variance estimation for $\hmu_a^1$ can be improved by accounting for the difference between $\hbbeta^1$ and $\hbbeta$. In the following, we propose a simple small-sample bias adjustment by leveraging this difference via the linearization method \citep{deville1999variance}, which can also be applied for variance estimation of $\hbmu$ and $\hbmu^2$. Analogously, such small-sample bias adjustment can also be constructed for FC-based estimators ($\hbmu^1$), which is also applicable for variance estimation of $\tbmu^0$ and $\tbmu^2$.

The linearization method for variance estimation of an estimator relies on deriving its \textit{linearized variable}. In the following theorem, we show that both $\hmu_a^1$ and $\tmu_a^1$ can be written in terms of the same linearized variable. The proof can be found in Appendix~\ref{append:proof-linear}. 
\begin{theorem} \label{thm:linear}
Under potential model misspecification, $\hmu_a^1 = n^{-1} \SUM{i=1}{n} \ell_i ( \hmu_a^1 ) + o_{\sp} ( n^{-1/2} )$ and $\tmu_a^1 = n^{-1} \SUM{i=1}{n} \ell_i ( \tmu_a^1 ) + o_{\sp} ( n^{-1/2} )$, where 
\begin{equation}
\ell_i ( \hmu_a^1 ) = \ell_i ( \tmu_a^1 ) = \pi_a^{-1} I ( A_i = a ) ( Y_i - m_i ) \cdot ( 1 + \barh_{ii} ) + m_{i|a} - \mu_a, \label{eq:if-adj}
\end{equation}
and $\barh_{ii} \coloneqq n^{-1} m_i' X_i^\top B^{-1} X_i$.
\end{theorem}
\noindent The above \textit{linearized variable} is related with the theoretical IF presented in \eqref{eq:if-mu-aipw}, except that the first term in \eqref{eq:if-mu-aipw} is multiplied by $1 + \barh_{ii}$ (and thus its expectation is nonzero), where $\barh_{ii}$ can be simply estimated by $\hh_{ii}$ or $\th_{ii}$, the corresponding leverage score. It is noteworthy that $\ell_i ( \hmu_a^1 )$ can be also written as $ \E [ m_{1|a}' \cdot X_{1|a}^\top ] \bpsi_i^{\bbeta} \cdot ( 1 + \barh_{ii} ) + m_{i|a} - \mu_a$. This alternative formulation is related with the empirical IF presented in \eqref{eq:if-mu-aipw}, except that $\bpsi^{\bbeta}$ in \eqref{eq:if-mu-aipw} is replaced by $( 1 + \barh_{ii} ) \cdot \bpsi^{\bbeta}$. In fact, $\ell_i ( \hmu_a^1 )$ is developed from this formulation; see Appendix~\ref{append:proof-linear}. The additional term (i.e., $\barh_{ii} \bpsi^{\bbeta}_i$) accounts for the difference between $\hbbeta^1$ and $\hbbeta$.

Let $\ell_i^\ast$ be an estimate of \textit{linearized variable} defined in \eqref{eq:if-adj} using corresponding estimates of $(\barh_{ii}, m_i, m_{i|a}, \mu_a)$. Variance estimation for $\hbmu^1$ and $\tbmu^1$ is convenient to implement in the same manner as the IF-based variance estimator. The corresponding variance estimator is then formally written as
\begin{equation}
\frac{1}{n(n-1)} \SUM{i=1}{n} \Big( \bl_i^\ast - \frac{1}{n} \SUM{j=1}{n} \bl_j^\ast \Big)^{\otimes2}, \label{eq:var-linear}
\end{equation}
which is the sample variance of $\bl_i^\ast = (\ell_1^\ast,\ldots,\ell_n^\ast)^\top$ scaled by $1/n$ \citep{deville1999variance}. Moreover, Theorem~\ref{thm:linear} and \eqref{eq:var-linear} are applicable to $\hbmu$, $\hbmu^2$, $\tbmu$, $\tbmu^0$, and $\hbmu^2$, since all of these estimators share the same asymptotic normal distribution. 

The proposed adjusted variance estimators are as easy to implement as standard IF-based variance estimators. They only modify the first term in the estimated \textit{theoretical} IF \eqref{eq:if-mu-aipw} by a multiplier, which is simply one plus the leverage score. This multiplier accounts for the estimation error in the estimated nuisance parameter, which contributes to the overall estimation error in g-computation estimators.  

\begin{remark}
Our proposal for small-sample bias adjustment is closely related to the well-known HC3 correction \citep{long2000using,mancl2001covariance}.  The multiplier in HC3 can be written as a geometric series,
$
(1-\hh_{ii})^{-1} = 1 + \hh_{ii} + \hh_{ii}^2 + \hh_{ii}^3 + \cdots,
$
provided that $\hh_{ii} < 1$.  The first two terms in this geometric series exactly constitute the multiplier in our proposal.  
\end{remark}

\begin{figure}[t!]
\centering
\includegraphics[width=\linewidth]{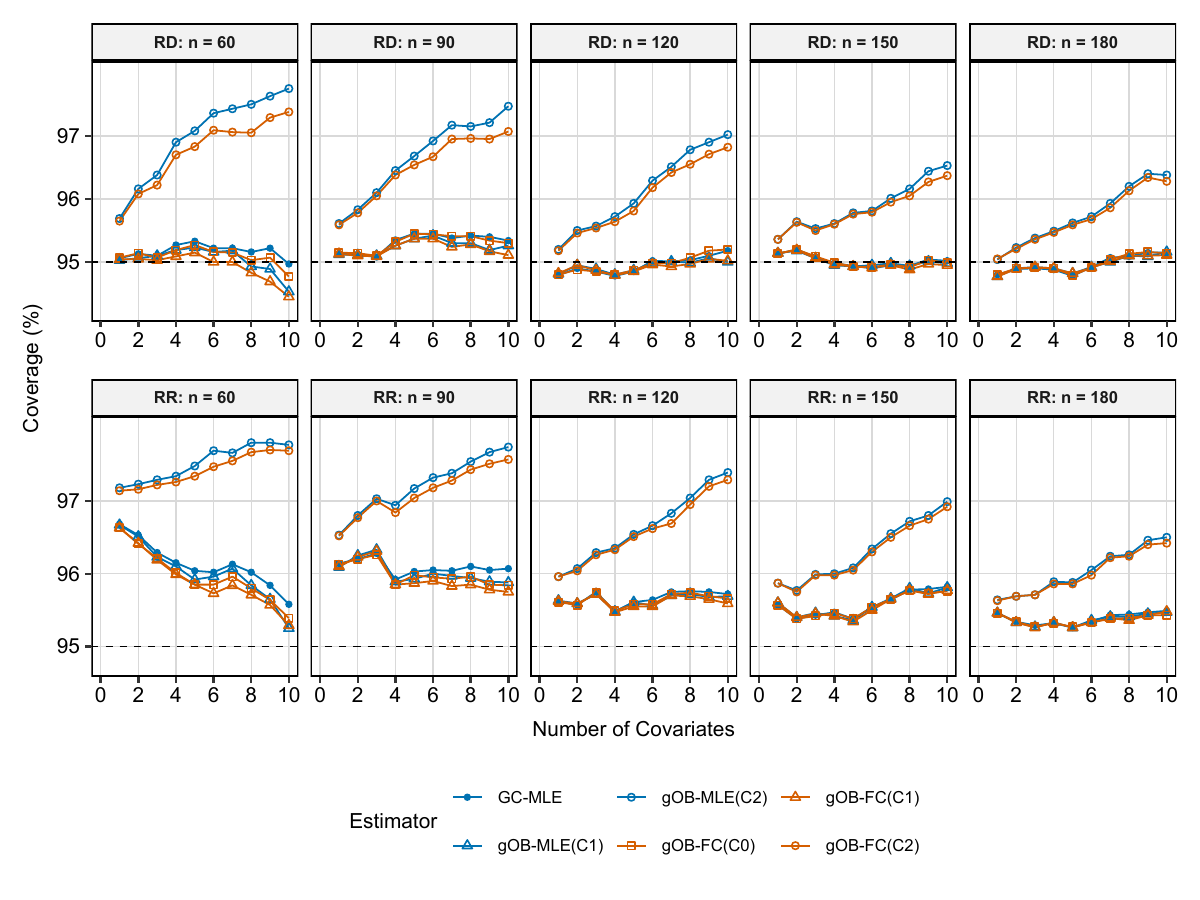}
\caption{
The coverage of the 95\% CI with the proposed small-sample bias adjustment for \textbf{GC-MLE} ($\hmu_a$), \textbf{gOB-MLE(C1)} ($\hmu_a^1$), \textbf{gOB-MLE(C2)} ($\hmu_a^2$), \textbf{gOB-FC(C0)} ($\tmu_a^0$), \textbf{gOB-FC(C1)} ($\tmu_a^1$) and \textbf{gOB-FC(C2)} ($\tmu_a^2$). RD: $\mu_2-\mu_1$; RR: $\mu_1/\mu_2$; $n$: sample size. 
}
\label{fig:bc_gob_adj}
\end{figure}

\paragraph{Simulation Experiment I(c)} 

Figure~\ref{fig:bc_gob_adj} presents the simulation results for the interval coverage of \textbf{GC-MLE} and the debiased gOB estimators with the proposed small-sample bias adjustment. For $\mu_2-\mu_1$, \textbf{GC-MLE}, \textbf{gOB-MLE(C1)}, \textbf{gOB-FC(C0)}, and \textbf{gOB-FC(C1)} show substantial improvements in interval coverage, staying close to or above the nominal level. In particular, with the adjusted variance estimator, \textbf{GC-MLE} maintains near or above nominal coverage as the number of adjusted covariates increases when $n=60$. By contrast, \textbf{gOB-MLE(C2)} and \textbf{gOB-FC(C2)} exhibit over-coverage when combined with the proposed small-sample bias adjustment; the same over-coverage occurs for all estimators for $\mu_1/\mu_2$.

Figure~\ref{fig:width} (Appendix~\ref{append:sim-res}) reports the 95\% CI widths for \textbf{GC-MLE} and all debiased gOB estimators, computed with the adjusted variance estimator. Widths are expressed relative to those of the corresponding \textbf{GC-MLE} under the adjusted variance estimator. \textbf{gOB-MLE(C2)} and \textbf{gOB-FC(C2)} yield the widest intervals, consistent with the over-coverage in Figure~\ref{fig:bc_gob_adj}. For the remaining estimators, interval-width profiles are largely indistinguishable.

\subsection{Recommendation} \label{sec:rec}

In the following, we provide guidance on the choice of estimation methods when pooled working models are used; for stratified working models, see Appendix~\ref{append:sim-2}.

For RD, in most simulation settings ($n\geq90$), $\hmu_a$ (\textbf{GC-MLE}), $\hmu_a^1$ (\textbf{gOB-MLE(C1)}), $\tmu_a^0$ (\textbf{gOB-FC(C0)}), and $\tmu_a^1$ (\textbf{gOB-FC(C1)}) perform well, exhibiting similar bias and interval coverage with the bias-adjusted variance estimator. Although the RMSE and the interval width vary across scenarios, the differences are modest. In summary, $\tmu_a^0$ is preferable when the sample size is small or data separation arises; otherwise $\hat{\mu}_{a}$ can still be used. 

For RR, $\hmu_a^2$ (\textbf{gOB-MLE(C2)}) and $\tmu_a^2$ (\textbf{gOB-FC(C2)}) with the standard variance estimator exhibit low bias and near-nominal coverage. Accordingly, we recommend $\hmu_a^2$ by default; when data separation arises, $\tmu_a^2$ is preferable.

\section{Application} \label{sec:app}

CTN-03 study \citep{ling2009buprenorphine} was a RCT to compare two taper schedules following a period of physiological stabilization on buprenorphine for opioid dependent individuals. In this trial, participants were randomized at a 1:1 ratio to receive either a 28-day or a 7-day taper. We define the former intervention arm as the control arm, and the latter as the treatment arm. The randomization was stratified by the maintenance dose (8, 16, and 24 mg). The objective of the statistical analysis was to compare the proportion of participants with opioid-free urine specimens at the end of the taper period between the two taper conditions. 

We perform two analyses for this trial to illustrate the utility of our proposed estimators. The first analysis (Table~\ref{tab:ctn03-subgroup}) demonstrates the utility of the proposed FC-based debiased estimator (Section~\ref{sec:method-mu}) under data separation, allowing covariate adjustment to improve precision. The details are provided in the rest of this section. The second analysis evaluates the benefit of our bias-adjusted variance estimator for \textbf{GC-MLE} (Section~\ref{sec:method-var}) in settings with many adjusted covariates. Owing to space constraints, this analysis is provided in Appendix~\ref{sec:app-1}.

We analyze the 8 mg maintenance-dose subgroup (48 participants: 26 control, 22 treatment). Implementing g-computation for this subgroup analysis is challenging because the small sample size creates a risk of data separation. Using a pooled working model (logistic regression) to adjust for a single covariate (baseline opioid urine toxicology level) produced an estimated coefficient of $-18.68$, implausibly large in magnitude, indicating potential data separation; this was confirmed using the \textbf{detectseparation} package \citep{kosmidis2022detectseparation}. Consequently, covariate adjustment is infeasible without employing FC.

Following on our recommendations in Section~\ref{sec:rec}, we estimate treatment effects using \textbf{gOB-FC(C0)}. We adopt a pooled working model, fitting a logistic regression that adjusts for sex and baseline opioid urine toxicology. As shown in Figure~\ref{fig:bc_gob}b, with the standard variance estimator, adjusting for two covariates at $n=60$ yields coverage of roughly 94\%. To mitigate this under-coverage, we report Wald CIs using the proposed bias–adjusted variance estimator. 

Table~\ref{tab:ctn03-subgroup} reports subgroup results for both the unadjusted analysis and \textbf{gOB-FC(C0)}. We also report results for \textbf{gOB-FC(C1)} and \textbf{gOB-FC(C2)} as sensitivity analyses. With FC, covariate adjustment becomes feasible, yielding an 8.41\% gain in relative efficiency. The unadjusted analysis yields a 95\% CI that does not rule out those values above 50\%, whereas the adjusted analysis using \textbf{gOB-FC(C0)} produces a 95\% CI with an upper bound of 48.90\%, suggesting the effect is below 50\%. 

Moreover, \textbf{gOB-FC(C1)} produces nearly identical (to \textbf{gOB-FC(C0)}) point/interval estimates, while \textbf{gOB-FC(C2)} yields a standard error slightly larger than \textbf{Unadj}. This aligns with our simulation results, indicating that \textbf{gOB-FC(C2)} should not be used when sample sizes are small.

\begin{table}[t!]
\centering
\begin{tabular}{l|ccccc}
\hline
    Estimator            &    RD (\%) & SE (\%) & $\mathrm{RE}_i$ (\%) & 95\% CI (\%)  & Width (\%) \\[0.5ex] \hline\hline
    \textbf{Unadj}       &   23.78    &   13.95 &                 NA   & -3.57, 51.12  & 54.69 \\
    \textbf{gOB-FC(C0)}  &   22.73    &   13.35 &                 8.41 & -3.43, 48.90  & 52.33 \\
    \textbf{gOB-FC(C1)}  &   22.73    &   13.35 &                 8.42 & -3.43, 48.90  & 52.33 \\
    \textbf{gOB-FC(C2)}  &   22.81    &   14.00 &                -0.74 & -4.63, 50.25  & 54.88 \\
    \hline
\end{tabular}
\caption{Summary of a subgroup analysis ($n=48$) for the CTN-03 study. RD: risk difference; SE: standard error; $\mathrm{RE}_i$: relative efficiency improvement to Unadj (one minus the ratio between the two variance estimates); CI: confidence interval; Width: interval width}.
\label{tab:ctn03-subgroup}
\end{table}

\section{Discussion} \label{sec:disc}

In this article, we develop a new bias-reduction approach for g-computation, that refines MLE- and FC-based nuisance estimation, improving point estimation (particularly with FC) and inference (for both MLE and FC). The focus is on improving finite-sample statistical properties, taking causal identification as established through randomization and g-computation. Our approach is easy to implement, requiring only minor modifications to standard point and variance estimators. The proposed debiased estimators take the gOB form and are bounded, and the bias-adjusted variance estimators are constructed by a simple modification of IFs using leverage scores. Simulation experiments, designed to mimic scenarios where conventional g-computation estimators fail, demonstrate the superior finite-sample performance of the proposed approach. The recommendations under different scenarios for RD/RR are provided in Section~\ref{sec:rec}. We further illustrate the practical utility of the proposed method in Section~\ref{sec:app}. 

In practice, caution is warranted when deciding whether to apply the proposed approach, particularly the bias‑adjusted variance estimator. When the number of covariates is small relative to the sample size, this adjustment can over-correct, producing unnecessarily wide CIs and reducing precision and power. We recommend using the proposed method only when simulation studies provide clear evidence of under-coverage in interval estimation.

Our proposed method has several limitations. Our method still needs positivity, which is reasonable to assume in RCTs in general. However, positivity violation can still occur through stratification by many strata or missingness in outcomes. If violated, our method cannot be directly applied without being modified. Our method is also developed exclusively for RCTs with i.i.d. observations. It is beyond the scope of this paper to generalize our method to observational studies or RCTs with non-i.i.d. observations. 

In Section~\ref{sec:app}, the small sample size ($n = 48$) limits the scope for systematic covariate selection in this subgroup analysis. However, since this is a RCT, the validity of the estimator is unaffected by the choice of adjustment covariates.

\FloatBarrier

\subsection*{Acknowledgments}

We thank two anonymous reviewers, the associate editor, and the co-editor for their helpful and constructive comments. We also thank Wei Ma, Linbo Wang, and Qingyuan Zhao for helpful discussions. The CTN-03 study dataset used in this article to illustrate our proposed method is available at \url{https://datashare.nida.nih.gov/study/nida-ctn-0003}. Lin Liu's research is supported by NSFC Grant No.12471274 and Science and Technology Talent and Platform Program of Yunnan Province Grant No.202605AF35007.
 
\bibliographystyle{apalike}        
\bibliography{main}

@article{cox1968general,
title = {A general definition of residuals},
author = {Cox, David R and Snell, Joyce E},
journal = {Journal of the Royal Statistical Society: Series B (Methodological)},
volume = {30},
number = {2},
pages = {248--265},
year = {1968}
}

@article{macrae1974matrix, 
  year     = {1974},
  title    = {Matrix Derivatives with an Application to an Adaptive Linear Decision Problem}, 
  author   = {{MacRae}, Elizabeth Chase}, 
  journal  = {The Annals of Statistics},
  pages    = {337--346}, 
  number   = {2}, 
  volume   = {2}
}

@article{albert1984existence, 
  year     = {1984}, 
  title    = {On the existence of maximum likelihood estimates in logistic regression models}, 
  author   = {Albert, A. and Anderson, J. A.}, 
  journal  = {Biometrika}, 
  pages    = {1--10}, 
  number   = {1}, 
  volume   = {71}
}

@article{robins1986new,
title = {A new approach to causal inference in mortality studies with a sustained exposure period--{Application} to control of the healthy worker survivor effect},
author = {Robins, James},
journal = {Mathematical Modelling},
volume = {7},
number = {9-12},
pages = {1393--1512},
year = {1986}
}

@article{cordeiro1991bias,
title = {Bias correction in generalized linear models},
author = {Cordeiro, Gauss M and McCullagh, Peter},
journal = {Journal of the Royal Statistical Society Series B: Statistical Methodology},
volume = {53},
number = {3},
pages = {629--643},
year = {1991}
}

@article{firth1993bias, 
year = {1993}, 
title = {Bias reduction of maximum likelihood estimates}, 
author = {Firth, David}, 
journal = {Biometrika}, 
pages = {27--38}, 
number = {1}, 
volume = {80}
}

@article{deville1999variance,
  title={Variance estimation for complex statistics and estimators: {Linearization} and residual techniques},
  author={Deville, Jean-Claude},
  journal={Survey Methodology},
  volume={25},
  number={2},
  pages={193--204},
  year={1999}
}

@article{long2000using, 
year = {2000}, 
title = {Using Heteroscedasticity Consistent Standard Errors in the Linear Regression Model}, 
author = {Long, J. Scott and Ervin, Laurie H.}, 
journal = {The American Statistician}, 
pages = {217--224}, 
number = {3}, 
volume = {54}
}

@article{mancl2001covariance, 
year = {2001}, 
title = {A covariance estimator for {GEE} with improved small‐sample properties}, 
author = {Mancl, Lloyd A and DeRouen, Timothy A}, 
journal = {Biometrics}, 
pages = {126--134}, 
number = {1}, 
volume = {57}
}

@article{heinze2002solution, 
year = {2002}, 
title = {A solution to the problem of separation in logistic regression}, 
author = {Heinze, Georg and Schemper, Michael}, 
journal = {Statistics in Medicine},
pages = {2409--2419}, 
number = {16}, 
volume = {21}
}

@article{stefanski2002calculus, 
year = {2002}, 
title = {The calculus of {$M$-estimation}}, 
author = {Stefanski, Leonard A and Boos, Dennis D}, 
journal = {The American Statistician}, 
pages = {29--38}, 
number = {1}, 
volume = {56}
}

@article{fay2004small, 
year = {2004}, 
title = {Small-sample adjustments for {Wald}-type tests using sandwich estimators}, 
author = {Fay, Michael P and Graubard, Barry I}, 
journal = {Biometrics}, 
pages = {1198--1206}, 
number = {4}, 
volume = {57}
}

@article{hernandez2004covariate,
  title={Covariate adjustment in randomized controlled trials with dichotomous outcomes increases statistical power and reduces sample size requirements},
  author={Hern{\'a}ndez, Adri{\'a}n V and Steyerberg, Ewout W and Habbema, J Dik F},
  journal={Journal of Clinical Epidemiology},
  volume={57},
  number={5},
  pages={454--460},
  year={2004}
}

@article{hernandez2006randomized,
  title={Randomized controlled trials with time-to-event outcomes: {H}ow much does prespecified covariate adjustment increase power?},
  author={Hern{\'a}ndez, Adri{\'a}n V and Eijkemans, Marinus JC and Steyerberg, Ewout W},
  journal={Annals of Epidemiology},
  volume={16},
  number={1},
  pages={41--48},
  year={2006}
}

@article{freedman2008randomization, 
  year     = {2008}, 
  title    = {Randomization Does Not Justify Logistic Regression}, 
  author   = {Freedman, David A}, 
  journal  = {Statistical Science}, 
  pages    = {237--249}, 
  number   = {2}, 
  volume   = {23}
}

@article{tsiatis2008covariate,
  title={Covariate adjustment for two-sample treatment comparisons in randomized clinical trials: {A} principled yet flexible approach},
  author={Tsiatis, Anastasios A and Davidian, Marie and Zhang, Min and Lu, Xiaomin},
  journal={Statistics in Medicine},
  volume={27},
  number={23},
  pages={4658--4677},
  year={2008}
}

@article{ling2009buprenorphine,
  title={Buprenorphine tapering schedule and illicit opioid use},
  author={Ling, Walter and Hillhouse, Maureen and Domier, Catherine and Doraimani, Geetha and Hunter, Jeremy and Thomas, Christie and Jenkins, Jessica and Hasson, Albert and Annon, Jeffrey and Saxon, Andrew and others},
  journal={Addiction},
  volume={104},
  number={2},
  pages={256--265},
  year={2009}
}

@article{ge2011covariate, 
year = {2011}, 
title = {Covariate-adjusted difference in proportions from clinical trials using logistic regression and weighted risk differences}, 
author = {Ge, Miaomiao and Durham, Kathryn L and Meyer, Daniel R and Xie, Wangang and Thomas, Neal}, 
journal = {Drug information journal : DIJ / Drug Information Association}, 
pages = {481--493}, 
number = {4}, 
volume = {45},
publisher={Drug Information Association}
}

@article{yuan2012variable, 
  year     = {2012}, 
  title    = {Variable selection for covariate‐adjusted semiparametric inference in randomized clinical trials}, 
  author   = {Yuan, Shuai and Zhang, Hao Helen and Davidian, Marie}, 
  journal  = {Statistics in Medicine}, 
  pages    = {3789--3804}, 
  number   = {29}, 
  volume   = {31}
}

@article{kahan2014risks,
  title={The risks and rewards of covariate adjustment in randomized trials: {An} assessment of 12 outcomes from 8 studies},
  author={Kahan, Brennan C and Jairath, Vipul and Dor{\'e}, Caroline J and Morris, Tim P},
  journal={Trials},
  volume={15},
  number={1},
  pages={139},
  year={2014},
  publisher={Springer}
}

@article{puhr2017firth, 
year = {2017}, 
title = {{Firth's} logistic regression with rare events: {Accurate} effect estimates and predictions?}, 
author = {Puhr, Rainer and Heinze, Georg and Nold, Mariana and Lusa, Lara and Geroldinger, Angelika}, 
journal = {Statistics in Medicine}, 
pages = {2302--2317}, 
number = {14}, 
volume = {36}, 
}

@article{sur2019modern,
  title={A modern maximum-likelihood theory for high-dimensional logistic regression},
  author={Sur, Pragya and Cand{\`e}s, Emmanuel J},
  journal={Proceedings of the National Academy of Sciences},
  volume={116},
  number={29},
  pages={14516--14525},
  year={2019}
}

@article{kosmidis2020jeffreys, 
year = {2020}, 
title = {Jeffreys-prior penalty, finiteness and shrinkage in binomial-response generalized linear models}, 
author = {Kosmidis, Ioannis and Firth, David}, 
journal = {Biometrika}, 
pages = {71--82}, 
number = {1}, 
volume = {108}
}

@article{benkeser2021improving,
  title={Improving precision and power in randomized trials for {COVID}-19 treatments using covariate adjustment, for binary, ordinal, and time-to-event outcomes},
  author={Benkeser, David and D{\'\i}az, Iv{\'a}n and Luedtke, Alex and Segal, Jodi and Scharfstein, Daniel and Rosenblum, Michael},
  journal={Biometrics},
  volume={77},
  number={4},
  pages={1467--1481},
  year={2021}
}

@article{joshi2022solutions, 
year = {2022}, 
title = {Solutions to problems of nonexistence of parameter estimates and sparse data bias in {Poisson} regression}, 
author = {Joshi, Ashwini and Geroldinger, Angelika and Jiricka, Lena and Senchaudhuri, Pralay and Corcoran, Christopher and Heinze, Georg}, 
journal = {Statistical Methods in Medical Research}, 
pages = {253--266}, 
number = {2}, 
volume = {31}
}

@article{lee2022benefits,
  title={The benefits of covariate adjustment for adaptive multi-arm designs},
  author={Lee, Kim May and Robertson, David S and Jaki, Thomas and Emsley, Richard},
  journal={Statistical Methods in Medical Research},
  volume={31},
  number={11},
  pages={2104--2121},
  year={2022}
}

@article{guo2023generalized,
  title={The generalized {Oaxaca-Blinder} estimator},
  author={Guo, Kevin and Basse, Guillaume},
  journal={Journal of the American Statistical Association},
  volume={118},
  number={541},
  pages={524--536},
  year={2023}
}

@article{tackney2023comparison, 
  year     = {2023}, 
  title    = {A Comparison of Covariate Adjustment Approaches under Model Misspecification in Individually Randomized Trials}, 
  author   = {Tackney, Mia S and Morris, Tim and White, Ian and Leyrat, Clemence and Diaz-Ordaz, Karla and Williamson, Elizabeth}, 
  journal  = {Trials}, 
  pages    = {14}, 
  number   = {1}, 
  volume   = {24}
}

@article{van2024covariate,
  title={Covariate adjustment in randomized controlled trials: {General} concepts and practical considerations},
  author={Van Lancker, Kelly and Bretz, Frank and Dukes, Oliver},
  journal={Clinical Trials},
  volume={21},
  number={4},
  pages={399--411},
  year={2024}
}

@article{kosmidis2024empirical, 
    year = {2024}, 
    title = {Empirical bias-reducing adjustments to estimating functions}, 
    author = {Kosmidis, Ioannis and Lunardon, Nicola}, 
    journal = {Journal of the Royal Statistical Society Series B: Statistical Methodology}, 
    pages = {62--89}, 
    number = {1}, 
    volume = {86}
}

@article{rilstone2024on, 
year = {2024}, 
title = {On the relationship between higher-order stochastic expansions, influence functions and {U-statistics for M-estimators}}, 
author = {Rilstone, Paul}, 
journal = {Communications in Statistics - Theory and Methods}, 
pages = {2103--2121}, 
number = {6}, 
volume = {53}
}

@article{zhang2025robust,
  title={A Robust Score Test in G-Computation for Covariate Adjustment in Randomized Clinical Trials Leveraging Different Variance Estimators via Influence Functions},
  author={Zhang, Xin and Chu, Haitao and Liu, Lin and Roychoudhury, Satrajit},
  journal={Statistics in Medicine},
  volume={44},
  number={7},
  pages={e70080},
  year={2025}
}

@incollection{boos2013essential,
author = {Boos, Dennis D and Stefanski, Leonard A},
title = {Jackknife},
booktitle = {Essential Statistical Inference: {Theory} and Methods},
publisher = {Springer New York},
volume = {120},
pages = {385--411},
year = {2013}
}

@book{mccullagh2018tensor,
  title={Tensor Methods in Statistics},
  author={McCullagh, Peter},
  series={Monographs on Statistics and Applied Probability},
  year={2018},
  publisher={Chapman and Hall/CRC},
  address = {Boca Raton, FL}
}

@misc{ema2015,
year = {2015},
month = {February},
title = {Guidelines on adjustment for baseline covariates in clinical trials},
author = {{European Medicines Agency}},
howpublished = {\url{https://www.ema.europa.eu/en/documents/scientific-guideline/guideline-adjustment-baseline-covariates-clinical-trials_en.pdf} (accessed Aug 20, 2025)}
}

@misc{fda2023,
year = {2023},
month = {May},
title = {Adjusting for covariates in randomized clinical trials for drugs and biological products},
author = {{US Food and Drug Administration}},
howpublished = {\url{https://www.fda.gov/media/148910/download} (accessed Aug 20, 2025)}
}

@misc{kosmidis2022detectseparation,
    title = {\texttt{detectseparation}: {Detect} and Check for Separation and Infinite Maximum Likelihood
Estimates},
    author = {Ioannis Kosmidis and Dirk Schumacher and Florian Schwendinger},
    year = {2022},
    howpublished = {\url{https://CRAN.R-project.org/package=detectseparation} (accessed Aug 20, 2025)}
  }

@misc{kosmidis2023brglm2,
    title = {\texttt{brglm2}: {Bias} Reduction in Generalized Linear Models},
    author = {Ioannis Kosmidis},
    year = {2023},
    howpublished = {\url{https://CRAN.R-project.org/package=brglm2} (accessed Aug 20, 2025)}
}

\clearpage

\allowdisplaybreaks

\appendix
\appendixpage
\beginappendix

\section{Main Proofs} \label{append:proof}

\subsection{Proof of \texorpdfstring{Lemma \ref{lem:predict-fc}}{Lemma 1}}  \label{append:proof-predict-fc}

From the modified score equation $\SUM{i}{} U_i ( \tbbeta ) + \Delta^{(n)} ( \tbbeta ) \equiv \bzero_p$, we have that 
\[
\SUM{i=1}{n} (Y_i - \twm_i) X_i + \frac{1}{2} \SUM{i=1}{n} \th_{ii} \twm_i'' / \twm_i' \cdot X_i \equiv \bzero_p.
\]
Since $I ( A_i = a )$ is one element in $X_i$, 
\begin{gather*}
\SUM{i=1}{n} I ( A_i = a ) ( Y_i - \twm_i )  + \frac{1}{2} \SUM{i=1}{n} I ( A_i = a ) \th_{ii} \twm_i'' / \twm_i' \equiv 0 \\
\implies \SUM{ i \colon A_i = a }{} ( Y_i - \twm_i ) + \frac{1}{2} \SUM{ i \colon A_i = a }{} \th_{ii} \twm_i'' / \twm_i' \equiv 0 \\ 
\implies \SUM{ i \colon A_i = a }{} \twm_i \equiv  \SUM{ i \colon A_i = a }{} \left( Y_i + \frac{1}{2} \th_{ii} \twm_i'' / \twm_i' \right).
\end{gather*}
To complete the proof, we show that $h_{ii} (\bbeta) = O_{\sp} ( n^{-1} )$. With $p$ fixed, it is not difficult to see that $\E [h_{ii} (\bbeta)] \lesssim p / n = O ( n^{-1} )$.  Furthermore, by similar arguments, $\E [ h_{ii} (\bbeta)^{2} ] \lesssim p^{2} / n^{2} = O ( n^{-2} )$.  Combing the above results, we have that $h_{ii} (\bbeta) = O_{\sp} ( p / n ) = O_{\sp} ( n^{-1} )$, and thus $n \tilde{h}_{ii} = O_{\sp} ( 1 )$.  Above all,
\[
\frac{1}{n} \SUM{i \colon A_i = a}{} n \th_{ii} \twm_i'' / \twm_i' = O_{\sp} (1),
\]
which completes the proof.

\subsection{Proof of \texorpdfstring{Proposition \ref{prop:beta-glm}}{Proposition 1}}  \label{append:proof-bias-beta-glm}

The $r$th element of $B \bb_1 ( \hbbeta )$ is $\DERIV{ \tr ( B^{-1} M ) }{r} / 2$, which, according to \eqref{proof-mle-bias-vector-ext}, can be written as 
\begin{equation}
\ \tr \left( B^{-1} \E \left[ U_i \DERIV{ U_i }{r}^\top \right] \right) + \frac{1}{2} \tr \left( B^{-1} M B^{-1} \E \left[ \DDERIV{ U_{ir} }{} \right] \right), \label{eq:proof-mle-bias-glm}    
\end{equation}
and the $r$th element of $B \bb_1 ( \tbbeta )$ is $\DERIV{ \tr ( B^{-1} M ) }{r} / 2 + \DERIV{ \log \det ( B ) }{r} / 2$, which, according to \eqref{eq:proof-fc-bias-vector}, can be written as
\begin{equation}
\eqref{eq:proof-mle-bias-glm} - \frac{1}{2} \tr \left( B^{-1} \E \left[ \DDERIV{ U_{ir} }{} \right] \right). \label{eq:proof-fc-bias-glm}         
\end{equation}
Since $U_i = (Y_i - m_i) X_i$, we have $\DERIV{ U_i }{r} = - m_i' X_{ir} \cdot X_i$, $\DDERIV{ U_{ir} }{} = - m_i'' X_{ir} \cdot X_i X_i^\top$, and thus 
\begin{align*}
& \E \left[ U_i \DERIV{U_i}{r}^\top \right] = - \E \left[ m_i' (Y_i - m_i) X_{ir} \cdot X_i X_i^\top \right], \\
& \E \left[ \DDERIV{U_{ir}}{} \right] = - \E \left[ m_i'' X_{ir} \cdot X_i X_i^\top  \right].
\end{align*}
Then for the first term in \eqref{eq:proof-mle-bias-glm}, we have that
\begin{align*}
\tr \left( B^{-1} \E \left[ U_i \DERIV{ U_i }{r}^\top \right] \right)
& = - \E \left[ \tr \left( B^{-1} X_i X_i^\top \cdot X_{ir} (Y_i - m_i) m_i' \right) \right] \\
& = - \E \left[ \tr \left( X_i^\top B^{-1} X_i\right) \cdot X_{ir} (Y_i - m_i) m_i' \right] \\
& = - \E \left[ X_i^\top B^{-1} X_i \cdot X_{ir} (Y_i - m_i) m_i' \right].
\end{align*}
For the second term in \eqref{eq:proof-mle-bias-glm}, we have that 
\begin{align*}
\tr \left( B^{-1 } M B^{-1} \E \left[ \DDERIV{U_{ir}}{} \right] \right) 
& = - \E \left[ \tr \left( B^{-1 } M B^{-1} \cdot X_i X_i^\top \cdot X_{ir} m_i'' \right) \right] \\
& = - \E \left[ \tr \left( X_i^\top B^{-1 } M B^{-1} X_i \right) \cdot X_{ir} m_i'' \right] \\
& = - \E \left[ X_i^\top B^{-1 } M B^{-1} X_i \cdot X_{ir} m_i'' \right].
\end{align*}
And for the last term in \eqref{eq:proof-fc-bias-glm}, we have that 
\begin{align*}
\tr \left( B^{-1} \E \left[ \DDERIV{U_{ir}}{} \right] \right) 
& = - \E \left[ \tr \left( B^{-1} X_i X_i^\top \cdot X_{ir} m_i'' \right) \right] \\
& = - \E \left[ \tr \left( X_i^\top B^{-1} X_i \right) \cdot X_{ir} m_i'' \right] \\
& = - \E \left[ X_i^\top B^{-1} X_i \cdot X_{ir} m_i'' \right]. 
\end{align*}

\subsection{Proof of \texorpdfstring{Theorem \ref{thm:gc}}{Theorem 1}} \label{eq:append-proof-bias-gc}

\subsubsection{The MLE part} \label{eq:append-proof-bias-gc-mle}

First, we note that $\E ( \hmu_a - \mu_a ) \equiv \E ( \hwm_{i|a} - m_{i|a} )$ since $\E ( \hmu_a ) \equiv \E ( \hwm_{i|a} )$ and $\E ( \mu_a ) \equiv \E ( m_{i|a} )$.  Without loss of generality, we consider the stochastic expansion of $\hwm_{1|a}$, 
\begin{equation}
\hwm_{1|a} - m_{1|a} = m_{1|a}' \cdot X_{1|a}^\top ( \hbbeta - \bbeta_0 ) + \frac{1}{2} m_{1|a}'' \cdot X_{1|a}^\top ( \hbbeta - \bbeta_0 ) ( \hbbeta - \bbeta_0 )^\top X_{1|a} + O_{\sp} ( n^{-3/2} ). \label{eq:proof-mle-bias-wm}
\end{equation}
Using the second-order stochastic expansion of $\hbbeta$, which is \eqref{eq:proof-beta-mle-expand-2rd} provided in Appendix~\ref{append:proof-beta-hoif}, the unconditional expectation of the first term of the RHS in \eqref{eq:proof-mle-bias-wm} is 
\begin{multline*}
\frac{1}{n} \E \left[ m_{1|a}' \cdot X_{1|a}^\top \bpsi_1^{\bbeta} \right] + \frac{1}{n} \SUM{ i \neq 1 }{} \E \left[ m_{1|a}' \cdot X_{1|a}^\top \bpsi_i^{\bbeta} \right] + \underbrace{ \frac{1}{n^2} \E \left[ m_{1|a}' \cdot X_{1|a}^\top \bpsi^{\bbeta,2}_{11} \right] }_{ O ( n^{-2} ) } + \\ 
\frac{1}{n^2} \SUM{ i \neq 1 }{} \E \left[ m_{1|a}' \cdot X_{1|a}^\top \bpsi^{\bbeta,2}_{i1} \right] + \frac{1}{n^2} \SUM{ j \neq 1 }{}\E \left[ m_{1|a}' \cdot X_{1|a}^\top \bpsi^{\bbeta,2}_{1j} \right] + \\
\frac{1}{n^2} \SUM{ i,j \neq 1 }{} \E \left[ m_{1|a}' \cdot X_{1|a}^\top \bpsi^{\bbeta,2}_{ij} \right] + O ( n^{-3/2} ).
\end{multline*}
Using the results provided in \eqref{eq:proof-bias-gc-suppl-1}, the above result can be further simplified, leading to 
\begin{align}
& \ \E \left[ m_{1|a}' \cdot X_{1|a}^\top ( \hbbeta - \bbeta_0 ) \right] \nonumber \\
= & \ \frac{1}{n} \ \E \left[ m_{1|a}' \cdot X_{1|a}^\top \bpsi_1^{\bbeta} \right] + \frac{1}{n^2} \SUM{ i,j \neq 1 }{} \E \left[ m_{1|a}' \cdot X_{1|a}^\top \right] \E [ \bpsi_{ij}^{\bbeta,2} ] + O ( n^{-3/2} ) \nonumber \\
= & \ \frac{1}{n} \ \E \left[ m_{1|a}' \cdot X_{1|a}^\top \bpsi_1^{\bbeta} \right] + \frac{1}{n} \ \E \left[ m_{1|a}' \cdot X_{1|a}^\top \right] \bb_1 ( \hbbeta ) + O ( n^{-3/2} ). \label{eq:proof-mle-bias-wm-1}
\end{align}
The last equality holds since $\E [ \bpsi_{ij}^{\bbeta,2} ] = \bzero_p$ if $i \neq j$ and $\E [ \bpsi_{ii}^{\bbeta,2} ] = \bb_1 ( \hbbeta )$ (see Appendix~\ref{append:proof-beta-hoif}).  Next, denote the $O_{\sp} ( n^{-1} )$ term in \eqref{eq:proof-beta-mle-expand-2rd} as $R_n$. The unconditional expectation of the second term of the right-hand side in \eqref{eq:proof-mle-bias-wm} is 
\begin{multline*}
\underbrace{ \frac{1}{n^2} \ \E \left[ m_{1|a}'' \cdot X_{1|a}^\top \bpsi^{\bbeta}_1 \bpsi^{\bbeta\top}_1 X_{1|a} \right] }_{ O ( n^{-2} ) } + \frac{1}{n^2} \SUM{i=2}{n} \E \left[ m_{1|a}'' \cdot X_{1|a}^\top \bpsi^{\bbeta}_i \bpsi^{\bbeta\top}_i X_{1|a} \right] + \\
\frac{1}{n^2} \SUM{i=2}{n} \E \left[ m_{1|a}'' \cdot X_{1|a}^\top \bpsi^{\bbeta}_i \bpsi^{\bbeta\top}_1 X_{1|a} \right] + \frac{1}{n^2} \SUM{j=2}{n} \E \left[ m_{1|a}'' \cdot X_{1|a}^\top \bpsi^{\bbeta}_1 \bpsi^{\bbeta\top}_j X_{1|a} \right] + \\ 
\frac{1}{n^2} \SUM{ i \neq j \neq 1 }{} \E \left[ m_{1|a}'' \cdot  X_{1|a}^\top \bpsi^{\bbeta}_i \bpsi^{\bbeta\top}_j X_{1|a} \right] + \underbrace{ \E \left[ m_{1|a}'' \cdot X_{1|a}^\top \left\{ \frac{1}{n} \SUM{i=1}{n} \bpsi^{\bbeta}_i \right\} R_n^\top X_{1|a} \right] }_{ O ( n^{-3/2} ) } + \\
\underbrace{ \E \left[ m_{1|a}'' \cdot X_{1|a}^\top R_n \left\{ \frac{1}{n} \SUM{j=1}{n} \bpsi^{\bbeta}_j \right\}^\top X_{1|a} \right] }_{ O ( n^{-3/2} ) }. 
\end{multline*}
Using the results provided in \eqref{eq:proof-bias-gc-suppl-2}, the above result can be further simplified, leading to  
\begin{equation}
\E \left[ m_{1|a}'' \cdot X_{1|a}^\top ( \hbbeta - \bbeta_0 ) ( \hbbeta - \bbeta_0 )^\top X_{1|a} \right] = \frac{1}{n} \ \E \left[ m_{1|a}'' \cdot X_{1|a}^\top B^{-1} M B^{-1} X_{1|a} \right] + O ( n^{-3/2} ). \label{eq:proof-mle-bias-wm-2}
\end{equation}
Lastly, using \eqref{eq:proof-bias-gc-suppl-3} and \eqref{eq:proof-bias-gc-suppl-4}, we have
\begin{align}
& \E \left[ m_{1|a}' \cdot X_{1|a}^\top \right] \bb_1 ( \hbbeta ) = - \E \left[ m_{1|a}' \cdot X_{1|a}^\top \right] B^{-1} ( H_1 + H_2 ) \nonumber \\
= & - \E \left[ m_1' \cdot X_1^\top \bpsi_1^{\bbeta} \middle\vert A_1 = a \right] - \frac{1}{2} \ \E \left[ m_{1|a}' \cdot X_{1|a}^\top B^{-1} M B^{-1} X_{1|a} \right]. \label{eq:proof-mle-bias-wm-3}
\end{align}
Armed with the above results, replacing \eqref{eq:proof-mle-bias-wm-1} -- \eqref{eq:proof-mle-bias-wm-3} in \eqref{eq:proof-mle-bias-wm}, we have 
\begin{align}
b_1 ( \hmu_a ) & = \E \left[ m_{1|a}' \cdot X_{1|a}^\top \bpsi_1^{\bbeta} \right] - \E \left[ m_1' \cdot X_1^\top \bpsi_1^{\bbeta} \middle\vert A_1 = a \right] \nonumber \\
& = ( 1 - \pi_a ) \left\{ \E \left[ m_{1|a}' \cdot X_{1|a}^\top \bpsi_1^{\bbeta} \middle\vert A_1 \neq a \right] - \E \left[ m_{1|a}' \cdot X_{1|a}^\top \bpsi_1^{\bbeta} \middle\vert A_1 = a \right] \right\}. \label{eq:proof-mle-bias-wm-4}
\end{align}
Besides, 
\begin{align*}
\E \left[ m_{1|a}' \cdot X_{1|a}^\top \bpsi_1^{\bbeta} \middle\vert A_1 = a \right] 
& = \E \left[ m_{1|a}' \cdot X_{1|a}^\top B^{-1} X_1 \cdot ( Y_1 - m_1 ) \middle\vert A_1 = a \right] \\
& = \E \left[ m_{1|a}' \cdot X_{1|a}^\top B^{-1} X_{1|a} \cdot \{ Y_1 (a) - m_{1|a} \} \right] \\
\E \left[ m_{1|a}' X_{1|a}^\top \bpsi_1^{\bbeta} \middle\vert A_1 \neq a \right]
& = \SUM{ b \neq a }{} \frac{\pi_b}{1-\pi_a} \ \E \left[ m_{1|a}' \cdot X_{1|a}^\top \bpsi_1^{\bbeta} \middle\vert A_1 = b \right] \\
& = \SUM{ b \neq a }{} \frac{\pi_b}{1-\pi_a} \ \E \left[ m_{1|a}' \cdot X_{1|a}^\top B^{-1} X_1 \cdot ( Y_1 - m_1 ) \middle\vert A_1 = b \right] \\
& = \SUM{ b \neq a }{} \frac{\pi_b}{1-\pi_a} \ \E \left[ m_{1|a}' \cdot X_{1|a}^\top B^{-1} X_{1|b} \cdot \{ Y_1 (b) - m_{1|b} \} \right], 
\end{align*}
which completes the proof for the formula of $b_1 ( \hmu_a )$. 

\subsubsection{The FC part}

We obtain $b_1 (\tmu_a)$ following the steps outlined in Appendix~\ref{eq:append-proof-bias-gc-mle}.  Note that the stochastic expansion of $\tbbeta$, provided in \eqref{eq:proof-beta-fc-expand-2rd} in Appendix~\ref{append:proof-beta-hoif}, differs from that of $\hbbeta$ only in one $O_{\sp} ( n^{-1} )$ term, which is $n^{-1} B^{-1} H_3$ (due to the augmentation term in the Firth's modified score equation).  With this additional $O_{\sp} ( n^{-1} )$ term, \eqref{eq:proof-mle-bias-wm-1} becomes 
\begin{align*}
& \E \left[ m_{1|a}' \cdot X_{1|a}^\top ( \tbbeta - \bbeta_0 ) \right] \\
= & \frac{1}{n} \ \E \left[ m_{1|a}' \cdot X_{1|a}^\top \bpsi_1^{\bbeta} \right] + \frac{1}{n} \ \E \left[ m_{1|a}' \cdot X_{1|a}^\top \right] \bb_1 ( \hbbeta ) + \frac{1}{n} \ \E \left[ m_{1|a}' \cdot X_{1|a}^\top \right] B^{-1} H_3 + O ( n^{-3/2} ) \\
= & \frac{1}{n} \E \left[ m_{1|a}' \cdot X_{1|a}^\top \bpsi_1^{\bbeta} \right] + \frac{1}{n} \ \E \left[ m_{1|a}' \cdot X_{1|a}^\top \right] \bb_1 ( \tbbeta ) + O ( n^{-3/2} ). 
\end{align*}
Furthermore, \eqref{eq:proof-mle-bias-wm-2} remains the same since the additional $O_{\sp} ( n^{-1} )$ term can be absorbed in $R_n$. Therefore, $b_1 (\tmu_a)$ differs from $b_1 (\hmu_a)$ only in one additional term, which is 
\[
\E \left[ m_{1|a}' \cdot X_{1|a}^\top \right] B^{-1} H_3 = \frac{1}{2} \ \E \left[ m_{i|a}'' \cdot X_{i|a}^\top B^{-1} X_{i|a} \right].
\]
The equality holds due to \eqref{eq:proof-bias-fc-suppl}.  This completes the proof for $b_1^{(2)} ( \tmu_a )$. 

\subsection{Proof of \texorpdfstring{Proposition~\ref{prop:bound-gc}}{Proposition 2}}
\label{append:proof-bound-gc}

\subsubsection{The MLE part}

We first derive the upper-bound for $\abs{ b_1^{(1)} ( \hmu_a ) }$ with pooled working models.  Write $\E [ \cdot | A_i = a ]$ by $\E [ \E [ \cdot | W_i, A_i = a ] | A_i = a ]$, and then
\begin{align*}
\abs{ b_1^{(1)} ( \hmu_a ) } & = ( 1 - \pi_a ) \cdot \left\lvert \E \left[ m_i' \cdot X_i^\top B^{-1} X_i \cdot ( r_{i|a} - m_{i|a} ) \middle\vert A_i = a \right] \right\rvert \\
& \leq ( 1 - \pi_a ) \cdot  \E \left[ m_i' \cdot X_i^{\top} B^{-1} X_i \middle\vert A_i = a \right] \cdot \norm{ r_{i|a} - m_{i|a} }_{\infty}.
\end{align*}
The last inequality holds by the H\"{o}lder's inequality and the non-negativity of $m_i' \cdot X_i^{\top} B^{-1} X_i$ (we assumed $m$ is non-decreasing which holds for most of GLMs encountered in practice). Besides, we have that 
\[
\E \left[ m_i' \cdot X_i^{\top} B^{-1} X_i \middle\vert A_i = a \right] = c_a / \pi_a \cdot \E \left[ m_i' \cdot X_i^{\top} B^{-1} X_i \right],
\]
where 
\[
c_a \coloneq \frac{ \pi_a \cdot \E \left[ m_i' \cdot X_i^\top B^{-1} X_i \middle\vert A_i = a \right] }{ \E \left[ m_i' \cdot X_i^\top B^{-1} X_i \right] } < 1,
\]
and 
\[
\E \left[ m_i' \cdot X_i^\top B^{-1} X_i \right] = \lim_{n\to\infty} \SUM{i=1}{n} h_{ii} = p,
\]
since $\SUM{i}{} h_{ii} = p$.  Above all, we complete the proof for $\abs{ b_1^{(1)} ( \hmu_a ) }$.

Next, we derive the upper-bound for $\abs{ b_1^{(2)} ( \hmu_a ) }$ with pooled working models.  Similar to $\abs{ b_1^{(1)} ( \hmu_a ) }$, we have that 
\begin{align*}
\abs{ b_1^{(2)} ( \hmu_a ) } & = ( 1 - \pi_a ) \cdot \left\lvert \E \left[ m_{i|a}' \cdot X_{i|a}^\top B^{-1} X_i \cdot ( \E [ Y_i | W_i, A_i ] - m_i ) \middle\vert A_i \neq a \right] \right\rvert \\
& \leq ( 1 - \pi_a ) \cdot \E \left[ \abs{ m_{i|a}' \cdot X_{i|a}^{\top} B^{-1} X_i } \middle\vert A_i \neq a \right] \cdot \sup_{ b \neq a }\norm{ r_{i|b} - m_{i|b} }_{\infty} \\ 
& \leq c_a^\ast \cdot ( 1 - \pi_a ) \cdot \E \left[ m_i' \cdot X_i^{\top} B^{-1} X_i \middle\vert A_i \neq a \right] \cdot \sup_{ b \neq a }\norm{ r_{i|b} - m_{i|b} }_{\infty} \\
& = c_a^\ast c_a^{\ast\ast} \cdot \E \left[ m_i' \cdot X_i^{\top} B^{-1} X_i \right] \cdot \sup_{ b \neq a }\norm{ r_{i|b} - m_{i|b} }_{\infty}
\end{align*}
where 
\begin{align*}
c_a^\ast &  \coloneq \sup_{ b \neq a} \frac{ \E \left[ \abs{ m_{i|a}' \cdot X_{i|a}^{\top} B^{-1} X_{i|b} } \right] }{ \E \left[  m_{i|b}' \cdot X_{i|b}^{\top} B^{-1} X_{i|b} \right] } > 0 \\
c_a^{\ast\ast} & \coloneq \frac{ (1-\pi_a) \cdot \E \left[ m_i' \cdot X_i^\top B^{-1} X_i \middle\vert A_i \neq a \right] }{ \E \left[ m_i' \cdot X_i^\top B^{-1} X_i \right] } < 1.
\end{align*}
We complete the proof for $\abs{ b_1^{(2)} ( \hmu_a ) }$ since $\E [ m_i' \cdot X_i^\top B^{-1} X_i ] = p$. 

Finally, for stratified working models, we have that $\SUM{ i \colon A_i = a }{} h_{ii} = p_a$, which completes the proof immediately.

\subsubsection{The FC part}

Similar to $\abs{ b_1^{(1)} ( \hmu_a ) }$, for pooled working models we have that 
\begin{align*}
\abs{ b_1^{(2)} ( \tmu_a ) } 
& = \frac{1}{2} \ \E \left[ m_i' \cdot X_i^\top B^{-1} X_i \cdot m_i'' / m_i' \middle \vert A_i = a \right] \\
& \leq \frac{1}{2} \ \E \left[ m_i' \cdot X_i^{\top} B^{-1} X_i \middle\vert A_i = a \right] \cdot \norm{ m_i'' / m_i' }_{\infty} \\
& = c_a / (2\pi_a) \cdot \E \left[ m_i' \cdot X_i^\top B^{-1} X_i \right] \cdot \norm{ m_i'' / m_i' }_{\infty},
\end{align*}
 which completes the proof since $\E [ m_i' \cdot X_i^\top B^{-1} X_i ] = p$.  For stratified working models, we have that $\SUM{ i \colon A_i = a }{} h_{ii} = p_a$, which completes the proof immediately.

\subsection{Proof of \texorpdfstring{Lemma \ref{lem:gc}}{Lemma 2}} \label{append:proof-gc-bias-conditional}

With the prediction unbiasedness \eqref{eq:predict-mle}, we have that 
\begin{align*}
\E [ \hwm_{i|a} - m_{i|a} ] 
& = \frac{1}{n} \, \E \left\{ \SUM{ i: A_i = a }{} Y_i + \SUM{ i: A_i \neq a }{} \hwm_{i|a} \right\} - \E [ m_{i|a} ] \\
& = \E [ I ( A_i = a ) Y_i ] + \E [ I ( A_i \neq a ) \hwm_{i|a} ] - \mu_a \\
& = ( 1 - \pi_a ) \cdot \E [ \hwm_{i|a} | A_i \neq a ] - ( 1 - \pi_a ) \mu_a \\ 
& = ( 1 - \pi_a ) \cdot \E [ \hwm_{i|a} - m_{i|a} | A_i \neq a ],
\end{align*}
which completes the proof. 

\subsection{Proof of \texorpdfstring{Theorem \ref{thm:gob-mle}}{Theorem 2}}
\label{append:proof-gob-mle}

We first show that $\E [ \hmu_a^2 - \mu_a ] = O ( n^{-3/2} )$ for pooled working models.  Our proof immediately suggests that $\E [ \hmu_a^1 - \mu_a ] = n^{-1} b_1^{(2)} ( \hmu_a ) + O ( n^{-3/2} )$ for pooled working models, and $O ( n^{-3/2} )$ for stratified working models.  

To start with, we have that 
\begin{align*}
\E [ \hmu_a^2 - \mu_a ] & = \frac{1}{n} \ \E \left\{ \SUM{ i: A_i = a }{} Y_i + \SUM{ i: A_i \neq a }{} m ( X_{i|a}^\top \hbbeta_i^2 ) \right\} - \mu_a \\
& = \E [ I ( A_i = a ) Y_i ] + \E [ I ( A_i \neq a ) m ( X_{i|a}^\top \hbbeta_i^2 ) ] - \mu_a \\
& = \pi_a \mu_a + ( 1 - \pi_a ) \cdot \E \left[ m ( X_{i|a}^\top \hbbeta_i^2 ) \middle\vert A_i \neq a \right] - \mu_a \\
& = ( 1 - \pi_a ) \left\{ \E \left[ m ( X_{i|a}^\top \hbbeta_i^2 ) \middle\vert A_i \neq a \right] - \mu_a \right\} \\
& = ( 1 - \pi_a ) \cdot \E \left[ m ( X_{i|a}^\top \hbbeta_i^2 ) - m ( X_{i|a}^\top \bbeta_{0} ) \middle\vert A_i \neq a \right].
\end{align*}
The last equality in the above holds since $\mu_a = \E [ m_{i|a} ] = \E [ m_{i|a} | A_i \neq a ]$.  Then write 
\begin{equation}
\label{eq:proof-gob-mle-bias-1}
\begin{split}
& \E \left[ m ( X_{i|a}^\top \hbbeta_i^2 ) - m ( X_{i|a}^\top \bbeta_{0} ) \middle\vert A_i \neq a \right] \\
= & \ \E \left[ m ( X_{i|a}^\top \hbbeta_i^2 ) - \hwm_{i|a} \middle\vert A_i \neq a \right] + \E \left[ \hwm_{i|a} - m_{i|a} \middle\vert A_i \neq a \right].  
\end{split}
\end{equation}
For the first term in the RHS of \eqref{eq:proof-gob-mle-bias-1}, without loss of generality, we consider the stochastic expansion of $m ( X_{1|a}^\top \hbbeta_1^2 )$ at $\hbbeta$, 
\begin{multline*}
\hwm_{1|a} + \hwm_{1|a}' \cdot X_{1|a}^\top ( \hbbeta_1^2 - \hbbeta ) + \frac{1}{2} \, \hwm_{1|a}'' \cdot X_{1|a}^\top ( \hbbeta_1^2 - \hbbeta ) ( \hbbeta_1^2 - \hbbeta )^\top X_{1|a} + \\ O_{\sp} \left( \left\{ X_{1|a}^\top ( \hbbeta_1^2 - \hbbeta ) \right\}^3 \right) = \hwm_{1|a} + \frac{1}{n} \, \hwm_{1|a}' \cdot X_{1|a}^\top \SUM{i=1}{n} \hh_{ii} \hbpsi_i^{\bbeta} - \frac{1}{n} \, \hwm_{1|a}' \cdot X_{1|a}^\top \hbpsi^{\bbeta}_1 + O_\sp ( n^{-2} ),
\end{multline*}
since $\hbbeta_1^2 - \hbbeta = n^{-1} \SUM{i}{} \hh_{ii} \hbpsi^{\bbeta}_i - n^{-1} \hbpsi^{\bbeta}_1 = O_\sp ( n^{-1} )$.  Then, we have that  
\begin{align}
& \E \left[ m ( X_{1|a}^\top \hbbeta_1^2 ) - \hwm_{1|a} \middle\vert A_1 \neq a \right] \nonumber \\
= & \ \frac{1}{n} \, \E \left[ \hwm_{1|a}' \cdot X_{1|a}^\top \SUM{i=1}{n} \hh_{ii} \hbpsi_i^{\bbeta} \middle\vert A_1 \neq a \right] - \frac{1}{n} \, \E \left[ \hwm_{1|a}' \cdot X_{1|a}^\top \hbpsi^{\bbeta}_1 \middle\vert A_1 \neq a \right] + O ( n^{-2} ) \nonumber \\
= & \ \frac{1}{n} \, \E \left[ \left\{ m_{1|a}' + O_{\sp} ( n^{-1/2} ) \right\} \cdot X_{1|a}^\top \underbrace{ \left\{ B^{-1} H_1 + O_{\sp} ( n^{-1/2} ) \right\} }_{ \eqref{eq:proof-gob-mle-suppl-2} } \middle\vert A_1 \neq a \right] - \nonumber \\
& \quad \frac{1}{n} \, \E \left[ \left\{ m_{1|a}' + O_{\sp} ( n^{-1/2} ) \right\} \cdot X_{1|a}^\top \underbrace{ \left\{ B^{-1} X_1 ( Y_1 - m_1 ) + O_{\sp} ( n^{-1/2} ) \right\} }_{ \eqref{eq:proof-gob-mle-suppl-1} } \middle\vert A_1 \neq a \right] + O ( n^{-2} ) \nonumber \\ 
= & \ \frac{1}{n} \, \E \left[ m_{1|a}' \cdot X_{1|a}^\top \middle\vert A_1 \neq a \right] B^{-1} H_1 - \frac{1}{n} \, \E \left[ m_{1|a}' \cdot X_{1|a}^\top \bpsi_1^{\bbeta} \middle\vert A_1 \neq a \right] + O ( n^{-3/2} ) \nonumber \\ 
= & \ \frac{1}{n} \, \E \left[ m_{1|a}' \cdot X_{1|a}^\top \right] B^{-1} H_1 - \frac{1}{n} \, \E \left[ m_{1|a}' \cdot X_{1|a}^\top \bpsi_1^{\bbeta} \middle\vert A_1 \neq a \right] + O ( n^{-3/2} ) \nonumber \\ 
= & \ \frac{1}{n} \, \underbrace{ \E \left[ m_{1|a}' \cdot X_{1|a}^\top \bpsi_1^{\bbeta} \middle\vert A_1 = a \right] }_{ \eqref{eq:proof-bias-gc-suppl-3} } - \frac{1}{n} \, \E \left[ m_{1|a}' \cdot X_{1|a}^\top \bpsi_1^{\bbeta} \middle\vert A_1 \neq a \right] + O ( n^{-3/2} ) \nonumber \\ 
& = - \frac{ b_1 ( \hmu_a ) }{ ( 1 - \pi_a ) n } + O ( n^{-3/2} ).  \label{eq:proof-gob-mle-bias-2}
\end{align}
The last equality holds due to \eqref{eq:proof-mle-bias-wm-4}.  For the second term in the RHS of \eqref{eq:proof-gob-mle-bias-1}, we have that 
\begin{equation}
\E [ \hwm_{i|a} - m_{i|a} | A_i \neq a ] = \underbrace{ ( 1 - \pi_a )^{-1} \cdot \E [ \hwm_{i|a} - m_{i|a} ] }_{ \text{Lemma~\ref{lem:gc}} } = \ \frac{ b_1 ( \hmu_a ) }{ ( 1 - \pi_a ) n } + O ( n^{-3/2} ).  \label{eq:proof-gob-mle-bias-3}  
\end{equation}
Above all, with \eqref{eq:proof-gob-mle-bias-1}--\eqref{eq:proof-gob-mle-bias-2} we have that 
\[
\E [ \hmu_a^2 - \mu_a ] = ( 1 - \pi_a ) \cdot \E [ m ( X_{i|a}^\top \hbbeta_i^2 ) - m_{i|a} | A_i \neq a ] = O ( n^{-3/2} ),
\]
which completes the proof for pooled working models.  For $\E [ \hmu_a^1 - \mu_a ]$, following the same manner, we have that
\begin{align*}
\E [ \hmu_a^1 - \mu_a ]
& = ( 1 - \pi_a ) \cdot \E [ m ( X_{i|a}^\top \hbbeta^1 ) - m_{i|a} | A_i \neq a ] \\ 
& = ( 1 - \pi_a ) \cdot \E [ m ( X_{i|a}^\top \hbbeta^1 ) - \hwm_{i|a} | A_i \neq a ] + ( 1 - \pi_a ) \cdot \E [ \hwm_{i|a} - m_{i|a} | A_i \neq a ] \\ 
& = ( 1 - \pi_a ) \cdot \left\{ - \frac{ b_1^{(1)} ( \hmu_a ) }{ ( 1 - \pi_a ) n } \right\} + ( 1 - \pi_a ) \cdot \frac{ b_1 ( \hmu_a ) }{ ( 1 - \pi_a ) n } = n^{-1} b_1^{(2)} ( \hmu_a ),
\end{align*}
which completes the proof ($b_1^{(2)} ( \hmu_a ) \equiv 0$ for stratified working models). 

\subsection{Proof of \texorpdfstring{Theorem \ref{thm:gob-fc}}{Theorem 3}}
\label{append:proof-gob-fc}

Similar to Appendix~\ref{append:proof-gob-mle}, we first show that $\E [ \tmu_a^2 - \mu_a ] = O ( n^{-3/2} )$ for pooled working models.  Our proof immediately suggests that $\E [ \tmu_a^1 - \mu_a ] = n^{-1} b_1^{(2)} ( \hmu_a ) + O ( n^{-3/2} )$ for pooled working models and $= O ( n^{-3/2} )$ for stratified working models, while $\E [ \tmu_a^0 - \mu_a ] = n^{-1} b_1 ( \hmu_a ) + O ( n^{-3/2} )$.  Analogous to $\hmu_a^2$, we have that
\begin{align}
\E [ \tmu_a^2 - \mu_a ]
& = ( 1 - \pi_a ) \cdot \E [ m ( X_{i|a}^\top \tbbeta_i^2 ) - m_{i|a} | A_i \neq a ] \nonumber \\ 
& = ( 1 - \pi_a ) \left\{ \E [ m ( X_{i|a}^\top \tbbeta_i^2 ) - \twm_{i|a} | A_i \neq a ]  + \E [ \twm_{i|a} - m_{i|a} | A_i \neq a ] \right\} \label{eq:proof-gob-fc-bias-1}  
\end{align}
For the first term in the RHS of \eqref{eq:proof-gob-fc-bias-1}, analogous to \eqref{eq:proof-gob-mle-bias-2}, we have
\begin{align}
& \E [ m ( X_{1|a}^\top \tbbeta_1^2 ) - \twm_{1|a} | A_1 \neq a ] \nonumber \\
& = \frac{1}{n} \, \E \left[ \twm_{1|a}' \cdot X_{1|a}^\top \SUM{i=1}{n} \th_{ii} \tbpsi_i^{\bbeta} \middle\vert A_1 \neq a \right] - \frac{1}{n} \, \E \left[ \twm_{1|a}' \cdot X_{1|a}^\top \tbpsi^{\bbeta}_1 \middle\vert A_1 \neq a \right] - \nonumber \\
& \qquad \frac{1}{n} \, \E \left[ \twm_{1|a}' \cdot X_{1|a}^\top \cdot \frac{1}{2} \SUM{i=1}{n} \th_{ii} \cdot \frac{\twm_i''}{\twm_i'} \cdot \tB^{-1} X_i \middle\vert A_1 \neq a  \right] + O ( n^{-2} ) \nonumber \\
& = \underbrace{ - \frac{ b_1 ( \hmu_a ) }{ ( 1 - \pi_a ) n } }_{ \text{ analogous to \eqref{eq:proof-gob-mle-bias-2} } } - \frac{1}{n} \, \E \left[ \{ m_{1|a}' + O_{\sp} ( n^{-1/2} ) \} \cdot X_{1|a}^\top \underbrace{ \{ B^{-1} H_3 + O_{\sp} ( n^{-1/2} ) \} }_{ \eqref{eq:proof-gob-mle-suppl-3} } \middle\vert A_1 \neq a \right] \nonumber \\
& \qquad + O \left( n^{-3/2} \right) \nonumber \\
& = - \frac{ b_1 ( \hmu_a ) }{ ( 1 - \pi_a ) n } - \frac{1}{n} \, \E [ m_{1|a}' \cdot X_{1|a}^\top | A_1 \neq a ] B^{-1} H_3 + O ( n^{-3/2} ) \nonumber \\ 
& = \ - \frac{ b_1 ( \hmu_a ) }{ ( 1 - \pi_a ) n } - \frac{1}{n} \, \E [ m_{1|a}' \cdot X_{1|a}^\top ] B^{-1} H_3 + O ( n^{-3/2} ). \label{eq:proof-gob-fc-bias-2}
\end{align}
For the second term in the RHS of \eqref{eq:proof-gob-fc-bias-1}, we have that 
\begin{align}
    \E [ \twm_{i|a} - m_{i|a} | A_i \neq a ]
    & = \frac{ b_1 ( \tmu_a ) }{ ( 1 - \pi_a ) n }  - \frac{ \pi_a b_1^{(2)} ( \tmu_a ) }{ ( 1 - \pi_a ) n } + O ( n^{-3/2} ) \nonumber \\ 
    & = \frac{ b_1 ( \hmu_a ) }{ ( 1 - \pi_a ) n } + \frac{1}{n} b_1^{(2)} ( \tmu_a ) + O ( n^{-3/2} ).  \label{eq:proof-gob-fc-bias-3}  
\end{align}
The first equality holds because that, with Lemma~\ref{lem:predict-fc} we have that 
\begin{align*}
& \E [ \twm_{i|a} - m_{i|a} ] \\
= & \ \frac{1}{n} \E \left( \frac{1}{2} \SUM{ i \colon A_i = a }{} \th_{ii} \cdot \frac{\twm_i''}{\twm_i'} + \SUM{ i: A_i = a }{} Y_i + \SUM{ i: A_i \neq a }{} \twm_{i|a} \right) - \E [ m_{i|a} ] \\
= & \ ( 1 - \pi_a ) \cdot \E [ \twm_{i|a} - m_{i|a} | A_i \neq a ] + \frac{1}{n} \, \E \left[ \frac{1}{2} \SUM{i=1}{n} I ( A_i = a ) \th_{ii} \cdot \frac{\twm_i''}{\twm_i'} \right] \nonumber \\
= & \ ( 1 - \pi_a ) \cdot \E [ \twm_{i|a} - m_{i|a} | A_i \neq a ] + \frac{1}{2n} \, \E \left[ \E [ I ( A_i = a ) m_i'' \cdot X_i^\top B^{-1} X_i ] + O_{\sp} ( n^{-1/2} ) \right] \\ 
= & \ ( 1 - \pi_a ) \cdot \E [ \twm_{i|a} - m_{i|a} | A_i \neq a ] + \frac{\pi_a}{2n} \, \E [ m_{i|a}'' \cdot X_{i|a}^\top B^{-1} X_{i|a} ] + O ( n^{-3/2} ).
\end{align*}
Above all, with \eqref{eq:proof-gob-fc-bias-1}--\eqref{eq:proof-gob-fc-bias-3} we have that 
\begin{align*}
\E [ \tmu_a^2 - \mu_a ]
& = \frac{ 1 - \pi_a }{n} \underbrace{ \left( \frac{1}{2} \E [ m_{i|a}'' \cdot X_{i|a}^\top B^{-1} X_{i|a} ] - \E [ m_{i|a}' X_{i|a}^\top ] B^{-1} H_3 \right) }_{ = 0 \ \eqref{eq:proof-bias-fc-suppl} } + O ( n^{-3/2} ) \\
& = O ( n^{-3/2} ),
\end{align*}
which completes the proof for pooled working models.  For $\E [ \tmu_a^1 - \mu_a ]$ and $\E [ \tmu_a^0 - \mu_a ]$, the proof is trivial following the same manner as above (and analogous to $\hmu_a^1$).

\subsection{Proof of \texorpdfstring{Theorem \ref{thm:linear}}{Theorem 4}} \label{append:proof-linear}

\subsubsection{The MLE part}

First of all, from $\hbbeta^1$ in \eqref{eq:gob-mle-bc} we have that 
\begin{align*}
\hbbeta^1 - \hbbeta 
& = \frac{1}{n} \SUM{i=1}{n} \hh_{ii} \hbpsi_i^{\bbeta} = \frac{1}{n} \SUM{i=1}{n} \frac{1}{n} \hwm_i' \cdot X_i^\top \hB^{-1} X_i\cdot \underbrace{ \left\{ \bpsi_i^{\bbeta} + O_{\sp} ( n^{-1/2} ) \right\} }_{\eqref{eq:proof-gob-mle-suppl-1}} \\ 
& = \frac{1}{n} \SUM{i=1}{n} \frac{1}{n} \{ m_i' + O_{\sp} ( n^{-1/2} ) \} \cdot X_i^\top  \underbrace{ \{ B^{-1} + O_\sp ( n^{-1/2} ) \} }_{\eqref{eq:proof-beta-bread-est-inv}} X_i \cdot \bpsi_i^{\bbeta} + o_{\sp} ( n^{-1} )\\
& = \frac{1}{n} \SUM{i=1}{n} \barh_{ii} \bpsi_i^{\bbeta} + o_{\sp} ( n^{-1} ). 
\end{align*}
Combining the above result with the fact that $\hbbeta - \bbeta_0 = n^{-1} \SUM{i}{} \bpsi_i^{\bbeta} + o_{\sp} ( n^{-1/2} )$, we have that 
\begin{equation}
\hbbeta^1 - \bbeta_0 = \frac{1}{n} \SUM{i=1}{n} ( 1 + \barh_{ii} ) \cdot \bpsi_i^{\bbeta} + o_{\sp} ( n^{-1/2} ). \label{eq:proof-linear-mle}
\end{equation}
Let $\hwm_{i|a}^1 \coloneqq m ( X_{i|a}^\top \hbbeta^1 )$.  We obtain the linearization of $\hmu_a^1$ as follows,
\begin{align*}
\hmu_a^1 - \mu_a
& = \frac{1}{n} \SUM{i=1}{n} \hwm_{i|a}^1 + \frac{1}{n} \SUM{ i \colon A_i = a }{} ( Y_i - \hwm_{i|a}^1 ) - \mu_a \\
& = \frac{1}{n} \SUM{i=1}{n} \left\{ m_{i|a} + m_{i|a}^1 \cdot X_{i|a}^\top ( \hbbeta^1 - \bbeta_0 ) + o_{\sp} ( n^{-1/2} ) \right\} + \frac{1}{n} \underbrace{ \SUM{ i \colon A_i = a }{} ( \hwm_{i|a} - \hwm_{i|a}^1 ) }_{ \text{\eqref{eq:predict-mle}} } - \mu_a  \\
& = \E [ m_{i|a}' \cdot X_{i|a}^\top ] ( \hbbeta^1 - \bbeta_0 ) + \underbrace{ \left\{ \frac{1}{n} \SUM{i=1}{n} m_{i|a}' \cdot X_{i|a}^\top - \E [ m_{i|a}' \cdot X_{i|a}^\top ] \right\} }_{ O_{\sp} ( n^{-1/2} ) } \underbrace{ ( \hbbeta^1 - \bbeta_0 ) }_{ O_{\sp} ( n^{-1/2} ) } + \\
& \qquad\qquad \frac{1}{n} \SUM{i=1}{n} ( m_{i|a} - \mu_a ) + \frac{1}{n} \SUM{ i \colon A_i = a }{} \underbrace{ \left\{ \hwm_{i|a}' \cdot X_{i|a}^\top ( \hbbeta^1 - \hbbeta ) + o_{\sp} (n^{-1}) \right\} }_{ O_{\sp} ( n^{-1} ) + o_{\sp} (n^{-1}) = O_{\sp} (n^{-1})} + o_{\sp} (n^{-1/2}) \\
& = \E [ m_{i|a}' \cdot X_{i|a}^\top ] \underbrace{ \SUM{i=1}{n} \frac{1}{n} ( 1 + \barh_{ii} ) \cdot \bpsi_i^{\bbeta} }_{ \text{\eqref{eq:proof-linear-mle}} } + \frac{1}{n} \SUM{i=1}{n} ( m_{i|a} - \mu_a ) + o_{\sp} ( n^{-1/2} ) \\
& = \ \frac{1}{n} \SUM{i=1}{n} \left\{ \frac{I(A_i = a)}{\pi_a} ( Y_i - m_i ) \cdot ( 1 + \barh_{ii} ) + m_{i|a} - \mu_a \right\} + o_{\sp} ( n^{-1/2} ),
\end{align*}
which completes the proof for $\hmu_a^1$ (the last equation holds due to \eqref{eq:zhang}).

\subsubsection{The FC part}

Analogously, we have that $n^{-1} \SUM{i}{} \th_{ii} \tbpsi_i^{\bbeta} = n^{-1} \SUM{i}{} \barh_{ii} \bpsi_i^{\bbeta} + o_{\sp} ( n^{-1} )$, and from \eqref{eq:gob-fc-bc} we have that
\begin{align*}
\tbbeta^1 - \tbbeta 
& = \frac{1}{n} \SUM{i=1}{n} \th_{ii} \tbpsi_i^{\bbeta} - \frac{1}{2n} \SUM{i=1}{n} \th_{ii} \cdot \frac{\twm_i''}{\twm_i'} \cdot \tB^{-1} X_i \\ 
& = \frac{1}{n} \SUM{i=1}{n} \barh_{ii} \bpsi_i^{\bbeta} - \frac{1}{2n} \SUM{i=1}{n} \th_{ii} \cdot \frac{\twm_i''}{\twm_i'} \cdot B^{-1} X_i + o_{\sp} ( n^{-1} ). 
\end{align*}
since $\tB^{-1} - B^{-1} = O_\sp ( n^{-1/2} )$ (analogous to \eqref{eq:proof-beta-bread-est-inv}). 

In $\tbbeta - \bbeta_0 = n^{-1} \SUM{i=1}{n} \bpsi_i^{\bbeta} + o_{\sp} ( n^{-1/2} )$, $o_{\sp} (n^{-1/2} )$ absorbs the first-order term due to the augmentation in \eqref{eq:score-fc}, which is $\Delta^{(n)} (\tbbeta)$.  Keeping this term in the stochastic expansion, we have that 
\[
\tbbeta - \bbeta_0 = \frac{1}{n} \SUM{i=1}{n} \bpsi_i^{\bbeta} + \frac{1}{2n} \SUM{i=1}{n}  \th_{ii} \cdot \frac{\twm_i''}{\twm_i'}\cdot B^{-1} X_i + o_{\sp} ( n^{-1/2} ).
\]

Combining the above two results, we further have that 
\begin{equation}
\tbbeta^1 - \bbeta_0 = \frac{1}{n} \SUM{i=1}{n} ( 1 + \barh_{ii} ) \cdot \bpsi_i^{\bbeta} + o_{\sp} ( n^{-1/2} ). \label{eq:proof-linear-fc}
\end{equation}
Then, we can obtain the linearization of $\tmu_a^1$ following the same steps for $\hmu_a^1$, except using \eqref{eq:proof-linear-fc} instead of \eqref{eq:proof-linear-mle}, which completes the proof.

\section{Stratified Working Models} \label{append:wm-strat}

\subsection{Maximum likelihood estimators (MLE)} \label{append:wm-strat-mle}

Let $m ( Z_{i|a}^\top \bbeta[a] )$ be the working model for $\mu_a$, where $Z_{i|a}$ starts with an intercept following with those variables defined by $W_i$.  The nuisance parameter, $\bbeta[a]$, is estimated by MLE using only those $D_i$s with $A_i = a$.  \citet[][Supporting Information Section A]{zhang2025robust} have shown that there exists a single working model, written as $m ( X_i^\top \bbeta )$, where  
\[
\bbeta = ( \bbeta[1]^\top \quad \cdots \quad \bbeta[k]^\top )^\top \text{ and } 
X_i = ( I ( A_i = 1 ) Z_{i|1}^\top \quad \cdots \quad I ( A_i = k ) Z_{i|k}^\top )^\top,
\]
such that $X_{i|a}^\top \bbeta \equiv Z_{i|a}^\top \bbeta[a]$ and thus $m ( X_i^\top \bbeta ) \equiv m ( Z_{i|a}^\top \bbeta[a] )$.  Subsequently, all the results for pooled working models are applicable for those with stratified working models.  The corresponding formulae for IFs and variance estimation are given as follows.   

Let $\hbbeta[a]$ be the solution of the score equation $\SUM{i \colon A_i = a}{} U_{i|a} ( \bbeta[a] ) = 0$, where $U_{i|a} ( \bbeta[a] ) \coloneqq \{ Y_i - m ( Z_{i|a}^\top \bbeta[a] ) \} Z_{i|a}$.  The score equation of $\hbbeta$ can be constructed using those for all $\hbbeta[a]$s,   
\begin{align*}
\begin{pmatrix}
\SUM{i \colon A_i = 1}{} U_{i|1} ( \bbeta[1] ) \\
\vdots \\
\SUM{i \colon A_i = k}{} U_{i|k} ( \bbeta[k] )
\end{pmatrix}
& = \SUM{i=1}{n} 
\begin{pmatrix}
I ( A_i = 1 ) \{ Y_i - m ( Z_{i|1}^\top \bbeta[1] ) \} Z_{i|1} \\
\vdots \\
I ( A_i = k ) \{ Y_i - m ( Z_{i|k}^\top \bbeta[k] ) \} Z_{i|k}
\end{pmatrix} \\
& = \SUM{i=1}{n} \{ Y_i - m ( X_i^\top \bbeta ) \} X_i = \ \SUM{i=1}{n} U_i ( \bbeta ).
\end{align*}
Then, applying either \eqref{eq:if-mu-mest} or \eqref{eq:if-mu-aipw} from Section~\ref{sec:pre-gc}, we can obtain an empirical version of the IF for $\hmu_a$ and the corresponding variance estimator. 

Specifically, the empirical version of the \textit{theoretical} IF \eqref{eq:if-mu-aipw} is simply written as
\[
\frac{ I ( A_i = a ) }{\hpi_a} ( Y_i - \hwm_{i|a} ) + \hwm_{i|a} - \hmu_a.
\]
For the one of the \textit{empirical} IF \eqref{eq:if-mu-mest}, we first write the empirical version of the IF for $\hbbeta$, which is $\hbpsi_{\bbeta} = \hB^{-1} \hU_i$, in terms of those quantities associated with $\hU_{i|a}$.  Since $\hbbeta$ is a stacked vector of $\hbbeta[a]$ for all $a$, we consider the subvector in $\hbpsi_{\bbeta}$ associated with $\hbbeta[a]$, written as
\begin{align}
& I ( A_i = a ) \left\{ n^{-1} \SUM{i=1}{n} I ( A_i = a ) \hwm_{i|a} Z_{i|a} Z_{i|a}^\top \right\}^{-1} Z_{i|a} \{ Y_i - \hwm_{i|a} \} \nonumber \\
= & \frac{ I ( A_i = a ) }{ \hpi_a } \underbrace{ \hB_a^{-1} Z_{i|a} \{ Y_i - \hwm_{i|a} \} }_{ \eqcolon \; \hbpsi_i^{ \bbeta[a] } }, \label{eq:proof-hwm-if-beta}
\end{align}
where $\hB_a = n_a^{-1} \SUM{ i \colon A_i = a }{} \hwm_{i|a} Z_{i|a} Z_{i|a}^\top$ and $\hbpsi_i^{ \bbeta[a] }$ presents the empirical versions of the bread matrix and the IF for $\hbbeta[a]$, respectively.  Then, the estimate of the \textit{empirical} IF of $\hmu_a$ is 
\begin{align*}
& \left\{ \frac{1}{n} \SUM{j=1}{n} \hwm_{j|a}' 
\begin{pmatrix}
\bzero_{a-1} \\ 
Z_{j|a} \\
\bzero_{k-a}
\end{pmatrix}^\top \right\} \hbpsi_i^{\bbeta} + \hwm_{i|a} - \hmu_a \\
= & \frac{ I ( A_i = a ) }{ \hpi_a } \left\{ \frac{1}{n} \SUM{j=1}{n} \hwm_{j|a}' Z_{j|a}^\top \right\} \hbpsi_i^{ \bbeta[a] } + \hwm_{i|a} - \hmu_a,
\end{align*}
of which all components can be directly obtained from the output of fitting a GLM using off-the-shelf software packages.

\subsection{Firth-corrected (FC) estimators} \label{append:wm-strat-fc}

Let $\tbbeta[a]$ be the solution of the modified score equation associated with the working model for $A_i = a$, written as 
\begin{align*}
\SUM{ i \colon A_i = a }{} U_{i|a} \left( \bbeta[a] \right) + \frac{1}{2} \SUM{ i \colon A_i = a }{} h_{ii|a} ( \bbeta[a] ) \cdot \frac{ m'' ( Z_{i|a}^\top \bbeta(a) ) }{ m' ( Z_{i|a}^\top \bbeta(a) ) } \cdot Z_{i|a} = \bzero_p
\end{align*}
where 
\[
h_{ii|a} ( \bbeta[a] ) \coloneqq \ m' ( Z_{i|a}^\top \bbeta[a] ) \cdot Z_{i|a}^\top \left\{ \SUM{ j \colon A_j = a }{} m' ( Z_{j|a}^\top \bbeta[a] ) \cdot Z_{j|a} Z_{j|a}^\top \right\}^{-1} Z_{i|a}.
\]
Now, we show that the stacked modified score equation for $a=1,\dots,k$ is equivalent to the stacked modified score equation of $\tbbeta$ defined in \eqref{eq:score-fc} in the main manuscript.  For leverage scores, we have that  $h_{ii} ( \bbeta ) \equiv h_{ii|a} ( \bbeta[a] )$ when $A_i = a$ since 
\begin{align}
h_{ii} ( \bbeta ) & = m' ( X_i^\top \bbeta ) \cdot Z_{i|a}^\top \left\{ \SUM{j=1}{n} I ( A_j = a ) m' ( X_j^\top \bbeta ) \cdot Z_{j|a} Z_{j|a}^\top \right\}^{-1} Z_{i|a} \nonumber \\
& = m' ( Z_{i|a}^\top \bbeta[a] ) \cdot Z_{i|a}^\top \left\{ \SUM{j \colon A_j = a}{} m' ( Z_{j|a}^\top \bbeta[a] ) \cdot Z_{j|a} Z_{j|a}^\top \right\}^{-1} Z_{i|a}. \label{eq:proof-hwm-leverage}
\end{align}
Then, we have that 
\begin{align*}
&
\begin{pmatrix}
\SUM{i \colon A_i = 1}{} h_{ii|1} ( \bbeta[1] ) \cdot m'' ( Z_{i|1}^\top \bbeta[1] ) / m' ( Z_{i|1}^\top \bbeta[1] ) \cdot Z_{i|1} / 2 \\
\vdots \\
\SUM{i \colon A_i = k}{} h_{ii|k} ( \bbeta[k] ) \cdot m'' ( Z_{i|k}^\top \bbeta[k] ) / m' ( Z_{i|k}^\top \bbeta[k] ) \cdot Z_{i|k} / 2 
\end{pmatrix} \\
= & \ \frac{1}{2} \SUM{i=1}{n} h_{ii} ( \bbeta ) \cdot \frac{ m'' ( X_i^\top \bbeta ) }{ m' ( X_i^\top \bbeta ) } \cdot 
\begin{pmatrix}
I ( A_i = 1 ) Z_{i|1}  \\
\vdots \\
I ( A_i = k ) Z_{i|k} 
\end{pmatrix} 
\end{align*}
which defines the augmentation term in the modified score equation.  Together with the results in Appendix~\ref{append:wm-strat-mle}, it suggests that the corresponding formulae of IFs and variance estimation for g-computation estimators of stratified working models with MLE are also applicable for those with FC. 

\subsection{Bias reduction} \label{append:wm-strat-br}

Our proposals in Sections~\ref{sec:method-mu}~and~\ref{sec:method-var} are developed to estimate $\bmu$ with a pooled working model (including treatment arm indicators).  They are also applicable to those with stratified working models.  In the following, we present our proposed estimators directly constructed using outputs from estimated stratified working models. 

In the following, we only consider those  $D_i$s with $A_i = a$, since only those data are used to estimate $\mu_a$.  Let $\hh_{ ii \vert a }$ and $\hbpsi_i^{ \bbeta[a] }$ (resp. $\th_{ ii \vert a }$ and  $\tbpsi_i^{ \bbeta[a] }$) be the corresponding leverage score and estimated IF, respectively, for $\hbbeta[a]$ (resp. $\tbbeta[a]$).  Under stratified working models, $\hmu_a^1$, $\tmu_a^0$, and $\tmu_a^1$ are constructed by placing $Z_{i|a}^\top \hbbeta^1[a]$, $Z_{i|a}^\top \tbbeta^0[a]$, and $Z_{i|a}^\top \tbbeta^1[a]$ in $m(\cdot)$ for \eqref{eq:gob-mle}, respectively: 
\begin{align*}
\hbbeta^1[a] & = \hbbeta[a] + \frac{1}{n_a} \SUM{ i \colon A_i = a }{} \hh_{ ii \vert a } \hbpsi_i^{ \bbeta[a] }, \\
\tbbeta^0[a] & = \tbbeta[a] - \frac{1}{2n_a} \tB_a^{-1} \SUM{ i \colon A_i = a }{} \th_{ ii \vert a } \cdot \frac{ \twm_{ i \vert a }'' }{ \twm_{ i \vert a }' } \cdot Z_{ i \vert a }, \\
\tbbeta^1[a] & = \tbbeta^0[a] + \frac{1}{n_a} \SUM{ i \colon A_i = a }{} \th_{ ii \vert a } \tbpsi_i^{ \bbeta[a] }, 
\end{align*}
where $\tB_a$ is the estimated bread matrix for $\tbbeta[a]$.  The above estimators are obtained directly from $\hbbeta^1$ in \eqref{eq:gob-mle-bc}, and $\tbbeta^0$, $\tbbeta^1$ in \eqref{eq:gob-fc-bc} using \eqref{eq:proof-hwm-if-beta} and \eqref{eq:proof-hwm-leverage}.  For the small-sample bias adjustment, the estimated linearized variable stays the same except for replacing $\hh_{ii}$ and $\th_{ii}$ by $\hh_{ ii \vert a }$ and $\th_{ ii \vert a }$, respectively. 

\section{Additional Information on Simulations} \label{append:sim}

\subsection{Data generating processes} \label{append:sim-dgp}

We design two simulation experiments, both with 1:1 randomization and binary outcomes, to evaluate the finite-sample performance of the proposed approach. The first one (\textbf{Experiment I} presented in the main manuscript) simulates a hypothetical trial of $( \mu_1, \mu_2 ) = ( 25 \%, 60 \% )$, with $n$ varying from $60$ to $180$. The other one (\textbf{Experiment II} presented in Appendix~\ref{append:sim-2}) simulates a trial of $( \mu_1, \mu_2 ) = ( 10 \%, 18.8 \% )$ with $n = 500$. In both experiments, the simulation results are based on 10,000 repeated runs.  

For each experiment, the outcomes are drawn from a Bernoulli distribution, with 
\[
\pr ( Y_i=1 \vert A_i=a, W_i ) = \text{expit} \left( \beta_a^A + \SUM{j=1}{q} \beta_j^W W_{ij}^\ast \right),
\]
where $\text{expit}(x)=1/(1+e^{-x})$ and $q=10$ for Experiment I and $35$ for Experiment II.  The randomization scheme, to generate $A_i\in\{1,2\}$, is designed to approximate simple randomization, while the sizes of the two arms are kept exactly the same, reflecting a completely randomized experiment.  Each of those baseline covariates is independently drawn from a standard normal distribution.  The values of the treatment and covariate effects are carefully calibrated to simulate the two hypothetical trials, 
\begin{itemize}
\item \textbf{Experiment I} ($q=10$): $( \beta_1^A, \beta_2^A ) = ( -1.5836, 0.5923 )$, $\beta_j^W=\sqrt{0.8\times\log(5)^2/4}$ for $j=1,\ldots,4$ and $\sqrt{0.2\times\log(5)^2/6}$ for $j=5,\ldots,10$, such that $\SUM{j=1}{q} \{ \beta_j^W \}^2 = \log(5)^2$  
\item \textbf{Experiment II} ($q=35$): $( \beta_1^A, \beta_2^A ) = ( -4.7173, -3.2523 )$, $\beta_j^W=\sqrt{\log(25)^2/35}$ for all $j$s, such that $\SUM{j=1}{q} \{ \beta_j^W \}^2 = \log(25)^2$. 
\end{itemize}
The choice ensures that the cumulative effect of all $W_{ij}^\ast$ on the outcome $Y_i$ is moderate (odds ratio of 5) and strong (odds ratio of 25), respectively for two hypothetical trials, while the effect of a single covariate is weak. 

In each simulated trial, a logistic regression is fitted to adjust for a set of covariates, transformed from $W_{ij}$s as follows:
\begin{itemize}
\item \textbf{Experiment I}: $W_{ij}=W_{ij}^\ast+5$ for $j=1,\ldots,4$ and $\abs{W_{ij}^\ast}+5$ for $j=5,\ldots,10$;  
\item \textbf{Experiment II}: $W_{ij}=W_{ij}^\ast+5$ for $j=1,\ldots,30$ and $\abs{W_{ij}^\ast}+5$ for $j=31,\ldots,35$. 
\end{itemize}
The transformation ensures that the working models are always misspecified. 

\subsection{Working model fit} \label{append:sim-est}

All MLEs are obtained using \textbf{glm} function in R Statistical Software and FC estimates are obtained using \textbf{brglm2} package \citep{kosmidis2023brglm2}.  We set the maximal number of iterations to be 500 for both MLE and FC, so their fitting algorithm should converge numerically in all simulations even when stationary points do not exist.  For example, MLE may not exist when the data is near complete separation for binary responses in our simulation experiments, and then, in theory, the fitting algorithm should diverge.  However, in practice, it would generally stops at unreasonably large absolute values with numerical convergence criteria being met. 

\clearpage

\subsection{Simulation Experiment I: additional simulation results} \label{append:sim-res}

\begin{figure}[ht!]
\centering
\includegraphics[width=\linewidth]{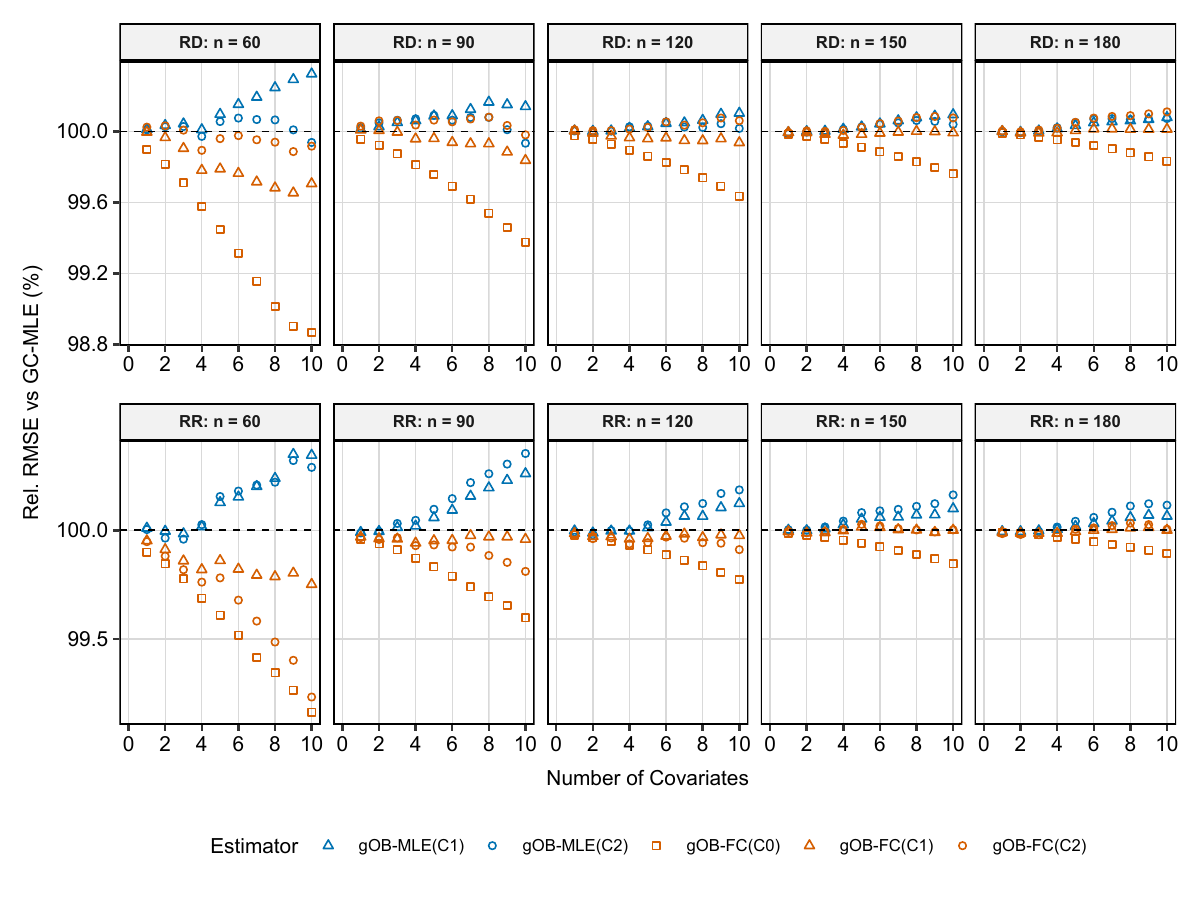}
\caption{
The RMSE---relative to \textbf{GC-MLE} ($\hmu_a$)---of \textbf{gOB-MLE(C1)} ($\hmu_a^1$), \textbf{gOB-MLE(C2)} ($\hmu_a^2$), \textbf{gOB-FC(C0)} ($\tmu_a^0$), \textbf{gOB-FC(C1)} ($\tmu_a^1$) and \textbf{gOB-FC(C2)} ($\tmu_a^2$). RD: $\mu_2-\mu_1$; RR: $\mu_1/\mu_2$; RMSE: Root Mean Square Error; $n$: sample size; Rel.: Relative. 
}
\label{fig:mse}
\end{figure}

\begin{figure}[ht!]
\centering
\includegraphics[width=\linewidth]{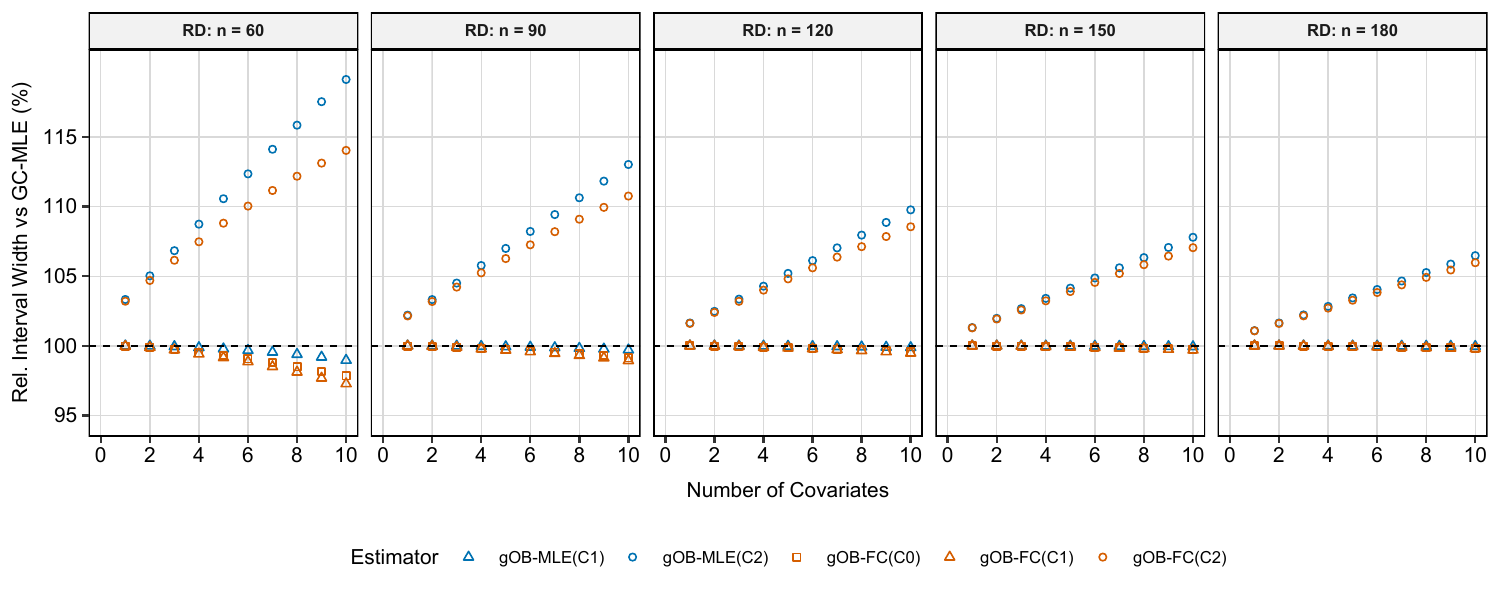}
\caption{
The width of the 95\% CI---relative to \textbf{GC-MLE} ($\hmu_a$) with the proposed small-sample bias adjustment---of \textbf{gOB-MLE(C1)} ($\hmu_a^1$), \textbf{gOB-MLE(C2)} ($\hmu_a^2$), \textbf{gOB-FC(C0)} ($\tmu_a^0$), \textbf{gOB-FC(C1)} ($\tmu_a^1$) and \textbf{gOB-FC(C2)} ($\tmu_a^2$). RD: $\mu_2-\mu_1$; $n$: sample size; Rel.: Relative. 
}
\label{fig:width}
\end{figure}

\clearpage

\FloatBarrier

\subsection{Simulation Experiment II} \label{append:sim-2}

\begin{figure}[ht!]
\centering
\includegraphics[width=\linewidth]{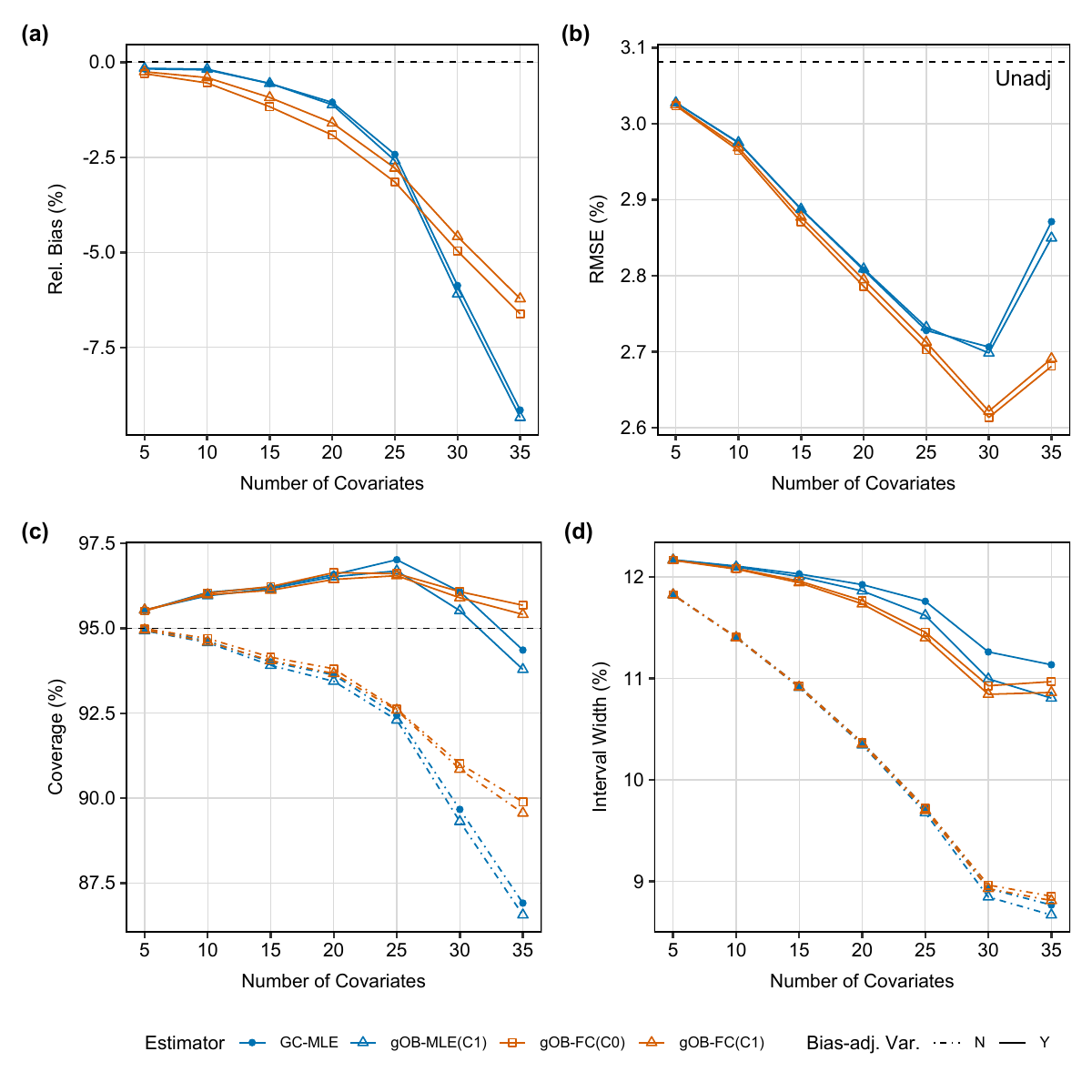}
\caption{
The bias/RMSE of treatment effect estimation (a \& b) and the coverage/width of the 95\% CI (c \& d) for \textbf{GC-MLE} ($\hmu_a$), \textbf{GC-MLE(C1)} ($\hmu_1$), \textbf{gOB-FC(C0)} ($\tmu_a^0$) and \textbf{gOB-FC(C1)} ($\tmu_a^1$) to estimate $\mu_2-\mu_1$. All estimators use two separate (i.e., stratified) working models. Unadj: unadjusted estimator; Bias-adj. Var.: bias-adjusted variance estimator.}
\label{fig:bc_strat}
\end{figure}

We conduct a simulation experiment to evaluate the finite-sample performance of the proposed estimators with stratified working models (Appendix~\ref{append:wm-strat-br}) adjusting for baseline covariates of moderately high dimensions.  We consider a hypothetical trial (1:1 ratio) of $n=500$, $\mu_1 = 10\%$, $\mu_2 = 18.8\%$, and the estimand is $\mu_2-\mu_1$. The details of the data generating process are provided in Appendix~\ref{append:sim-dgp}, and that for nuisance parameter estimation are provided in Appendix~\ref{append:sim-est}. A stratified working model (logistic regression) is applied, that is, a separate logistic regression is fitted to each arm and each model adjusts for up to $35$ baseline covariates. We compare \textbf{GC-MLE(C1)} ($\hmu_a^1$), \textbf{gOB-FC(C0)} ($\tmu_a^0$) and \textbf{gOB-FC(C1)} ($\tmu_a^1$), proposed in \eqref{eq:gob-mle-bc} and \eqref{eq:gob-fc-bc}, with \textbf{GC-MLE} ($\hmu_a$). Note: under stratified working models, \textbf{gOB-MLE(C1)} and \textbf{gOB-FC(C1)} are already free of $O(n^{-1})$ bias (Theorem~\ref{thm:gob-mle}~\&~\ref{thm:gob-fc}).

Figure~\ref{fig:bc_strat} presents simulation results for the four estimators---relative bias, RMSE, interval coverage, and interval width. Their performance is largely indistinguishable when the number of covariates is fewer than 25. With further increased covariates, the two FC-based estimators outperform the two MLE-based estimators, yielding lower bias and RMSE. The under-performance of the MLE-based estimators likely reflects data separation: when adjusting for many covariates, some $\abs{\hbbeta_r}$ values become unreasonably large, deviating substantially from their true values. Moreover, the proposed small-sample adjustment substantially improves interval coverage. In particular, for the two FC-based debiased estimators, coverage remains at or above the nominal level even with 35 adjusted covariates. These estimators deliver the most precise intervals while preserving nominal-or-better coverage.

\FloatBarrier

\begin{table}[t!]
\centering
\begin{tabular}{l|ccccc|c}
\hline
    Estimator          &     RD (\%) & SE (\%) & $\mathrm{RE}_i$ (\%) & 95\% CI (\%)  & Width (\%) & Adj. Var. \\[0.5ex] \hline\hline
         \textbf{Unadj} &    14.04    &   4.206 &                   NA & 5.79, 22.28  & 16.49 & \multirow{2}{*}{No} \\
    \multirow{2}{*}{\textbf{GC-MLE}}  & \multirow{2}{*}{12.06}  
                                      &   3.841 &                16.57 & 4.53, 19.59  & 15.06 & \\\cline{7-7}
                         &            &   4.042 &                 7.63 & 4.14, 19.98  & 15.84 & \multirow{4}{*}{Yes} \\ 
    \textbf{gOB-MLE(C1)} &   11.97    &   4.041 &                 7.68 & 4.05, 19.89  & 15.84 & \\
    \textbf{gOB-FC(C0)}  &   12.05    &   4.040 &                 7.73 & 4.13, 19.97  & 15.84 & \\
    \textbf{gOB-FC(C1)}  &   12.00    &   4.039 &                 7.76 & 4.08, 19.91  & 15.83 & \\
    \hline
\end{tabular}
\caption{Summary of the full population analysis ($n=516$) for the CTN-03 study. RD: risk difference; SE: standard error; $\mathrm{RE}_i$: relative efficiency improvement to Unadj (one minus the ratio between the two variance estimates); CI: confidence interval; Width: interval width; Adj. Var.: bias-adjusted variance estimator}.
\label{tab:ctn03-full}
\end{table}

\section{Application: Analysis II} \label{sec:app-1}

This analysis (Table~\ref{tab:ctn03-full}) evaluates the benefit of our bias‑adjusted variance estimator for \textbf{GC-MLE} (Section~\ref{sec:method-var}) in settings with many adjusted covariates. In such settings, the standard IF‑based variance estimator tends to underestimate variability, resulting in sub-nominal coverage and unreliable inference, as demonstrated in \textbf{Simulation Experiment II} (Figure~\ref{fig:bc_strat}c, Appendix~\ref{append:sim-2}). For \textbf{GC-MLE} under stratified working models with total sample size 500, adjusting for more than 10 variables yields CIs with coverage below the nominal level when the standard IF‑based estimator is used. By contrast, the proposed biased-adjusted variance estimator maintains near‑nominal coverage while adjusting for more than 30 variables. 

This analysis includes 516 participants (261 control, 255 treatment). We fit stratified working models, that is, separate logistic regressions by arm, adjusting for age (continuous), sex (2 levels), race (5 levels), the stratification factor (3 levels), opioid urine toxicology (2 levels), Adjective Rating Scale for Withdrawal (ARSW) Score (continuous), and Clinical Opiate Withdrawal Scale (COWS) Score (continuous). Including an intercept, each model comprises 12 regression coefficients. 

Table~\ref{tab:ctn03-full} reports the unadjusted and adjusted analyses. We present estimated standard errors for \textbf{GC-MLE} both without and with the small-sample bias adjustment. We also report results for \textbf{gOB-MLE(C1)}, \textbf{gOB-FC(C0)}, and \textbf{gOB-FC(C1)} as sensitivity analyses.

\textbf{GC-MLE} with the standard IF-based variance estimator \eqref{eq:if-mu-aipw} shows a 16.57\% gain in relative efficiency and yields narrower 95\% CIs. However, the bias-adjusted variance estimator indicates that only 7.63\% of this gain reflects genuine efficiency improvement, implying that more than half of the apparent gain stems from underestimation of variability. The point estimates from \textbf{gOB-MLE(C1)}, \textbf{gOB-FC(C0)}, and \textbf{gOB-FC(C1)} are nearly identical to those from \textbf{GC-MLE}. With the proposed small-sample bias adjustment, the associated standard errors and CIs are likewise similar.

Moreover, the standard IF-based variance estimator yields a 95\% CI that excludes values below 4.53\%, suggesting the RD is unlikely to be smaller than this threshold. This is misleading: the bias-adjusted variance estimator produces a much wider 95\% CI with a lower-bound near 4.10\% (about 9.5\% lower) highlighting the risk of overconfident inference.

\section{Auxiliary Technical Results}

\subsection{Equivalence of empirical and theoretical IFs} \label{append:proof-zhang}

Let $X_i^\top = ( J_i^\top, Z_i^\top )$ and $X_{i|a}^\top = ( J(a)^\top, Z_{i|a}^{\top} )$, where $J_i = ( I ( A_i = 1 ), I ( A_i = 2 ), \ldots,  I ( A_i = k ))^\top$ and $J(a)$ is a $k$-dimensional vector with only the $a$-th row being 1 and the other being 0.  This formulation of $X_i$ is equivalent to the usual one, with an intercept and the treatment assignment.  The former formulation is adopted for mathematical convenience. 

Since $J_i^\top J ( a ) \equiv I ( A_i = a )$, $J(a)$ is not a random vector and 
\[
B = \E [ m_i' \cdot X_i X_i^\top ] = 
\begin{pmatrix}
\E [ m_i' \cdot J_i J_i^\top ] & \E [ m_i' \cdot J_i Z_i^\top ] \\ 
\E [ m_i' \cdot Z_i J_i^\top ] & \E [ m_i' \cdot Z_i Z_i^\top ]
\end{pmatrix}.
\]
we have that,
\begin{align*}
B \cdot \begin{pmatrix} J(a) / \pi_a \\ \bzero_{p-k} \end{pmatrix}
& = \begin{pmatrix} 
    \E [ m_i' \cdot J_i I ( A_i = a ) ] / \pi_a \\ 
    \E [ m_i' \cdot Z_i I ( A_i = a ) ] / \pi_a \end{pmatrix} 
= \begin{pmatrix}  
    \E [ m_{i|a}' \cdot J(a) | A_i = a ] \\ 
    \E [ m_{i|a}' \cdot Z_{i|a} | A_i = a ] \end{pmatrix} \\
& = \ \E [ m_{i|a}' \cdot X_{i|a} | A_i = a ] = \E [ m'_{i|a} \cdot X_{i|a} ]. 
\end{align*}
The last equality holds due to simple randomization.  Therefore, we obtain \eqref{eq:zhang}:
\[
\E [ m'_{i|a} \cdot X_{i|a}^\top ] B^{-1} X_i
= \begin{pmatrix} J (a)^\top/\pi_a & \bzero_{p-k}^\top \end{pmatrix} \cdot \begin{pmatrix} J_i \\ Z_i \end{pmatrix} 
=  \frac{I(A_i=a)}{\pi_a}.
\]

\subsection{First-order bias formulae for nuisance parameter estimators} \label{append:proof-bias-beta}

\subsubsection{MLEs under misspecification}
\label{app:MLE expansion}

In this section, we prove \eqref{bias-MLE} in the main manuscript (Section~\ref{sec:pre-bias}).  Let both $U_i ( \bbeta )$ and $\bbeta$ be of $p$ dimensions, and write $U_i ( \bbeta ) \coloneqq ( U_{i1} ( \bbeta ), \ldots, U_{ip} ( \bbeta ) )^\top$.  Since $\hbbeta$ is the MLE of $\bbeta$, $\hU^{(n)} \coloneqq \SUM{i=1}{n} \hU_i \equiv 0$.  We further write $U^{(n)} \coloneqq U^{(n)} ( \bbeta_0 )$, $U^{(n)} = ( U_1^{(n)}, \ldots, U_p^{(n)} )^\top$, $\DERIV{ U_r^{(n)} }{} \coloneqq \DERIV{ U_r^{(n)} }{\bbeta}$, and $\DERIV{ U^{(n)} }{} \coloneqq \DERIV{ U^{(n)} }{ \bbeta } = ( \DERIV{ U_1^{(n)} }{}, \ldots, \DERIV{ U_p^{(n)} }{} )$.  Following the standard first-order approximation of $M$-estimators \citep{stefanski2002calculus}, we have 
\begin{equation}
\label{first-order MLE}
\hbbeta - \bbeta_0 = \underbrace{ B^{-1} U^{(n)} / n }_{ O_{\sp} ( n^{-1/2} ) } + O_{\sp} ( n^{-1} ).
\end{equation}
Assuming that $U_i ( \bbeta )$ is the score equation of a correctly specified parametric model, \citet{cox1968general} derived the first-order bias of $\hbbeta$ by a standard second-order Taylor expansion of the score equation $\hat{U}^{(n)}_r = 0$: for $r = 1, \ldots, p$,
\begin{equation}
U^{(n)}_r + \DERIV{ U^{(n)}_r }{}^\top ( \hbbeta - \bbeta_0 ) + \frac{1}{2} ( \hbbeta - \bbeta_0 )^\top \DDERIV{U^{(n)}_r }{} ( \hbbeta - \bbeta_0 ) + O_{\sp} ( n^{-1} ) = 0, \label{eq:proof-score-2order}
\end{equation}
where we let $\DDERIV{U^{(n)}_r}{} \coloneqq \nabla^{2}_{\bbeta \bbeta} U_{r}^{(n)}$ to simplify the notation.  After taking the expectation on both sides of \eqref{eq:proof-score-2order}, the first term in the LHS of \eqref{eq:proof-score-2order} has mean zero.  In the following, we compute the means of the next two terms without assuming that the working GLM is correct.   

We first consider the mean of the second term in the LHS of \eqref{eq:proof-score-2order}.  Let $\hbeta_s$ (resp. $\beta_{0s}$) be the $s$th element of $\hbbeta$ (resp. $\bbeta_0$), $\DERIV{U^{(n)}_r }{s} \coloneqq \DERIV{U^{(n)}_r }{\beta_s}$, and $B_{r\colon}^\top$ be the $r$th row of $B$.  We have that 
\begin{align}
\E \left[ \DERIV{ U^{(n)}_r }{}^\top ( \hbbeta - \bbeta_0 ) \right] 
& = \SUM{s}{} \E \left[ \DERIV{U^{(n)}_r }{s} ( \hbeta_s - \beta_{0s} ) \right] \nonumber \\
& = \SUM{s}{} \E \left[ \DERIV{ U^{(n)}_r }{s} \right] \E \left[ \hbeta_s - \beta_{0s} \right] + \SUM{s}{} \Cov \left( \DERIV{ U^{(n)}_r }{s}, \hbeta_s - \beta_{0s} \right) \nonumber \\
& = - n B_{r\colon}^\top \E \left[ \hbbeta - \bbeta_0 \right] + \SUM{s}{} \Cov \left( \DERIV{ U^{(n)}_r }{s}, \hbeta_s - \beta_{0s} \right). \label{eq:proof-score-2order-1-1}
\end{align}
Let $B^{s\colon\top}$ be the $s$th row for $B^{-1}$ and $B^{st}$ be the $(s,t)$ cell of $B^{-1}$.  Applying \eqref{first-order MLE}, we can replace $\hbeta_s - \beta_{s0}$ by $B^{s\colon\top} U^{(n)} / n + R_s$, where $R_s = O_{\sp} ( n^{-1} ) $ is a random variable, in the second term of the RHS of \eqref{eq:proof-score-2order-1-1}
\begin{align}
& \ \SUM{s}{} \Cov \left( \DERIV{ U^{(n)}_r }{s}, \hbeta_s - \beta_{0s} \right) \nonumber \\
= & \ \frac{1}{n} \SUM{s}{} \Cov \left( \DERIV{ U^{(n)}_r }{s}, U^{(n)} \right) B^{s\colon} + \SUM{s}{} \Cov \left( \DERIV{ U^{(n)}_r }{s}, R_s \right) \nonumber \\
= & \ \SUM{s}{} \underbrace{ \Cov \left( \DERIV{U_{ir}}{s}, U_i \right) }_{ \text{due to simple randomization} } B^{s\colon}  + \SUM{s}{} \left\{ \E \left[ \DERIV{ U^{(n)}_r }{s} \cdot R_s \right] - \E \left[ \DERIV{ U^{(n)}_r }{s} \right] \E [ R_s ] \right\} \nonumber \\
= & \ \SUM{s}{} \E \left[ \DERIV{U_{ir}}{s} \cdot U_i^\top \right] B^{s\colon} + \SUM{s}{} \E \big[ \underbrace{ \left( \DERIV{ U^{(n)}_r }{s} - \E \left[ \DERIV{ U^{(n)}_r }{s} \right] \right) }_{ O_{\sp} ( n^{1/2} ) } R_s \big] \nonumber \\
= & \ \SUM{s,t}{} \E \left[ \DERIV{U_{ir}}{s} \cdot U_{it} \right] B^{st} + O ( n^{-1/2} ). \label{eq:proof-score-2order-1-2}
\end{align}
With \eqref{eq:proof-score-2order-1-1} and \eqref{eq:proof-score-2order-1-2}, we have that 
\begin{equation}
\E \left[ \DERIV{ U^{(n)}_r }{}^\top ( \hbbeta - \bbeta_0 ) \right] = - n B_{r\colon}^\top \E [ \hbbeta - \bbeta_0 ] + \SUM{s,t}{} \E \left[ \DERIV{U_{ir}}{s} \cdot U_{it} \right] B^{st} + O ( n^{-1/2} ). \label{eq:proof-score-2order-1}
\end{equation}

We next consider the mean of the third term in the LHS of \eqref{eq:proof-score-2order}.  Following a similar strategy, we obtain that 
\begin{align}
& \E \left[ \left( \hbbeta - \bbeta_0 \right)^\top \DDERIV{ U^{(n)}_r }{} \left( \hbbeta - \bbeta_0 \right) \right] 
= \SUM{s,t}{} \E \left[ \left( \hbeta_s - \beta_{s0} \right) \left( \hbeta_t - \beta_{t0} \right) \DDERIV{ U^{(n)}_r }{st} \right] \nonumber \\
= & \SUM{s,t}{} \E \left[ \left( \hbeta_s - \beta_{s0} \right) \left( \hbeta_t - \beta_{t0} \right) \right] \E \left[ \DDERIV{ U^{(n)}_r  }{st} \right] + \SUM{s,t}{} \Cov \left( \left\{ \hbeta_s - \beta_{s0} \right\} \left\{ \hbeta_t - \beta_{t0} \right\}, \DDERIV{ U^{(n)}_r }{st} \right),  \label{eq:proof-score-2order-2-1}
\end{align}
where we let $\DDERIV{ U^{(n)}_r }{st} \coloneqq \DDERIV{ U^{(n)}_r }{\beta_s\beta_t}$ to simplify the notation.  Applying \eqref{first-order MLE}, we can replace $\hbeta_s - \beta_{s0}$ by $B^{s\colon\top} U^{(n)} / n + R_s$ (similarly for $\hbeta_t$). For the first term in the RHS in \eqref{eq:proof-score-2order-2-1}, we have that 
\begin{align}
& \SUM{s,t}{} \E \left[ \left( \hbeta_s - \beta_{s0} \right) \left( \hbeta_t - \beta_{t0} \right) \right] \E \left[ \DDERIV{ U^{(n)}_r  }{st} \right] \nonumber \\
= & \SUM{s,t}{} \Big\{ n^{-2} B^{s\colon\top} \E \left[ U^{(n)} U^{(n)\top} \right] B^{t\colon} + n^{-1} \underbrace{ \E \left[ R_s \cdot U^{(n)\top} \right] B^{t\colon} }_{ O ( n^{-1/2} ) } + n^{-1} \underbrace{ B^{s\colon\top} \E \left[ U^{(n)} \cdot R_t \right] }_{ O ( n^{-1/2} ) } \nonumber \\
& \qquad  + \underbrace{ \E \left[ R_s R_t \right] }_{ O ( n^{-2} ) } \Big\} \underbrace{ \E \left[ \DDERIV{ U^{(n)}_r }{st} \right] }_{ O ( n ) } \nonumber \\
= & \SUM{s,t}{} \underbrace{ B^{s\colon\top} \E \left[ U_i U_i^\top \right] B^{t\colon} \cdot \E \left[ \DDERIV{U_{ir}}{st} \right] }_{ \text{due to simple randomization} } + O ( n^{-1/2} ) \nonumber \\
= & \SUM{s,t}{} B^{s\colon\top} M B^{t\colon} \cdot \E \left[ \DDERIV{U_{ir}}{st} \right] + O ( n^{-1/2} ).\label{eq:proof-score-2order-2-2}
\end{align}
For the second term in the RHS in \eqref{eq:proof-score-2order-2-1}, we have that 
\begin{align}
& \SUM{s,t}{} \Cov \left( \left\{ \hbeta_s - \beta_{s0} \right\} \left\{ \hbeta_t - \beta_{t0} \right\}, \DDERIV{ U^{(n)}_r }{st} \right) \nonumber \\
= & \SUM{s,t}{} \Cov \left( n^{-2 } B^{s\colon\top} U^{(n)} U^{(n)\top} B^{t\colon} + n^{-1} R_s \cdot U^{(n)\top} B^{t\colon} + n^{-1} B^{s\colon\top} U^{(n)} \cdot R_t + R_s R_t, \DDERIV{ U^{(n)}_r }{st} \right) \nonumber \\
= & \SUM{s,t}{} n^{-1}  \underbrace{ \Cov \left( B^{s\colon\top} U_i U_i^\top B^{t\colon}, \DDERIV{ U_{ir} }{st} \right) }_{ \text{due to simple randomization} } + \nonumber \\
& \qquad n^{-1} \SUM{s,t}{} \E \Big( \underbrace{ R_s \cdot U^{(n)\top} B^{t\colon} }_{ O_{\sp} ( n^{-1/2} ) } \underbrace{ \left\{ \DDERIV{ U^{(n)}_r }{st} - \E \left[ \DDERIV{ U^{(n)}_r }{st} \right] \right\} }_{ O_{\sp} ( n^{1/2} ) } \Big) + \nonumber \\
& \qquad n^{-1} \SUM{s,t}{} \E \Big( \underbrace{ B^{s\colon\top} U^{(n)} \cdot R_t }_{ O_{\sp} ( n^{-1/2} ) } \underbrace{ \left\{ \DDERIV{ U^{(n)}_r }{st} - \E \left[ \DDERIV{ U^{(n)}_r }{st} \right] \right\} }_{ O_{\sp} ( n^{1/2} ) } \Big) + \nonumber \\
& \qquad \SUM{s,t}{} \E \Big( \underbrace{ R_s R_t }_{ O_{\sp} ( n^{-2} ) } \underbrace{ \left\{ \DDERIV{ U^{(n)}_r }{st} - \E \left[ \DDERIV{ U^{(n)}_r }{st} \right] \right\} }_{ O_{\sp} ( n^{1/2} ) } \Big) = \ O ( n^{-1} ).  \label{eq:proof-score-2order-2-3}
\end{align}
With \eqref{eq:proof-score-2order-2-1}, \eqref{eq:proof-score-2order-2-2} and \eqref{eq:proof-score-2order-2-3}, we have that 
\begin{equation}
\E \left[ \left( \hbbeta - \bbeta_0 \right)^\top \DDERIV{ U^{(n)}_r }{} \left( \hbbeta - \bbeta_0 \right) \right] = \SUM{s,t}{} B^{s\colon\top} M B^{t\colon} \cdot \E \left[ \DDERIV{U_{ir}}{st} \right] + O ( n^{-1/2} ). \label{eq:proof-score-2order-2}
\end{equation}

Finally, combining \eqref{eq:proof-score-2order}, \eqref{eq:proof-score-2order-1}, and \eqref{eq:proof-score-2order-2} for $r = 1, \ldots, p$, we have that  
\begin{multline*}
- n B \E \left[ \hbbeta - \bbeta_0 \right] + \SUM{s,t}{}  B^{st} \E \left[ U_{it} \cdot \DERIV{U_i}{s} \right] + \\ \frac{1}{2} \SUM{s,t}{} B^{s\colon\top} M B^{t\colon} \cdot \E \left[ \DDERIV{U_i}{st} \right] + O ( n^{-1/2} ) = \bzero_p.
\end{multline*}
Then, we have that $ \E [ \hbbeta - \bbeta_0 ] = n^{-1} \bb_1 ( \hbbeta ) + O ( n^{-3/2} )$, where
\begin{equation}
\bb_1 ( \hbbeta ) = B^{-1} \SUM{s,t}{} \left\{ B^{st} \E \left[ U_{it} \cdot \DERIV{U_i}{s} \right] + \frac{1}{2} B^{s\colon\top} M B^{t\colon} \cdot \E \left[ \DDERIV{U_i}{st} \right] \right\} \label{eq:proof-mle-bias-vector}.
\end{equation}
This is exactly the same as the result provided in \citet[][Equation (5)]{kosmidis2024empirical} without the augmentation term.  When the model is correctly specified, $B^{-1} = M$ and thus $B^{s\colon}{}^\top M B^{t\colon} = B^{st}$.  Then, the RHS of \eqref{eq:proof-mle-bias-vector} reduces to 
\[
B^{-1} \SUM{s,t}{} B^{st} \E \left[ \frac{1}{2} \DDERIV{U_i}{st} + U_{it} \cdot \DERIV{U_i}{s} \right],
\]
which is exactly the same as the result provided in \citet[][Equation (20)]{cox1968general}.  Further, when this true model is a canonical GLM, $\E \left( U_{it} \cdot \DERIV{U_i}{s} \right) = \bzero_p$ \citep{cordeiro1991bias,firth1993bias}, the RHS of \eqref{eq:proof-mle-bias-vector} reduces to 
\[
\frac{1}{2} B^{-1} \SUM{s,t}{} B^{st} \E \left( \DDERIV{U_i}{st} \right).
\]

Finally, we show that $2 B \bb_1 ( \hbbeta ) \equiv \DERIV{ \tr ( B^{-1} M ) }{} \equiv \DERIV{ \tr ( B^{-1} M ) }{\bbeta}$ to complete the proof.  For $r = 1, \cdots, p$, we take the partial derivative of $\tr ( B^{-1} M )$ with respect to $\beta_r$ and write $\DERIV{ \tr ( B^{-1} M ) }{\beta_r}$ as $\DERIV{ \tr ( B^{-1} M ) }{r}$ for short,
\begin{align*}
& \DERIV{ \tr ( B^{-1} M ) }{r} = \ \tr \left( \DERIV{ \{ B^{-1} M \} }{r} \right) = \ \tr \left( \DERIV{ B^{-1} }{r} \cdot M \right) + \tr \left( B^{-1} \DERIV{M}{r} \right) \\
= & \ \tr \left( - B^{-1} \DERIV{ B }{r} \cdot B^{-1} M \right) + \tr \left( B^{-1} \DERIV{ \E \left[ U_i U_i^\top \right] }{r} \right) \\ 
= & \ \tr \left( B^{-1} \DERIV{ \{ \DERIV{ \E \left[ U_i \right] \} }{} }{r} B^{-1} M \right) + \tr \left( B^{-1} \E \left[ \DERIV{ \left\{ U_i U_i^\top \right\} }{r} \right] \right) \\
= & \ \tr \left( \E \left[ \DERIV{ \{ \DERIV{ U_i }{} \} }{r} \right] B^{-1} M B^{-1} \right) + \tr \left( B^{-1} \E \left[ \DERIV{ U_i }{r} \cdot U_i^\top \right] \right) + \tr \left( B^{-1} \E \left[ U_i \DERIV{ U_i }{r}^\top \right] \right) \\
= & \ \tr \left( B^{-1} M B^{-1} \E \left[ \DDERIV{ U_{ir} }{} \right] \right) + \E \left[ U_i^\top B^{-1} \DERIV{ U_{ir} }{} \right] + \E \left[ \DERIV{ U_{ir} }{}^\top \cdot B^{-1} U_i \right] \\
= & \ \SUM{s,t}{} B^{s\colon\top} M B^{t\colon} \E \left[ \DDERIV{U_{ir}}{st} \right] + \SUM{s,t}{} B^{st} \E \left[ U_{it} \DERIV{U_{ir}}{s} \right] + \SUM{s,t}{} B^{st} \E \left[ \DERIV{U_{ir}}{t} \cdot U_{is} \right]  \\ 
= & \ \SUM{s,t}{} B^{s\colon\top} M B^{t\colon} \E \left[ \DDERIV{U_{ir}}{st} \right] + 2 \SUM{s,t}{} B^{st} \E \left[ U_{it} \DERIV{U_{ir}}{s} \right].
\end{align*}
Comparing the last term in the above display with \eqref{eq:proof-mle-bias-vector} completes the proof for $\bb_1 ( \hbbeta )$.  Besides, the last identity in the above equation suggests that 
\begin{equation}
\DERIV{ \tr ( B^{-1} M ) }{r} = \ \tr \left( B^{-1} M B^{-1} \E \left[ \DDERIV{ U_{ir} }{} \right] \right) + 2 \tr \left( B^{-1} \E \left[ U_i \DERIV{ U_i }{r}^\top \right] \right). \label{proof-mle-bias-vector-ext} 
\end{equation}

\subsubsection{FC estimators under misspecification}
\label{app:FC expansion}

In this section, we prove \eqref{bias-Firth} in the main manuscript (Section~\ref{sec:pre-bias}).  \citet{firth1993bias} proposed the modified score equation to remove the first-order bias of $\hbbeta$.  It is written as $ U^{(n)} ( \bbeta ) + \Delta^{(n)} ( \bbeta ) = \bzero_p$, where the $r$th row of $\Delta^{(n)} ( \bbeta )$ (the augmented term) reads as 
\begin{align*}
\Delta_r^{(n)} ( \bbeta ) 
& \coloneq -\frac{1}{2} \SUM{s,t}{} \eB^{st} ( \bbeta ) \SUM{i=1}{n} \DDERIV{ U_{ir} ( \bbeta ) }{st} / n \\ 
& = \frac{1}{2} \tr \left( \eB^{-1} ( \bbeta ) \{ \DERIV{\eB}{r} ( \bbeta ) \} \right) \\ 
& = \frac{1}{2} \DERIV{ \log \det ( \eB ( \bbeta ) ) }{r}.  
\end{align*}
Here $\eB ( \bbeta ) \coloneq - n^{-1} \SUM{i=1}{n} \DERIV{ U_i }{} ( \bbeta )$ \citep[][Section 3.1]{firth1993bias}.  Let $\Delta^{(n)} \coloneq \Delta^{(n)} ( \bbeta_0 )$ and $\eB \coloneq \eB ( \bbeta_0 )$.  From \eqref{eq:proof-beta-bread-est-inv}, we have that $\eB^{st} = B^{st} + O_{\sp} ( n^{-1/2} )$ and thus
\begin{align*}
\E \left[ \Delta^{(n)} \right] 
& = -\frac{1}{2} \E \left[ \SUM{s,t}{} \left\{ B^{st} + O_{\sp} ( n^{-1/2} ) \right\} \SUM{i=1}{n} \DDERIV{ U_i }{st} / n \right] \\
& = -\frac{1}{2} \SUM{s,t}{} \left\{ B^{st} + O ( n^{-1/2} ) \right\} \E \left[ \DDERIV{ U_i }{st} \right] \\
& = - \frac{1}{2} \SUM{s,t}{} B^{st} \E \left[ \DDERIV{U_{i}}{st} \right] + O ( n^{-1/2} ) \\
& = \frac{1}{2} \DERIV{ \log \det ( B ) }{} + O ( n^{-1/2} ),
\end{align*}
since 
\begin{align}
\DERIV{ \log \det ( B ) }{r} & = \tr \left( B^{-1} \{ \DERIV{B}{r} \} \right) = - \tr \left( B^{-1} \{ \DERIV{\E \left[ \DERIV{U_i}{} \right] }{r} \} \right) \nonumber \\ 
& = - \tr \left(  B^{-1} \E \left[ \DERIV{ \{ \DERIV{U_i}{} }{r} \} \right] \right) = - \tr \left( B^{-1} \E \left[ \DDERIV{ U_{ir} }{} \right] \right) \label{eq:proof-fc-bias-vector} \\
& = - \SUM{s,t}{} B^{st} \E \left[ \DDERIV{U_{ir}}{st} \right]. \nonumber
\end{align}
Finally, following \citet[][Equation (5)]{kosmidis2024empirical},  we have that 
\begin{align*}
\E [ \tbbeta - \bbeta_0 ] & = \frac{1}{n} B^{-1} \E \left[ \Delta^{(n)} \right] + \frac{1}{n} \bb_1 ( \hbbeta ) + O_{\sp} ( n^{-3/2} ) \\ 
& = \ \frac{1}{n} \left\{ \frac{1}{2}  B^{-1} \DERIV{ \log \det ( B ) }{} + \bb_1 ( \hbbeta ) \right\} + O_{\sp} ( n^{-3/2} ),
\end{align*}
which completes the derivation for $\bb_1 ( \tbbeta )$.

\subsection{Higher-order stochastic expansion} \label{append:proof-beta-hoif}

We obtain the second-order stochastic stochastic expansion of $\hbbeta$ by replacing the notations of Equation (3.3)--(3.5) in \citet{rilstone2024on} as follows: $q_i = U_i$, $\bar{q}^{(1)} = \E [ \DERIV{U_i}{} ] = - B$, $\tilde{q}_i^{(1)} = \DERIV{U_i}{} - \E [ \DERIV{U_i}{} ]$, and $\bar{q}_1^{(2)} \coloneqq \E [ \mathrm{d} \, \DERIV{U_1}{} / \mathrm{d} \, \bbeta ]$, where $\mathrm{d} \, \DERIV{U_1}{} / \mathrm{d} \, \bbeta$ is Kronecker matrix differentiation \citep{macrae1974matrix}, denoting a $p \times p^2$ matrix of the second order derivatives of $U_i$ with respect to $\bbeta$. Obviously, we have that, $d_i = - ( \bar{q}^{(1)} )^{-1} q_i = B^{-1} U_i = \bpsi^{\bbeta}_i$.  Replacing $i_1,i_2$ in \citet{rilstone2024on} with $i,j$, we further have that  
\begin{align*}
\bpsi^{\bbeta, 2}_{ij} \equiv \ d_{ij} 
& = - ( \bar{q}^{(1)} )^{-1} \left\{ \tilde{q}_i^{(1)} d_j + \frac{1}{2}  \bar{q}_1^{(2)} ( d_i \otimes d_j ) \right\} \\
& = B^{-1} \left\{ \left( \DERIV{U_i}{} - \E [ \DERIV{U_i}{} ] \right) \bpsi^{\bbeta}_j + \frac{1}{2} \E [ \mathrm{d} \, \DERIV{U_1}{} / \mathrm{d} \, \bbeta ] ( \bpsi^{\bbeta}_i \otimes \bpsi^{\bbeta}_j ) \right\}.
\end{align*}
After translating the notation from \citet{rilstone2024on} to our notation, we have 
\begin{equation} 
\hbbeta -\bbeta_0 = \frac{1}{n} \SUM{i=1}{n} \bpsi^{\bbeta}_i + \frac{1}{n^2} \SUM{i,j=1}{n} \bpsi^{\bbeta, 2}_{ij} + O_{\sp} ( n^{-3/2} ). \label{eq:proof-beta-mle-expand-2rd}
\end{equation}
Obviously, $\SUM{i}{} \bpsi^{\bbeta}_i = O_{\sp} ( n^{1/2} )$ and $\SUM{ij}{} \bpsi^{\bbeta,2}_{ij} = O_{\sp} ( n )$. Besides, $\E [ \bpsi^{\bbeta,2}_{ij} ] = \bzero_p$ for $i \neq j$ and $\E [ \bpsi^{\bbeta, 2}_{ii} ] = \bb_1 ( \hbbeta )$, since 
\[
\frac{1}{n} \bb_1 ( \hbbeta ) = \E \left[ \frac{1}{n^2} \SUM{i,j=1}{n} \bpsi^{\bbeta,2}_{ij} \right] = \frac{1}{n} \E \left[ \bpsi^{\bbeta,2}_{ii} \right].
\]

Next, we consider the second-order stochastic expansion of $\tbbeta$.  Following \citet{kosmidis2024empirical} (online supplementary material, expression (S1) in Section S3), we have
\begin{align} 
\tbbeta -\bbeta_0 & = \frac{1}{n} \SUM{i=1}{n} \bpsi^{\bbeta}_i + \frac{1}{n^2} \SUM{i,j=1}{n} \bpsi^{\bbeta, 2}_{ij} + \frac{1}{n} B^{-1} \Delta^{(n)} + O_{\sp} ( n^{-3/2} ) \nonumber \\
& = \frac{1}{n} \SUM{i=1}{n} \bpsi^{\bbeta}_i + \frac{1}{n^2} \SUM{i,j=1}{n} \bpsi^{\bbeta, 2}_{ij} + \frac{1}{n} B^{-1} H_3 + O_{\sp} ( n^{-3/2} ), \label{eq:proof-beta-fc-expand-2rd}
\end{align}
since 
\begin{align*}
\Delta^{(n)} & = \frac{1}{2} \SUM{i=1}{n} m_i'' X_i^\top \left\{ \SUM{j=1}{n} m_j' X_j X_j^\top \right\}^{-1} X_i \cdot X_i \\ 
& = \frac{1}{n} \SUM{i=1}{n} \frac{1}{2} m_i'' X_i^\top B^{-1} X_i \cdot X_i + O_{\sp} ( n^{-1/2} ) \\
& = H_3 + \left( \frac{1}{n} \SUM{i=1}{n} \frac{1}{2} m_i'' X_i^\top B^{-1} X_i \cdot X_i - H_3 \right) + O_{\sp} ( n^{-1/2} ) \\
& = H_3 + O_{\sp} ( n^{-1/2} ),
\end{align*}
where we recall that $H_3$ is defined in Proposition~\ref{prop:beta-glm} and the last line follows from the standard central limit theorem. The second equality in the above holds, because $n^{-1} \SUM{j}{} m_j' X_j X_j^\top$ can be written as $B +  O_\sp ( n^{-1/2} )$ and 
\begin{align}
\left\{ \frac{1}{n} \SUM{j}{} m_j' X_j X_j^\top \right\}^{-1} 
& = \left\{ B + \frac{1}{n^{1/2}} \cdot \underbrace{ O_\sp ( 1 ) }_{ p \times p } \right\}^{-1} = B^{-1} - \frac{1}{n^{1/2}} \cdot B^{-1} \cdot \underbrace{ O_\sp ( 1 ) }_{ p \times p } \cdot B^{-1} \nonumber \\
& = B^{-1} + O_\sp ( n^{-1/2} ). \label{eq:append-proof-aux-inv}
\end{align}

\subsection{Auxiliary results used in the proofs} \label{eq:append-proof-aux-others}

Finally, we collect a set of useful auxiliary results in this section. Related notation can be found in Appendix~\ref{app:MLE expansion} and Appendix~\ref{app:FC expansion}.

Denote the $r$th row of $\bpsi^{\bbeta}_i$ as $\psi^{\bbeta}_{ir}$. For $i,j \neq 1$, we have the following set of technical results that are useful in the above derivations.

\begin{lemma}
\label{lem:aux}
The following hold for $i \neq j \neq 1$:
\begin{equation}
\label{eq:proof-bias-gc-suppl-1}
\begin{split}
& \E \left[ m_{1|a}' \cdot X_{1|a}^\top \bpsi_i^{\bbeta} \right] = 0, \quad \E \left[ m_{1|a}' \cdot X_{1|a}^\top \bpsi_{i,1}^{\bbeta,2} \right] = 0, \\
& \E \left[ m_{1|a}' \cdot X_{1|a}^\top \bpsi_{1,j}^{\bbeta,2} \right] = 0, \quad \E \left[ m_{1|a}' \cdot X_{1|a}^\top \bpsi_{i,j}^{\bbeta,2} \right] = \E \left[ m_{1|a}' \cdot X_{1|a}^{\top} \right] \E [ \bpsi_{i,j}^{\bbeta,2} ],
\end{split}
\end{equation}
and
\begin{equation}
\label{eq:proof-bias-gc-suppl-2}
\begin{split}
\E \left[ m_{1|a}'' \cdot X_{1|a}^\top \bpsi^{\bbeta}_i \bpsi^{\bbeta\top}_i X_{1|a} \right] & = \E \left[ m_{1|a}'' \cdot X_{1|a}^\top B^{-1} M B^{-1}  X_{1|a} \right], \\
\E \left[ m_{1|a}'' \cdot X_{1|a}^\top \bpsi^{\bbeta}_i \bpsi^{\bbeta\top}_1 X_{1|a} \right] & = 0, \\
\E \left[ m_{1|a}'' \cdot X_{1|a}^\top \bpsi^{\bbeta}_1 \bpsi^{\bbeta\top}_i X_{1|a} \right] & = 0, \\
\E \left[ m_{1|a}'' \cdot X_{1|a}^\top \bpsi^{\bbeta}_i \bpsi^{\bbeta\top}_j X_{1|a} \right] & = 0.
\end{split}
\end{equation}
\end{lemma}

\begin{proof}
\begin{align*}
\E \left[ m_{1|a}' \cdot X_{1|a}^\top \bpsi_i^{\bbeta} \right] & = \E \left[ m_{1|a}' \cdot X_{1|a}^\top \right] \E [ \bpsi_i^{\bbeta} ] = 0, \\ 
\E \left[ m_{1|a}' \cdot X_{1|a}^\top \bpsi_{i,1}^{\bbeta,2} \right] 
& = \E \left[ m_{1|a}' \cdot X_{1|a}^\top B^{-1} \left( \DERIV{U_i}{} - \E [ \DERIV{U_i}{} ] \right) \bpsi^{\bbeta}_1 \right] \\
& \quad + \frac{1}{2} \E \left[ m_{1|a}' \cdot X_{1|a}^\top B^{-1} \E [ \mathrm{d} \, \DERIV{U_1}{} / \mathrm{d} \, \bbeta ] ( \bpsi^{\bbeta}_i \otimes \bpsi^{\bbeta}_1 ) \right] \\ 
& = \E \left[ m_{1|a}' \tr \left( \bpsi^{\bbeta}_1 X_{1|a}^\top B^{-1} \left\{ \DERIV{U_i}{} - \E [ \DERIV{U_i}{} ] \right\} \right) \right] \\ 
& \quad + \frac{1}{2} \E \left[ \tr \left( B^{-1} \E [ \mathrm{d} \, \DERIV{U_1}{} / \mathrm{d} \, \bbeta ] \{ \bpsi^{\bbeta}_i \otimes \bpsi^{\bbeta}_1 \} X_{1|a}^\top \cdot m_{1|a}' \right) \right] \\ 
& = \tr \left( \E \left[ m_{1|a}' \cdot \bpsi^{\bbeta}_1 X_{1|a}^\top B^{-1} \right] \E \left[ \DERIV{U_i}{} - \E [ \DERIV{U_{i}}{} ] \right] \right) \\ 
& \quad + \frac{1}{2} \tr \left( B^{-1} \E [ \mathrm{d} \, \DERIV{U_1}{} / \mathrm{d} \, \bbeta ] \E \left[ \bpsi^{\bbeta}_i \otimes \{ \bpsi^{\bbeta}_1 X_{1|a}^\top \cdot m_{1|a}' \} \right] \right) \\ 
& = \frac{1}{2} \tr \left( B^{-1} \E [ \mathrm{d} \, \DERIV{U_1}{} / \mathrm{d} \, \bbeta ] \cdot \E [ \bpsi^{\bbeta}_i ] \otimes \E \left[ \bpsi^{\bbeta}_1 X_{1|a}^\top \cdot m_{1|a}' \right] \right) = 0, \\ 
\E \left[ m_{1|a}' \cdot X_{1|a}^\top \bpsi_{1,j}^{\bbeta,2} \right] 
& = \E \left[ m_{1|a}' \cdot X_{1|a}^\top B^{-1} \left( \DERIV{U_1}{} - \E [ \DERIV{U_1}{} ] \right) \bpsi^{\bbeta}_j \right] \\ 
& \quad + \frac{1}{2} \E \left[ m_{1|a}' \cdot X_{1|a}^\top B^{-1} \E [ \mathrm{d} \, \DERIV{U_1}{} / \mathrm{d} \, \bbeta ] ( \bpsi^{\bbeta}_1 \otimes \bpsi^{\bbeta}_j ) \right] \\
& = \E \left[ m_{1|a}' \cdot X_{1|a}^\top B^{-1} \left( \DERIV{U_1}{} - \E [ \DERIV{U_1}{} ] \right) \right] \E [ \bpsi^{\bbeta}_j ] \\ 
& \quad + \frac{1}{2} \SUM{s,t=1}{p} \E \left[ m_{1|a}' \cdot X_{1|a}^\top B^{-1} \E [ \DDERIV{U_1}{st} ] \cdot \psi^{\bbeta}_{1s} \psi^{\bbeta}_{jt} \right] \\
& = \frac{1}{2} \SUM{s,t=1}{p} \E \left[ m_{1|a}' \cdot X_{1|a}^\top B^{-1} \E [ \DDERIV{U_1}{st} ] \cdot \psi^{\bbeta}_{1s} \right] \E [ \psi^{\bbeta}_{jt} ] = 0, \\ 
\E \left[ m_{1|a}' \cdot X_{1|a}^\top \bpsi_{i,j}^{\bbeta,2} \right] 
& = \E \left[ m_{1|a}' \cdot X_{1|a}^\top \right] \E [ \bpsi_{ij}^{\bbeta,2} ]. 
\end{align*}

Next: 
\begin{align*}
\E \left[ m_{1|a}'' \cdot X_{1|a}^\top \bpsi^{\bbeta}_i \bpsi^{\bbeta\top}_i X_{1|a} \right] 
& = \E \left[ \tr \left( \bpsi^{\bbeta}_i \bpsi^{\bbeta\top}_i X_{1|a}  X_{1|a}^\top \cdot  m_{1|a}'' \right) \right] \\
& = \tr \left( \E \left[ \bpsi^{\bbeta}_i \bpsi^{\bbeta\top}_i X_{1|a}  X_{1|a}^\top \cdot m_{1|a}'' \right] \right) \\
& = \tr \left( \E \left[ \bpsi^{\bbeta}_i \bpsi^{\bbeta\top}_i \right] \E \left[ X_{1|a}  X_{1|a}^\top \cdot m_{1|a}'' \right] \right) \\ 
& = \tr \left( B^{-1} M B^{-1} \E \left[ X_{1|a} X_{1|a}^\top \cdot m_{1|a}'' \right] \right) \\ 
& = \E \left[ \tr \left( B^{-1} M B^{-1} \cdot X_{1|a} X_{1|a}^\top \cdot m_{1|a}'' \right) \right] \\   
& = \E \left[ m_{1|a}'' \cdot X_{1|a}^\top B^{-1} M B^{-1} X_{1|a} \right], \\
\E \left[ m_{1|a}'' \cdot X_{1|a}^\top \bpsi^{\bbeta}_i \bpsi^{\bbeta\top}_1 X_{1|a} \right] 
& = \E \left[ \tr \left( \bpsi^{\bbeta}_i \bpsi^{\bbeta\top}_1 X_{1|a} X_{1|a}^\top \cdot m_{1|a}'' \right) \right] \\ 
& = \tr \left( \E \left[ \bpsi^{\bbeta}_i \bpsi^{\bbeta\top}_1 X_{1|a} X_{1|a}^\top \cdot  m_{1|a}'' \right] \right) \\ 
& = \tr \left( \E [ \bpsi^{\bbeta}_i ] \E \left[ \bpsi^{\bbeta\top}_1 X_{1|a} X_{1|a}^\top \cdot m_{1|a}'' \right] \right) = 0, \\
\E \left[ m_{1|a}'' \cdot X_{1|a}^\top \bpsi^{\bbeta}_1 \bpsi^{\bbeta\top}_i X_{1|a} \right] 
& = \E \left[ m_{1|a}'' \cdot X_{1|a}^\top \bpsi^{\bbeta}_i \bpsi^{\bbeta\top}_1 X_{1|a} \right]^\top = 0, \\
\E \left[ m_{1|a}'' \cdot X_{1|a}^\top \bpsi^{\bbeta}_i \bpsi^{\bbeta\top}_j X_{1|a} \right] & = \E \left[ \tr \left( \bpsi^{\bbeta}_i \bpsi^{\bbeta\top}_j X_{1|a} X_{1|a}^\top \cdot m_{1|a}'' \right) \right] \\
& = \tr \left( \E \left[ \bpsi^{\bbeta}_i \bpsi^{\bbeta\top}_j X_{1|a} X_{1|a}^\top \cdot m_{1|a}'' \right] \right) \\ 
& = \tr \left( \E [ \bpsi^{\bbeta}_i ] \E [ \bpsi^{\bbeta\top}_j ] \E \left[ X_{1|a} X_{1|a}^\top \cdot m_{1|a}'' \right] \right) = 0.
\end{align*}
\end{proof}

\begin{lemma}
\label{lem:marginal-conditional}
The following hold for $i, j \neq 1$:
\begin{align}
\E \left[ m_{1|a}' \cdot X_{1|a}^\top \right] B^{-1} H_1 & = \E \left[ m_i' \cdot X_i^\top \bpsi_i^{\bbeta}  \middle\vert A_i = a \right], \label{eq:proof-bias-gc-suppl-3} \\
\E \left[ m_{1|a}' \cdot X_{1|a}^\top \right] B^{-1} H_2 & = \frac{1}{2} \ \E \left[ m_{i|a}'' \cdot X_{i|a}^\top B^{-1} M B^{-1} X_{i|a} \right], \label{eq:proof-bias-gc-suppl-4} \\
\E \left[ m_{1|a}' \cdot X_{1|a}^\top \right] B^{-1} H_3 & = \frac{1}{2} \ \E \left[ m_{i|a}'' \cdot X_{i|a}^\top B^{-1} X_{i|a} \right]. \label{eq:proof-bias-fc-suppl}
\end{align}
\end{lemma}

\begin{proof}
Using \eqref{eq:zhang}, we have the following results,
\begin{align*}
\E \left[ m_{1|a}' \cdot X_{1|a}^\top \right] B^{-1} H_1 
& = \E \left[ m_{1|a}' \cdot X_{1|a}^\top \right] B^{-1} \E \left[ m_i' ( Y_i - m_i ) \cdot X_i^\top B^{-1} X_i \cdot X_i \right] \\ 
& = \E \left[ m_i'( Y_i - m_i )  \cdot X_i^\top B^{-1} X_i \cdot \underbrace{ \E \left[ m_{1|a}' X_{1|a}^\top \right] B^{-1} X_i }_{ \eqref{eq:zhang} } \right] \\
& = \E \left[ m_i' ( Y_i - m_i ) \cdot X_i^\top B^{-1} X_i \cdot \frac{ I ( A_i = a ) }{ \pi_a } \right] \\ 
& = \E \left[ m_i' \cdot X_i^\top \underbrace{ B^{-1} X_i ( Y_i - m_i ) }_{= \, \bpsi_i^{\bbeta} }  \middle\vert A_i = a \right], \\ 
\E \left[ m_{1|a}' \cdot X_{1|a}^\top \right] B^{-1} H_2 
& = \frac{1}{2} \ \E \left[ m_{1|a}' X_{1|a}^\top \right] B^{-1} \E \left[ m_i'' \cdot X_i^\top B^{-1} M B^{-1} X_i \cdot X_i \right] \\
& = \frac{1}{2} \ \E \left[ m_i'' \cdot X_i^\top B^{-1} M B^{-1} X_i \cdot \underbrace{ \E \left[ m_{1|a}' X_{1|a}^\top \right] B^{-1} X_i }_{ \eqref{eq:zhang} } \right] \\
& = \frac{1}{2} \ \E \left[ m_i'' \cdot X_i^\top B^{-1} M B^{-1} X_i \cdot \frac{ I ( A_i = a ) }{ \pi_a } \right] \\
& = \frac{1}{2} \ \E \left[ m_i'' \cdot X_i^\top B^{-1} M B^{-1} X_i \middle\vert A_i = a \right]\\ 
& = \frac{1}{2} \ \E \left[ m_{i|a}'' \cdot X_{i|a}^\top B^{-1} M B^{-1} X_{i|a} \right], \\
\E \left[ m_{1|a}' \cdot X_{1|a}^\top \right] B^{-1} H_3 
& = \frac{1}{2} \ \E \left[ m_{1|a}' X_{1|a}^\top \right] B^{-1} \E \left[ m_i'' \cdot X_i^\top B^{-1} X_i \cdot X_i \right] \\
& = \frac{1}{2} \ \E \left[ m_i'' \cdot X_i^\top B^{-1} X_i \cdot \underbrace{ \E \left[ m_{1|a}' X_{1|a}^\top \right] B^{-1} X_i }_{ \eqref{eq:zhang} } \right] \\
& = \frac{1}{2} \ \E \left[ m_i'' \cdot X_i^\top B^{-1} X_i \cdot \frac{ I ( A_i = a ) }{\pi_a} \right] \\ 
& = \frac{1}{2} \ \E \left[ m_i'' \cdot X_i^\top B^{-1} X_i \middle\vert A_i = a \right] \\
    & = \frac{1}{2} \ \E \left[ m_{i|a}'' X_{i|a}^\top B^{-1} X_{i|a} \right].
\end{align*}
\end{proof}

\begin{lemma}
\label{lem:analysis}
Let $\hat{B} \coloneqq n^{-1} \SUM{i=1}{n} \hwm_i' \cdot X_i X_i^\top$ and $\tilde{B} \coloneqq n^{-1} \SUM{i=1}{n} \tilde{m}_i' \cdot X_i X_i^\top$. We have that
\begin{equation}
\label{eq:proof-beta-bread-est-inv}
\begin{split}
\norm{ \hat{B}^{-1} - B^{-1} }_{op} = O_{\sp} ( n^{-1/2} ), \\
\norm{ \tB^{-1} - B^{-1} }_{op} = O_\sp ( n^{-1/2} ).
\end{split}
\end{equation}
Furthermore, the following hold:
\begin{align}
& \hat{\bpsi}_{i}^{\bbeta} - B^{-1} X_{i} (Y_{i} - m_{i}) = O_{\sp} ( n^{-1 /2} ), \label{eq:proof-gob-mle-suppl-1} \\
& \SUM{i=1}{n} \hat{h}_{ii} \hat{\bpsi}_{i}^{\bbeta} - B^{-1} H_{1} = O_{\sp} ( n^{-1/2} ), \label{eq:proof-gob-mle-suppl-2} \\
& \frac{1}{2} \SUM{i=1}{n} \tilde{h}_{ii} \cdot \frac{\tilde{m}_{i}''}{\tilde{m}_{i}'} \cdot \tB^{-1} X_i - B^{-1} H_3 = O_{\sp} (n^{-1/2}). \label{eq:proof-gob-mle-suppl-3}
\end{align}
\end{lemma}

\begin{proof}
Since $\hwm_i = m_i + O_{\sp} ( n^{-1/2} )$ and $p$ is fixed, similar to \eqref{eq:append-proof-aux-inv}, we have that 
\begin{align*}
\norm{ \hat{B}^{-1} - B^{-1} }_{op} = \left\Vert \left( \frac{1}{n} \SUM{i=1}{n} \hwm_i' \cdot X_i X_i^\top \right)^{-1} - \left( \E [m_i' \cdot X_i X_i^\top] \right)^{-1} \right\Vert_{op} = O_{\sp} ( n^{-1/2} ).
\end{align*}

Next, we have 
\begin{align*}
& \ \hat{\bpsi}_i^{\bbeta} - B^{-1} X_{i} (Y_i - m_i) \\
= & \ \hat{B}^{-1} X_i (Y_i - \hat{m}_i) - B^{-1} X_{i} (Y_i - m_i) \\
= & \ \hat{B}^{-1} X_i (m_i - \hat{m}_i) + (\hat{B}^{-1} - B^{-1}) X_i (Y_i - m_i) \\
= & \ O_{\sp} (n^{-1/2}),
\end{align*}
where the last line follows from triangle inequality and \eqref{eq:proof-beta-bread-est-inv}.

Then by the same argument,
\begin{align*}
& \ \SUM{i=1}{n} \hat{h}_{ii} \hat{\bpsi}^{\bbeta} - B^{-1} H_1 \\
= & \ \frac{1}{n} \SUM{i=1}{n} \hat{m}_i' \cdot X_i^\top \hat{B}^{-1} X_i \cdot \hat{B}^{-1} X_i (Y_i - \hat{m}_i) - \E \left[ m' \cdot X^\top B^{-1} X_i \cdot B^{-1} X_i ( Y_i - m_i ) \right] \\
= & \ O_{\sp} (n^{-1/2}).
\end{align*}

Finally, analogous to \eqref{eq:proof-beta-bread-est-inv}, we have $\Vert \tB^{-1} - B^{-1} \Vert_{op} = O_\sp (n^{-1/2})$ and thus
\begin{align*}
& \ \frac{1}{2} \SUM{i=1}{n} \tilde{h}_{ii} \cdot \frac{\twm_i''}{\twm_i'} \cdot \tB^{-1} X_i - B^{-1} H_3 \\
= & \ \frac{1}{2} \left\{ \frac{1}{n} \SUM{i=1}{n} \twm_i'' \cdot X_i^\top \tB^{-1} X_i \cdot \tB^{-1} X_i - \E \left[ m_i'' \cdot X_i^{\top} B^{-1} X_i \cdot B^{-1} X_i \right] \right\} \\
= & \ O_{\sp} ( n^{-1/2} ).
\end{align*}
\end{proof}

\clearpage

\section{R Code Demo}  \label{append:code}

\subsection{Functions for variance estimation and statistical inference}

\begin{lstlisting}

# for gcomp_func.R file
# functions for variance estimation and statistical inference 
# (differences/ratios) of g-computation and debiased gOB estimators

vcov_mu_sandwich <- function(mu, predict, y, x, beta, inv.bread = NULL) {
  
  # Variance estimation for the vector of treatment-specific means (empirical IF)
  # mu: est per arm (k arms)
  # predict: predicted values of each subject for each arm
  # y: observed outcomes
  # x: design matrix
  # beta: est of nuisance paramters
  # inv.bread: inverse of bread matrix
  
  mu_deriv <- sapply(0:1, function(a) {   # derivates used in the variance estimation
    xa <- x; 
    xa[, 2] <- a;  
    colMeans(predict[, a+1] * (1 - predict[, a+1]) * xa)})
  y_fitted <- c(plogis(x %*% beta))
  
  if (is.null(inv.bread)) {
    bread <- t(x) %*% diag(y_fitted * (1 - y_fitted)) %*% x / length(y)
    inv.bread <- solve(bread) # solve() / qr.solve() : bread matrix inversion - naive or QR
  }
  
  if_beta <- (x * (y - y_fitted)) %*% inv.bread   # influence function for beta

  # influence function for mu
  if_mu <- if_beta %*% mu_deriv + predict - 
    matrix(mu, nrow = length(y), ncol = 2, byrow = TRUE)  
  vcov_mu <- var(if_mu) / length(y)              # variance 
  
  return(list(vcov = vcov_mu, ifunc = if_mu, var.wm = var(if_beta) / length(y), 
              ifunc.wm = if_beta, y.hat = y_fitted, 
              deriv.mu = mu_deriv, bread.inv = inv.bread))
}

vcov_mu_eif <- function(mu, predict, y, arm, adjust = 0) { 
  
  # Variance estimation for the vector of treatment-specific means (theoretical IF)
  # mu: est per arm (k arms)
  # predict: predicted values of each subject for each arm
  # y: observed outcomes
  # arm: treatment groups (starting from 0)
  # adjust = hatvalues for small-sample bias adjustment (0 for no adjustment)
  
  n <- length(y)
  nA <- table(arm)
  eif <- sapply(0:1, function(a) 
    (1 + adjust) * ifelse(arm == a, 1, 0) / (nA[a+1]/length(y)) * 
      (y - predict[, a+1]) + predict[, a+1] - mu[a+1]
  )
  
  vcov_mu <- var(eif) / length(y)
  
  return(list(vcov = vcov_mu, ifunc = eif))
}

test_diff <- function(mu, v, n, null = 0, level = 0.95, upper = TRUE) {
  
  # Statistical inference for the diference of two means
  # mu: gcomp for control and tested arms
  # v: variance of estimators for treatment-specific means 
  # n: sample sizes
  # null: null value
  # level: confidence level
  # upper: direction of alternative value
  
  delta <- ifelse(upper, 1, -1) * (mu[2] - mu[1])
  v_delta <- v[1, 1] - 2 * v[1, 2] + v[2, 2]              
  
  # Wald test
  
  z_wald <- (delta - null) / sqrt(v_delta)
  p_wald <- 1 - pnorm(z_wald)
  ci_wald <- delta + sqrt(v_delta) * qnorm(c((1-level)/2, (1+level)/2)) 
  
  # score test
  
  z_score <- delta / sqrt(v_delta + (delta - null)^2 / n)
  p_score <- 1 - pnorm(z_score)
  ci_score <- delta + sqrt(v_delta * qchisq(level, df = 1) / 
                             (1 - qchisq(level, df = 1) / n)) * c(-1, 1)
  
  return(list( 
    delta = delta, var = v_delta, 
    ward = list(pval = p_wald, ci = ci_wald, z = z_wald),
    score = list(pval = p_score, ci = ci_score, z = z_score),
    null = null, level = level))
}

test_ratio <- function(mu, v, n, null = 1, level = 0.95, upper = TRUE) {
  
  # Statistical inference for the ratio of two means
  # mu: gcomp for control and tested arms
  # v: variance of estimators for treatment-specific means 
  # n: sample sizes
  # null: null value
  # level: confidence level
  # upper: direction of alternative value
  
  # Wald test
  
  log_delta <- log(mu[2]/mu[1]) - log(null)
  v_log_delta <- v[1, 1] / mu[1]^2 - 2 * v[1, 2] / (mu[1] * mu[2]) + v[2, 2] / mu[2]^2
  
  z_wald <- ifelse(upper, 1, -1) * log_delta / sqrt(v_log_delta)
  p_wald <- 1 - pnorm(z_wald)
  ci_wald <- exp(log_delta + sqrt(v_log_delta) * qnorm(c((1-level)/2, (1+level)/2)))
  
  # score test
  
  z_score <- ifelse(upper, 1, -1) * (mu[2] - mu[1] * null) / 
    sqrt(v[1, 1] * null^2 - 2 * v[1, 2] * null + v[2, 2] + 
           (mu[2] - mu[1] * null)^2 / n)
  p_score <- 1 - pnorm(z_score)  # 1-sided
  
  a <- (1 - qchisq(level, df = 1) * (v[1, 2] / (mu[1] * mu[2]) + 1 / n)) / 
    (1 - qchisq(level, df = 1) * (v[1, 1] / mu[1]^2 + 1 / n))         # for ci
  b <- (1 - qchisq(level, df = 1) * (v[2, 2] / mu[2]^2 + 1 / n)) / 
    (1 - qchisq(level, df = 1) * (v[1, 1] / mu[1]^2 + 1 / n))         # for ci
  ci_score <- mu[2] / mu[1] * (a + sqrt(a^2 - b) * c(-1, 1)) 
  
  return(list( 
    delta = mu[2]/mu[1], varLog = v_log_delta, 
    ward = list(pval = p_wald, ci = ci_wald, z = z_wald),
    score = list(pval = p_score, ci = ci_score, z = z_score),
    null = null, level = level))
}

\end{lstlisting}

\subsection{R code demo for debiased gOB estimators with MLE}

\begin{lstlisting}

# R code demo for g-computation and debiased gOB estimators with MLE
# pooled working models

library(dplyr)
source("./gcomp_func.R")

# simulation for a hypothetical trial of N = 60 with pi1 = 25% and pi2 = 60%

n <- 60
p_max <- 10
betaW <- c(rep(sqrt(0.8 * log(5)^2/4), 4), rep(sqrt(0.2 * log(5)^2/6), 6))
betaA <- c(-1.5836, 0.5923)

set.seed(12345)

df_sim <- data.frame(Y = rep(0, n), A = rep(0, n))

for (bsvar in paste0("W", 1:p_max)) {
  df_sim[, bsvar] <- rnorm(n, 0)
}

df_sim$A <- sample(rep(0:1, n/2), size = n, replace = FALSE)

df_sim$Y <- mapply(function(a, cum) {
  rbinom(n = 1, size = 1, prob = plogis(betaA[a+1]+cum))
}, df_sim$A, tcrossprod(as.matrix(df_sim[, -c(1, 2)]), matrix(betaW, ncol = p_max)))

df_sim[, 3:6] <- df_sim[, 3:6] + 5          # covariates are not centered at zeros
df_sim[, 7:12] <- abs(df_sim[, 7:12]) + 5   # ensure all working models are wrong

# unadjsted analysis 

nA <- table(df_sim$A)
p_unadj <- tapply(df_sim$Y, df_sim$A, mean)
unadj<- p_unadj[2] - p_unadj[1]  
se_unadj <- sqrt(p_unadj[2] * (1 - p_unadj[2]) / nA[2] + 
                   p_unadj[1] * (1 - p_unadj[1]) / nA[1])  # se

# g-computation (adjust for W1-W4)

wm <- glm(Y ~ ., data = df_sim[, 1:(2+4)], family = binomial, x = TRUE,
          control = list(epsilon = 1e-06, maxit = 200))

predict_gc <- sapply(0:1, function(a) {  
  dmat_a <- wm$x; 
  dmat_a[, 2] <- a; 
  plogis(dmat_a %*% coef(wm))})

mu_gc <- colMeans(predict_gc) 
vhat <- hatvalues(wm)
vcov_gc <- vcov_mu_sandwich(mu_gc, predict_gc, wm$y, wm$x, coef(wm), n * vcov(wm)) # unadj var
bcv_mu <- vcov_mu_eif(mu_gc, predict_gc, wm$y, wm$x[, 2], adjust = vhat)  # adj var
# this is for the risk difference; use test_ratio() for the risk ratio
d_gc_bcv <- test_diff(mu_gc, bcv_mu$vcov, n) 

# gOB-MLE(C1)

beta_c1 <- coef(wm) + colMeans(vhat * vcov_gc$ifunc.wm)

predict_bc1 <- sapply(0:1, function(a) { 
  dmat_a <- wm$x; 
  dmat_a[, 2] <- a; 
  plogis(dmat_a %*% beta_c1)})

mu_bc1gob <- colMeans(cbind(
  ifelse(df_sim$A == 0, wm$y, predict_bc1[, 1]), 
  ifelse(df_sim$A == 1, wm$y, predict_bc1[, 2])))
bcv_bc1gob <- vcov_mu_eif(mu_bc1gob, predict_bc1, wm$y, wm$x[, 2], adjust = vhat)
# this is for the risk difference; use test_ratio() for the risk ratio
d_bc1gob_bcv <- test_diff(mu_bc1gob, bcv_bc1gob$vcov, n)

# gOB-MLE(C2)

beta_c2 <- sapply(1:n, function(ii) {
  coef(wm) - vcov_gc$ifunc.wm[ii, ] / n + colMeans(vhat * vcov_gc$ifunc.wm)
})

predict_bc2 <- sapply(0:1, function(a) { 
  dmat_a <- wm$x;
  dmat_a[, 2] <- a;
  plogis(diag(dmat_a %*% beta_c2))})

mu_bc2gob <- colMeans(cbind(
  ifelse(df_sim$A == 0, wm$y, predict_bc2[, 1]), 
  ifelse(df_sim$A == 1, wm$y, predict_bc2[, 2])))
bcv_bc2gob <- vcov_mu_eif(mu_bc2gob, predict_bc2, wm$y, wm$x[, 2], adjust = vhat)
# this is for the risk difference; use test_ratio() for the risk ratio
d_bc2gob_bcv <- test_diff(mu_bc2gob, bcv_bc2gob$vcov, n)

# output 

c(unadj, d_gc_bcv$delta, d_bc1gob_bcv$delta, d_bc2gob_bcv$delta)
c(se_unadj^2, d_gc_bcv$var, d_bc1gob_bcv$var, d_bc2gob_bcv$var)

\end{lstlisting}

\subsection{R code demo for debiased gOB estimators with FC}

\begin{lstlisting}

# R code demo for debiased gOB estimators with FC
# pooled working models

library(dplyr)
library(brglm2)
source("./gcomp_func.R")

# simulation for a hypothetical trial of N = 60 with pi1 = 25% and pi2 = 60%

n <- 60
p_max <- 10
betaW <- c(rep(sqrt(0.8 * log(5)^2/4), 4), rep(sqrt(0.2 * log(5)^2/6), 6))
betaA <- c(-1.5836, 0.5923)

set.seed(12345)

df_sim <- data.frame(Y = rep(0, n), A = rep(0, n))

for (bsvar in paste0("W", 1:p_max)) {
  df_sim[, bsvar] <- rnorm(n, 0)
}

df_sim$A <- sample(rep(0:1, n/2), size = n, replace = FALSE)

df_sim$Y <- mapply(function(a, cum) {
  rbinom(n = 1, size = 1, prob = plogis(betaA[a+1]+cum))
}, df_sim$A, tcrossprod(as.matrix(df_sim[, -c(1, 2)]), matrix(betaW, ncol = p_max)))

df_sim[, 3:6] <- df_sim[, 3:6] + 5          # covariates are not centered at zeros
df_sim[, 7:12] <- abs(df_sim[, 7:12]) + 5   # ensure all working models are wrong

# unadjsted analysis 

nA <- table(df_sim$A)
p_unadj <- tapply(df_sim$Y, df_sim$A, mean)
unadj<- p_unadj[2] - p_unadj[1]  
se_unadj <- sqrt(p_unadj[2] * (1 - p_unadj[2]) / nA[2] + 
                   p_unadj[1] * (1 - p_unadj[1]) / nA[1])  # se

# gOB-FC(C0) (adjust for W1-W4)

wm <- glm(Y ~ ., data = df_sim[, 1:(2+4)], family = binomial, x = TRUE,
          method = "brglmFit", type = "MPL_Jeffreys", 
          control = list(epsilon = 1e-06, maxit = 1200))

predict_gc <- sapply(0:1, function(a) {  
  dmat_a <- wm$x; 
  dmat_a[, 2] <- a; 
  plogis(dmat_a %*% coef(wm))})
mu_gc <- colMeans(predict_gc)
vcov_gc <- vcov_mu_sandwich(mu_gc, predict_gc, wm$y, wm$x, coef(wm), n * vcov(wm)) # unadj var
vhat <- hatvalues(wm)

beta_c <- coef(wm) +
  colMeans(vhat * (- (0.5 - vcov_gc$y.hat) * wm$x %*% (n * vcov(wm))))

predict_bc <- sapply(0:1, function(a) {   # individual prediction for each arm: n * 2
  dmat_a <- wm$x;
  dmat_a[, 2] <- a;
  plogis(dmat_a %*% beta_c)})

mu_bcgob <- colMeans(cbind(
  ifelse(df_sim$A == 0, wm$y, predict_bc[, 1]), 
  ifelse(df_sim$A == 1, wm$y, predict_bc[, 2])))

bcv_bcgob <- vcov_mu_eif(mu_bcgob, predict_bc, wm$y, wm$x[, 2], adjust = vhat)
# this is for the risk difference; use test_ratio() for the risk ratio
d_bcgob_bcv <- test_diff(mu_bcgob, bcv_bcgob$vcov, n)

# gOB-FC(C1)

beta_c1 <- coef(wm) + 
  colMeans(vhat * (vcov_gc$ifunc.wm - (0.5 - vcov_gc$y.hat) * wm$x %*% (n * vcov(wm)))) 

predict_bc1 <- sapply(0:1, function(a) { 
  dmat_a <- wm$x; 
  dmat_a[, 2] <- a; 
  plogis(dmat_a %*% beta_c1)})

mu_bc1gob <- colMeans(cbind(
  ifelse(df_sim$A == 0, wm$y, predict_bc1[, 1]), 
  ifelse(df_sim$A == 1, wm$y, predict_bc1[, 2])))
bcv_bc1gob <- vcov_mu_eif(mu_bc1gob, predict_bc1, wm$y, wm$x[, 2], adjust = vhat)
# this is for the risk difference; use test_ratio() for the risk ratio
d_bc1gob_bcv <- test_diff(mu_bc1gob, bcv_bc1gob$vcov, n)

# gOB-FC(C2)

beta_c2 <- sapply(1:n, function(ii) {
  coef(wm) - vcov_gc$ifunc.wm[ii, ] / n + 
    colMeans(vhat * (vcov_gc$ifunc.wm - (0.5 - vcov_gc$y.hat) * wm$x %*% (n * vcov(wm)))) 
})

predict_bc2 <- sapply(0:1, function(a) { 
  dmat_a <- wm$x;
  dmat_a[, 2] <- a;
  plogis(diag(dmat_a %*% beta_c2))})

mu_bc2gob <- colMeans(cbind(
  ifelse(df_sim$A == 0, wm$y, predict_bc2[, 1]), 
  ifelse(df_sim$A == 1, wm$y, predict_bc2[, 2])))
bcv_bc2gob <- vcov_mu_eif(mu_bc2gob, predict_bc2, wm$y, wm$x[, 2], adjust = vhat)
# this is for the risk difference; use test_ratio() for the risk ratio
d_bc2gob_bcv <- test_diff(mu_bc2gob, bcv_bc2gob$vcov, n)

# output 

c(unadj, d_bcgob_bcv$delta, d_bc1gob_bcv$delta, d_bc2gob_bcv$delta)
c(se_unadj^2, d_bcgob_bcv$var, d_bc1gob_bcv$var, d_bc2gob_bcv$var)

\end{lstlisting}

\subsection{R code demo for estimators with stratified working models}

\begin{lstlisting}

# R code demo for stratified working models
# small-sample bias adjustment for variance estimation

library(dplyr)
library(brglm2)
source("./gcomp_func.R")

# simulation for a hypothetical trial of N = 200 with pi1 = 10% and pi2 = 25%

n <- 200
p_max <- 35 # adjust for 5 - 35
betaW <- rep(sqrt(log(25)^2/p_max), p_max) 
betaA <- c(-4.7173, -2.4760)

set.seed(12345)

df_sim <- data.frame(A = rep(0, n), Y = rep(0, n))

for (bsvar in paste0("W", 1:p_max)) {
  df_sim[, bsvar] <- rnorm(n)
}

df_sim$A <- sample(rep(0:1, n/2), size = n, replace = FALSE)

df_sim$Y <- mapply(function(a, cum) {
  rbinom(n = 1, size = 1, prob = plogis(betaA[a+1]+cum))
}, df_sim$A, tcrossprod(as.matrix(df_sim[, -c(1, 2)]), matrix(betaW, ncol = p_max)))

df_sim[, 3:32] <- df_sim[, 3:32] + 5   
df_sim[, 33:37] <- abs(df_sim[, 33:37]) + 5# ensure all working models are wrong

# unadjsted analysis 

nA <- table(df_sim$A)
p_unadj <- tapply(df_sim$Y, df_sim$A, mean)
unadj<- p_unadj[2] - p_unadj[1]  
se_unadj <- sqrt(p_unadj[2] * (1 - p_unadj[2]) / nA[2] + 
                   p_unadj[1] * (1 - p_unadj[1]) / nA[1])  # se

# GC-MLE 

ls_wm_ml <- lapply(0:1, function(a)  # adjust for W1-W10
  glm(Y ~ ., data = df_sim[df_sim$A == a, 2:(2+10)], family = binomial, x = TRUE, 
      control = list(epsilon = 1e-06, maxit = 200)))
dmat <- as.matrix(cbind(1, df_sim[, 3:(10+2)]))  # design matrix

predict_gc <- sapply(1:2, function(a) 
  plogis(dmat %*% coef(ls_wm_ml[[a]])))

mu_gc <- colMeans(predict_gc)   

vhat_ml <- rep(0, n)
vhat_ml[df_sim$A == 0] <- hatvalues(ls_wm_ml[[1]])
vhat_ml[df_sim$A == 1] <-  hatvalues(ls_wm_ml[[2]])

bcv_mu <- vcov_mu_eif(mu_gc, predict_gc, df_sim$Y, df_sim$A, adjust = vhat_ml)
# this is for the risk difference; use test_ratio() for the risk ratio
d_gc_bcv <- test_diff(mu_gc, bcv_mu$vcov, n)  # 

# gOB-MLE(C1)

if_beta_ml <- lapply(1:2, function(a) {    # influence function for beta
  wm <- ls_wm_ml[[a]];
  (wm$x * (wm$y - fitted(wm))) %*% (nobs(wm) * vcov(wm))})

beta_c_ml <- lapply(1:2, function(a) 
  coef(ls_wm_ml[[a]]) + colMeans(vhat_ml[df_sim$A == (a-1)] * if_beta_ml[[a]]))

predict_bc_ml <- sapply(1:2, function(a)    # individual prediction for each arm: n * 2
  plogis(dmat %*% beta_c_ml[[a]]))

mu_bcgob_ml <- colMeans(cbind(
  ifelse(df_sim$A == 0, df_sim$Y, predict_bc_ml[, 1]), 
  ifelse(df_sim$A == 1, df_sim$Y, predict_bc_ml[, 2])))

bcv_bcgob_ml <- vcov_mu_eif(mu_bcgob_ml, predict_bc_ml, df_sim$Y, df_sim$A, adjust = vhat_ml)
# this is for the risk difference; use test_ratio() for the risk ratio
d_bcgob_bcv_ml <- test_diff(mu_bcgob_ml, bcv_bcgob_ml$vcov, n)

# gOB-FC(C0)

ls_wm_fc <- lapply(0:1, function(a)   # adjust for W1-W10
  glm(Y ~ ., data = df_sim[df_sim$A == a, 2:(2+10)], family = binomial, x = TRUE,
      method = "brglmFit", type = "MPL_Jeffreys", 
      control = list(epsilon = 1e-06, maxit = 5000)))

vhat_fc <- rep(0, n)
vhat_fc[df_sim$A == 0] <- hatvalues(ls_wm_fc[[1]])
vhat_fc[df_sim$A == 1] <-  hatvalues(ls_wm_fc[[2]])

beta_c0_fc <- lapply(1:2, function(a) 
  coef(ls_wm_fc[[a]]) + colMeans(
    vhat_fc[df_sim$A == (a-1)] * 
      (- (0.5 - fitted(ls_wm_fc[[a]])) * 
         ls_wm_fc[[a]]$x %*% (nobs(ls_wm_fc[[a]]) * vcov(ls_wm_fc[[a]])))))

predict_bc0_fc <- sapply(1:2, function(a)    # individual prediction for each arm: n * 2
  plogis(dmat %*% beta_c0_fc[[a]]))

mu_bc0gob_fc <- colMeans(cbind(
  ifelse(df_sim$A == 0, df_sim$Y, predict_bc0_fc[, 1]), 
  ifelse(df_sim$A == 1, df_sim$Y, predict_bc0_fc[, 2])))  

bcv_bc0gob_fc <- vcov_mu_eif(mu_bc0gob_fc, predict_bc0_fc, df_sim$Y, df_sim$A, adjust = vhat_fc)
# this is for the risk difference; use test_ratio() for the risk ratio
d_bc0gob_bcv_fc <- test_diff(mu_bc0gob_fc, bcv_bc0gob_fc$vcov, n)

# gOB-FC(C1)

if_beta_fc <- lapply(1:2, function(a) {    # influence function for beta
  wm <- ls_wm_fc[[a]];
  (wm$x * (wm$y - fitted(wm))) %*% (nobs(wm) * vcov(wm))})

beta_c_fc <- lapply(1:2, function(a) 
  coef(ls_wm_fc[[a]]) + colMeans(
    vhat_fc[df_sim$A == (a-1)] * 
      (if_beta_fc[[a]] - (0.5 - fitted(ls_wm_fc[[a]])) * 
         ls_wm_fc[[a]]$x %*% (nobs(ls_wm_fc[[a]]) * vcov(ls_wm_fc[[a]])))))

predict_bc_fc <- sapply(1:2, function(a)    # individual prediction for each arm: n * 2
  plogis(dmat %*% beta_c_fc[[a]]))

mu_bcgob_fc <- colMeans(cbind(
  ifelse(df_sim$A == 0, df_sim$Y, predict_bc_fc[, 1]), 
  ifelse(df_sim$A == 1, df_sim$Y, predict_bc_fc[, 2])))
bcv_bcgob_fc <- vcov_mu_eif(mu_bcgob_fc, predict_bc_fc, df_sim$Y, df_sim$A, adjust = vhat_fc)
# this is for the risk difference; use test_ratio() for the risk ratio
d_bcgob_bcv_fc <- test_diff(mu_bcgob_fc, bcv_bcgob_fc$vcov, n)

# output 
 
c(unadj, d_gc_bcv$delta, d_bcgob_bcv_ml$delta, d_bc0gob_bcv_fc$delta, d_bcgob_bcv_fc$delta)
c(se_unadj^2, d_gc_bcv$var, d_bcgob_bcv_ml$var, d_bc0gob_bcv_fc$var, d_bcgob_bcv_fc$var)

\end{lstlisting}

\end{document}